\documentclass[twoside, 12pt]{article}
\usepackage{pgfplots}
\usepackage{tikz}
\usepackage{subcaption}
\usepackage[a4paper,top=3.5cm,bottom=3cm,left=2.5cm,right=2cm,bindingoffset=0mm]{geometry}

\usepackage{aditya}
\tikzset{
    >=latex,
    pil/.style={
            draw,
      <-, 
      decorate,
      decoration={snake,,amplitude=.02cm, pre length=.2cm,post length=.2cm,}
              }}

\hypersetup{colorlinks=true,linkcolor=red, citecolor=blue, urlcolor=red}

\title{\bf \sc \textcolor{darkblue}{Reputational Spillovers}}

\author{
\begin{tabular}{ccc}
Aditya Kuvalekar
& \hspace{1cm} &
Anna Sanktjohanser\thanks{Kuvalekar: University of Essex. Email: \texttt{a.kuvalekar@essex.ac.uk}, Sanktjohanser: Toulouse School of Economics. Email: \texttt{anna.sanktjohanser@tse-fr.eu}. Kuvalekar thanks the Indian School of Business for their hospitality. This paper was written while he was visiting ISB. We thank Martino Banchio, Deepal Basak, Joyee Deb, Jack Fanning, Marina Halac, Johannes H\"orner, Elliot Lipnowski, Doron Ravid, and Alex Wolitzky.
}
\end{tabular}
}

\date{March 2026}
\begin{document}
\maketitle

\begin{abstract}
{\small \noindent}  \\
\vspace{-0.7cm}

We analyze a reputational bargaining game in which a central player negotiates simultaneously with two peripheral players. Each player is either rational or a commitment type who never concedes and insists on a fixed share, and concessions are publicly observed. The central player’s type is global, so actions in one dispute update beliefs in the other and generate reputational spillovers. The game admits a unique equilibrium, enabling a sharp comparison with the bilateral benchmark of \cite{abreu2000bargaining}. Spillovers are payoff-relevant if and only if a peripheral is uniquely the most reputable player initially. In that case, spillovers overturn the bilateral prediction that toughness pays: the central player is never strictly better off and can be strictly worse off; the strongest peripheral loses; and the weakest peripheral can benefit, especially when the center’s higher-stakes dispute is with the other peripheral.

\bigskip

\vspace{0.2cm}
\end{abstract}

\begin{spacing}{1.1}
\section{Introduction}

A reputation for toughness is a strategic asset in bilateral bargaining. In the canonical reputational war of attrition of \cite{abreu2000bargaining}, the mere possibility that one party is a commitment type who never concedes improves her terms: delay is consistent with commitment, so the opponent concedes earlier. However, many important bargaining relationships are not isolated. Governments negotiate multiple disputes in parallel, large firms bargain simultaneously with suppliers and unions, and sovereign borrowers negotiate with several creditor classes at once. When concessions are publicly observed across disputes, a concession in one negotiation can immediately affect beliefs---and hence behavior---in all others. This leaves little scope for selective flexibility: yielding on one front may reveal rationality on every front. We ask whether the bilateral logic---that toughness pays---survives when a single reputation must be maintained across multiple negotiations.

To study these questions, we consider a three-player environment in which a central player bargains simultaneously with two peripheral opponents. Each bilateral negotiation follows a continuous-time war-of-attrition protocol: at any time a rational player may concede, concessions are publicly observed, and the non-conceding party receives a fixed share of the surplus. Each player is either rational or a commitment type who never concedes and insists on that share. The key feature is that the center has a \emph{global type}: the same underlying commitment (or lack thereof) governs her behavior in both negotiations. Consequently, a concession by the center in one dispute is informative in the other as well, so beliefs about the center must remain consistent across negotiations. This cross-negotiation learning creates a belief-consistency constraint: the center cannot be perceived as tough in one dispute while behaving flexibly in the other.

Our main finding is that this constraint can overturn the bilateral prediction that toughness benefits the player who possesses it: the center's equilibrium payoff is never higher than if the two negotiations occurred in isolation, and can, in fact, be strictly lower.
In contrast, a peripheral opponent who is \emph{weaker} in the bilateral sense can be strictly better off, and the distributional consequences are systematic: the strongest peripheral never gains---and can strictly lose---from spillovers, while the weakest peripheral never loses and may strictly gain. The weakest peripheral benefits most precisely when the center's \emph{other} negotiation is higher-stakes.

These dynamics find natural parallels in practice. Sovereign debt
restructurings are a canonical example: a sovereign debtor faces
multiple creditor classes simultaneously, and any concession to one
class immediately reveals flexibility to the others. Our model
suggests that holdout creditors---often the weakest in terms of
individual exposure---can extract outsized terms precisely because
the sovereign's higher-stakes negotiation lies elsewhere. The
sovereign, constrained by a global reputation, cannot appear tough
selectively. A similar logic arises in labor negotiations when a
firm bargains simultaneously with multiple unions. If management
concedes to one union, this is publicly observed and immediately
undermines its bargaining position with the others---it cannot
credibly claim toughness on one front after yielding on another.
Our results predict that the smaller union benefits
disproportionately from this constraint, especially when
management's more consequential negotiation is with the larger
union.

We establish these conclusions by providing a complete characterization of equilibrium behavior and payoffs. Characterizing multilateral reputational bargaining is challenging because the public state is a vector of evolving reputations, and any concession induces cross-negotiation belief updates that reshape continuation values in all unresolved disputes. As a result, equilibrium strategies typically may feature non-stationary—and potentially history-dependent—concession hazards. Despite this, we show that the simultaneous three-player environment admits a \emph{unique} equilibrium with a sharp phase structure (Figure~\ref{Figure: schematic}). Along the path where neither dispute has yet been resolved, behavior has the following form. There may first be a time-0 adjustment, implemented through an immediate concession atom. If one peripheral is initially sufficiently reputable, this is followed by an initial phase in which that peripheral remains inactive while the center bargains with the other peripheral. Once all players are active, all players concede at the same constant rate as in the canonical bilateral benchmark until a common finite terminal time at which remaining players are believed committed with probability one. The role of the initial phase is to reconcile cross-negotiation belief consistency with the incentives created by a global reputation.

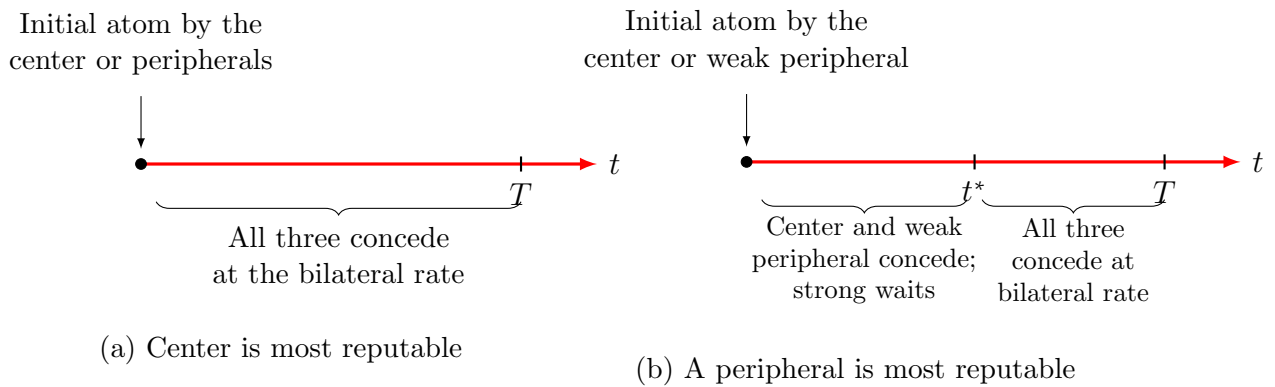
\begin{figure}[ht]
\centering
\centering\hspace{-2cm}
\begin{minipage}{0.45\textwidth}
\centering
\begin{tikzpicture}
  
  \draw[->, very thick, red] (0,0) -- (6,0) node[right, black]{$t$};
  \draw[thick] (5,0.1) -- (5,-0.1) node[below]{$T$};
  \draw[fill=black] (0,0) circle (2pt);
  \draw[decorate, decoration={brace, amplitude=5pt, mirror}]
    (0.2,-0.5) -- (5,-0.5);
  \node[align=center, font=\small] at (0, 1.6)
    {Initial atom by the\\center or peripherals};
  \draw[->, thin] (0,0.9) -- (0,0.2);
  \node[align=center, font=\small] at (2.6,-1.2)
    {All three concede\\at the bilateral rate};
\end{tikzpicture}
\vspace{-0.2cm}
\subcaption{Center is most reputable}\label{Figure: schematic center}
\end{minipage}
\begin{minipage}{0.45\textwidth}
\centering
\begin{tikzpicture}
  \draw[->, very thick, red] (0,0) -- (6.5,0) node[right, black]{$t$};
  \draw[->, very thick, white] (0,2.13) -- (6.5,2.13) node[right, white]{$t$};
  \draw[thick] (3,0.1) -- (3,-0.1) node[below]{$t^*$};
  \draw[thick] (5.5,0.1) -- (5.5,-0.1) node[below]{$T$};
  \draw[fill=black] (0,0) circle (2pt);
  \draw[decorate, decoration={brace, amplitude=5pt, mirror}]
    (0.2,-0.5) -- (2.9,-0.5);
  \draw[decorate, decoration={brace, amplitude=5pt, mirror}]
    (3.1,-0.5) -- (5.5,-0.5);
  \node[align=center, font=\small] at (0, 1.6)
    {Initial atom by the\\center or weak peripheral};
  \draw[->, thin] (0,0.9) -- (0,0.2);
  \node[align=center, font=\footnotesize] at (1.55,-1.3)
    {Center and weak\\peripheral concede;\\strong waits};
  \node[align=center, font=\footnotesize] at (4.3,-1.3)
    {All three\\concede at\\bilateral rate};
\end{tikzpicture}
\vspace{-0.2cm}
\subcaption{A peripheral is most reputable}\label{Figure: schematic peripheral}
\end{minipage}
\caption{Schematic equilibrium structure.}\label{Figure: schematic}
\end{figure}

The equilibrium structure delivers a simple and transparent condition for when spillovers are payoff relevant relative to the bilateral benchmark due to \cite{abreu2000bargaining} (henceforth AG).\footnote{As a benchmark, we compare two separate bilateral AG interactions. In the bilateral AG benchmark, the weaker player (one who is more likely to be a commitment type) concedes with an atom to the stronger player up front, and then the two players start conceding at a constant hazard rate. } Spillovers overturn the payoff rankings of AG
\emph{precisely when the player with the highest initial reputation is a peripheral}. At the beginning of the game, if the center is at least as reputable as both peripherals, the public linkage of negotiations does not distort incentives: equilibrium outcomes coincide with those from two separate bilateral reputational wars of attrition (Figure~\ref{Figure: schematic center}). By contrast, when a peripheral begins with the strongest reputation (Figure~\ref{Figure: schematic peripheral}), that strong peripheral can profitably wait at the outset. This changes the center's incentives because any concession by the center is effectively \emph{global}: it immediately resolves both disputes and reveals that the center is the flexible (rational) type. In this region, sustaining gradual concession requires a nontrivial reallocation of who concedes early and how quickly, and this reallocation drives the starkly different payoff implications relative to the AG benchmark.

The economic mechanism is intuitive. When the strongest peripheral waits initially, the center and the weaker peripheral are the only parties ``bargaining in earnest'' at the beginning. But the center's tradeoff in this early phase differs from the bilateral benchmark: conceding is more costly because it settles not one dispute but both, sacrificing the option value of holding out in the other negotiation. To keep the center willing to delay rather than concede immediately, the weaker peripheral must concede \emph{faster} early on than he would in an isolated bilateral negotiation. Once all parties concede at positive rates, however, the logic of reputational bargaining pins down local incentives in a way that forces posteriors to align and concession rates to settle down to the familiar bilateral rate. The equilibrium therefore cannot absorb the early asymmetry in concession behavior smoothly over time. Instead, it reconciles early ``catch-up'' dynamics through \emph{time-zero behavior}: relative to independent bilateral bargaining, immediate concessions become less concentrated on the weakest peripheral and can shift onto the center, even in the center's negotiation with the weaker opponent. This is the precise channel through which spillovers can reverse the standard bilateral intuition.

These forces generate sharp payoff implications. When the center is initially weakest, spillovers do not change the center's aggregate payoff relative to two separate bilateral negotiations, but they redistribute surplus across the peripherals: the strongest peripheral loses and the weaker peripheral gains. When the center's initial reputation lies between those of the two peripherals, spillovers are most consequential: the center's equilibrium payoff falls below the sum of her bilateral benchmark payoffs, the strongest peripheral is strictly worse off, and the weakest peripheral benefits. Moreover, the magnitude of the center's loss and the weakest peripheral's gain increases with the stakes of the center's negotiation with the stronger peripheral. Put differently, a peripheral can benefit from the center ``fighting a bigger battle elsewhere,'' because the center's global reputation makes any concession especially costly when another high-stakes dispute is unresolved.

Furthermore, Proposition~\ref{prop:vanishing} shows that the payoff reversal is not driven by large commitment probabilities. Even as all three commitment probabilities vanish, the center remains strictly worse off than in the bilateral benchmark, with a payoff gap that converges to a positive limit when relative reputations are held fixed.

We also show that the disadvantage of a global reputation is not an artifact of the baseline timing, information assumptions or the number of peripherals. We study three extensions. First, we consider a sequential environment in which the center bargains with one peripheral first and only subsequently bargains with the other. In this case, the prospect of the future negotiation enters the center's incentives in the first stage and can induce asymmetric concession behavior even when initial reputations are symmetric, again potentially requiring an immediate concession by the center. Second, we relax full observability of concessions and assume that the uninvolved peripheral observes only that an agreement has been reached elsewhere (and when), but not who conceded. Then an agreement itself is informative about the center's flexibility, creating a discrete belief update that can trigger an immediate concession by the center in the remaining negotiation.

Finally, we numerically study a four-player star—a center bargaining simultaneously with three peripherals. Our simulations suggest that both the phase structure and the qualitative insight that toughness can be a liability for the central player carry over. Together, these extensions underscore that the key friction is the combination of a global type and cross-negotiation learning.

In summary, in contrast to the bilateral benchmark, where toughness pays, we show that maintaining a single reputation across multiple negotiations can reverse that logic, systematically shifting surplus away from the center and toward weaker opponents.

\subsection{Related Literature}\label{Section: literature}

Our paper contributes to the literature on reputational bargaining by identifying a clean, empirically relevant environment in which reputational considerations can \emph{harm} the player whose reputation is global. Much of the literature on reputational bargaining derives predictions for the limiting case where the probability of facing a committed opponent is small. In contrast, we treat these priors as fixed primitives and characterize equilibrium behavior and payoffs at these interior priors. This focus is motivated by evidence that delay and disagreements are empirically salient in bargaining environments. Using data from eBay’s Best Offer platform, \cite{backus2020sequential} show that only about one‑third of bargaining threads end in immediate agreement and that many negotiations end in disagreement, including after non‑trivial delay. Related evidence comes from laboratory implementations of the AG protocol: \cite{embrey2015bargaining} implement the two‑stage demand‑then‑concession environment and find that subjects respond to the possibility of obstinate types, yet bargaining features more aggressive demands and longer conflicts than the benchmark, consistent with non‑negligible ``effective'' obstinacy. These patterns motivate treating commitment probabilities as fixed primitives (rather than vanishingly small), so that the spillover effects we characterize are first‑order; at the same time, we also study the vanishing-prior cases in Section~\ref{sec:vanishing} to show that reputational spillovers persist qualitatively even in the limit.

Our paper builds on the classic reputational bargaining models (\cite{abreu2000bargaining} and follow-up papers), where toughness pays in bilateral negotiations: players who may be ``hard'' types obtain better terms.\footnote{For follow-up papers, see for instance \cite{AbreuPearce2007,AbreuPearceStacchetti2015,kambe1999bargaining,AbreuSethi2003,Fanning2016,Fanning2018,Fanning2021,Wolitzky2012,Sanktjohanser2022}
 among others. See \cite{fanning2022reputational} for a survey.} We depart from this benchmark by examining a multilateral environment. This connects to the network bargaining literature \citep{manea2011bargaining,abreu2012bargaining}, which highlights how outcomes depend on network position, and to \cite{compte2002role}, who show that toughness need not always pay once players face multiple opponents with outside options (see also \cite{atakan2014bargaining}).\footnote{Reversals of the ``toughness pays'' prediction arise through different channels in other contexts, e.g., public information about commitment \citep{basak2024social} and collective bargaining with endogenous coalition structure \citep{ma2023efficiency}. Our mechanism is distinct: it works through belief spillovers across simultaneously observed negotiations.} As emphasized by \cite{fanning2022reputational}, a central difficulty in multilateral reputational bargaining is that continuation values depend on multiple reputations, so equilibrium concession rates typically evolve over time. We show that in a parsimonious multilateral environment---a central player bargaining simultaneously with two opponents---this evolution can nevertheless be characterized sharply, and it yields systematic payoff reversals relative to independent bilateral bargaining.

Our work is also related to the literature on multilateral wars of attrition. \cite{eraslan2023multilateral} study a three-player war of attrition with majority rule and show that the central player can be worse off than in bilateral bargaining; in our setting, by contrast, inefficiency arises from reputational spillovers across simultaneous negotiations. \cite{ozyurt2015bargaining} analyzes sequential bargaining between a buyer and two sellers, where frictions emerge from the outside option; we instead consider simultaneous bargaining, in which a single reputation must be maintained across negotiations. \cite{kambe2019n} examines multilateral wars of attrition with incomplete information, where equilibrium exit patterns depend on payoff asymmetries; our contribution is to show that reputational spillovers across negotiations provide a distinct channel through which central players may be disadvantaged.

The remainder of the paper is organized as follows. Section~\ref{sec:model} presents the model. Section~\ref{Section: benchmark} recalls the bilateral benchmark with independent negotiations. Section~\ref{Section: analysis and main result} characterizes the unique equilibrium and derives the payoff implications and comparative statics. Section~\ref{sec:vanishing} studies the limit of our model as the probability of commitment types goes to zero. And finally, Section~\ref{Section: Extensions} studies extensions that relax simultaneity and full observability. Proofs are collected in the Appendix.

\section{Model}\label{sec:model}

\paragraph{Players and negotiations.}
There are three players, $A$, $B$, and $C$.\footnote{For expositional clarity, we refer to $C$ as ``she'' and to $A$ and $B$
as ``he'' throughout.}
Player $C$ is simultaneously engaged in two bilateral negotiations: one with $A$ and one with $B$.
Players $A$ and $B$ do not negotiate with each other.
For $i\in\{A,B\}$, the negotiation between $i$ and $C$ yields a surplus $\pi_{iC}>0$ that is realized once either party
concedes in that negotiation. We normalize $\pi_{BC}=1$.

\paragraph{Time and actions.}
Time is continuous, $t\in[0,\infty)$, and all players discount at a common rate $\discount>0$.
At any time $t$, each rational player may either wait or concede. A concession is irreversible and publicly observed. 
If, in the negotiation between $i\in\{A,B\}$ and $C$, one party concedes to the other at time $t$, then the party receiving
the concession obtains $\alpha\pi_{iC}$ and the conceding party obtains $(1-\alpha)\pi_{iC}$, with $\alpha>1/2$, discounted by $e^{-\discount t}$.

Player $C$ may, if rational, concede to $A$, to $B$, or to both; a concession to $i$ settles the negotiation between $i$ and $C$ immediately.
(As shown below, in equilibrium any positive-time concession by $C$ is simultaneous across negotiations.)

\paragraph{Types and information.}
At time $0$, Nature draws each player's type. Player $i\in\{A,B,C\}$ is behavioral with probability $z_i(0)\in(0,1)$ and rational
with probability $1-z_i(0)$. Types are independent across players, and each player privately observes her own type.
The prior $(z_A(0),z_B(0),z_C(0))$ is common knowledge.

A behavioral type never concedes. A rational type may concede at any time.
Importantly, player $C$ has a \emph{single global type} that governs her behavior in both negotiations:
if $C$ is behavioral she never concedes in either negotiation; if $C$ is rational she may concede in either negotiation.

\paragraph{Public histories and filtration.}
A \emph{concession event} is a triple $(i,j,s)$ indicating that player $i$ conceded to counterparty $j$ at calendar time $s$.
A \emph{public history} at time $t$, denoted $h^t$, consists of $t$ together with the (finite) list of all concession
events that occurred strictly before $t$. Public histories $h^t$ record all concession events that occur strictly before calendar time $t$. If one or more concessions occur at calendar time $t$, they are publicly observed immediately
and may trigger further concessions with zero delay. We denote by $h^{t^+}$ the public history immediately
after all such time-$t$ concessions have been realized. Thus $t^+$ refers to a post-event history
node at the \emph{same} calendar time $t$ (no time elapses). Let $\mathcal H$ denote the set of public histories. Let
$(\mathcal F_t)_{t\ge 0}$ be the filtration generated by public histories. 

\paragraph{Strategies as stopping times.}
A rational player can take only one irreversible action (concede). Accordingly, a strategy for a rational player
is equivalently specified by a concession time, i.e.\ an $(\mathcal F_t)$-stopping time with values in $[0,\infty]$.

For $i\in\{A,B\}$, a strategy specifies an $(\mathcal F_t)$-stopping time (possibly random) $\tau_i\in[0,\infty]$ at which $i$ concedes to $C$.
For $C$, we \emph{do not impose simultaneity a priori}: a strategy specifies a pair of stopping times
$(\tau_C^A,\tau_C^B)\in[0,\infty]^2$, where $\tau_C^A$ (resp.\ $\tau_C^B$) is the time at which $C$ concedes to $A$ (resp.\ $B$).
Behavioral types never concede, i.e.\ $\tau_i=\infty$ for $i\in\{A,B\}$ and $\tau_C^A=\tau_C^B=\infty$.

\paragraph{Payoffs and expected utilities.}
Payoffs are additive across negotiations.
Fix a strategy profile $\sigma$ and a belief system $z$. Under $(\sigma,z)$, let
$\tau_A,\tau_B\in[0,\infty]$ be the (possibly random) times at which $A$ and $B$ concede to $C$, and let
$(\tau_C^A,\tau_C^B)\in[0,\infty]^2$ be the (possibly random) times at which $C$ concedes to $A$ and to $B$, respectively,
where $\infty$ means ``never concedes.''

For each $i\in\{A,B\}$, define the resolution time of the negotiation between $i$ and $C$ by
\[
T_{iC}:=\tau_i\wedge \tau_C^i.
\]
If $T_{iC}=\infty$, the negotiation between $i$ and $C$ is never resolved and yields payoff $0$ to both parties. If both parties concede simultaneously, we assume they split the surplus equally, so each receives $\tfrac12\pi_{iC}$ (discounted). If $T_{iC}<\infty$, then the realized (discounted) payoff to player $i$ from this negotiation is
\[
U_i^{\,iC}
:= e^{-\discount T_{iC}}\pi_{iC}\Big(
(1-\alpha)\mathbf 1\{\tau_i<\tau_C^i\}
+\alpha\mathbf 1\{\tau_C^i<\tau_i\}
+\tfrac12\mathbf 1\{\tau_i=\tau_C^i<\infty\}
\Big),
\]
and the realized payoff to $C$ from the same negotiation is
\[
U_C^{\,iC}
:= e^{-\discount T_{iC}}\pi_{iC}\Big(
\alpha\mathbf 1\{\tau_i<\tau_C^i\}
+(1-\alpha)\mathbf 1\{\tau_C^i<\tau_i\}
+\tfrac12\mathbf 1\{\tau_i=\tau_C^i<\infty\}
\Big).
\]
Total realized payoffs are
\[
U_A:=U_A^{\,AC},\qquad U_B:=U_B^{\,BC},\qquad U_C:=U_C^{\,AC}+U_C^{\,BC}.
\]

At any public history $h^t$, a rational type of player $k\in\{A,B,C\}$ evaluates strategies by expected utility:
their continuation value is the conditional expectation
\[
V_k(h^t):=E_{\sigma,z}\!\big[\,U_k \,\big|\, h^t\big].
\]

\paragraph{Equilibrium.}
The solution concept is weak Perfect Bayesian equilibrium (PBE). A weak PBE is a pair $(\sigma, z)$ consisting of a strategy
profile $\sigma$ and a belief system $z=(z_A,z_B,z_C)$ that assigns to each public history $h^t$ a posterior distribution over types such that: (i) the strategy maximizes a player's expected utility given beliefs, i.e., $\sigma$ is sequentially rational
given $z$, and (ii) $z$ is derived from $\sigma$ by Bayes' rule at every history reached with positive probability under
$\sigma$; off-path beliefs are unrestricted. Throughout, ``equilibrium'' refers to weak PBE.

\paragraph{Continuation games.}
Fix a public history $h^t$ at which exactly one negotiation remains contested, between
player $i\in\{A,B\}$ and player $C$. Let $z_i(h^t)$ and $z_C(h^t)$ denote the common posterior
probabilities at $h^t$ that $i$ and $C$ are behavioral types. Re-index time by
\emph{continuation time} $\tau:=s-t\ge 0$. Then the continuation subgame from $h^t$ is
strategically equivalent (up to this change of time origin) to the bilateral AG
reputational war of attrition between $i$ and $C$ over surplus $\pi_{iC}$, with initial
reputations $(z_i(h^t),z_C(h^t))$. Since the bilateral AG game has a unique equilibrium,
play in any PBE after $h^t$ must coincide with this unique AG equilibrium. In particular,
the AG equilibrium may feature an atom at continuation time $\tau=0$, corresponding to an
\emph{immediate} concession at $t+$.

To record the relevant objects, consider a generic bilateral AG game between players $i$ and
$j$ over surplus $\pi>0$, with initial reputations $(z_i,z_j)\in[0,1)^2$.\footnote{We interpret
expressions below by continuity at boundary cases; in particular $g_i^j(z_i,0)=1$ for $z_i>0$
and $g_i^j(0,z_j)=0$.}
Let
\[
g_i^{\,j}(z_i,z_j):=\max\left\{1-\frac{z_j}{z_i},\,0\right\}
\]
denote the equilibrium probability of an \emph{immediate concession by $j$ to $i$} at
continuation time $\tau=0$ (i.e.\ the size of $j$'s atom at $\tau=0$). We sometimes just write $g^j_i(t)$ for short, when $z_j$ and $z_i$ are understood. Let
$V^{AG}_{ij}(\pi;z_i,z_j)$ denote the equilibrium payoff to player $i$'s \emph{rational type} at
continuation time $\tau=0$ in this bilateral AG game. Then
\[
V^{AG}_{ij}(\pi;z_i,z_j)
=\pi\Big((1-\alpha)+(2\alpha-1)\,g_i^{\,j}(z_i,z_j)\Big).
\]
When $\pi$ is clear from context we write $V^{AG}_{ij}(z_i,z_j)$ for short. Finally, define
the \emph{time-0 AG benchmark payoff} to player $i$ against $j$ by
\[
v^{AG}_{ij}:=V^{AG}_{ij}(\pi_{ij};z_i(0),z_j(0)).
\]

Since continuation play after any concession is uniquely determined, equilibrium behavior is fully characterized by strategies in the no-concession subgame. We therefore summarize strategies in the initial phase -- when both negotiations are still contested -- by concession‑time distributions.

\paragraph{Strategies in the no-concession subgame.}
We focus on public histories at which both negotiations are still contested (i.e., no
concession has occurred yet). In the no-concession subgame, the public history is summarized by calendar time $t$ alone. We represent each player's
equilibrium behavior by a cumulative distribution function (cdf) of the concession time conditional on no concession before
that time. For $i\in\{A,B\}$, let $\cdf_i:[0,\infty)\to[0,1]$ denote the probability that $i$ has conceded to $C$ by time $t$,
conditional on no prior concession. For $C$, let $\cdf_C^A(t)$ (resp.\ $\cdf_C^B(t)$) denote the probability that $C$ has
conceded to $A$ (resp.\ $B$) by time $t$, conditional on no prior concession. These cdfs are nondecreasing and right-continuous,
with $\cdf_i(0-)=0$.\footnote{Throughout, we reserve the notation $(\cdot-)$
for left limits of cdfs, e.g.\ $\cdf(t-)=\lim_{s\uparrow t}F(s)$. We do not use $(\cdot+)$ to denote
right limits; all cdfs are right-continuous and we write $F(0)$ for the time-0 atom.}

Here $F_i(t)$ is unconditional over types (but conditional on no earlier concession in the public history).
Since behavioral types never concede, $F_i(t)\le 1-z_i(0)$ for all $t$.\footnote{In the equilibrium characterized below, all rational concession mass is exhausted by a
finite time $T$, so $F_i(T)=F_i(\infty)=1-z_i(0)$.}

\blemma\label{Lemma: C concedes to both once she concedes to one}
In any PBE, at any public history in which both negotiations are still contested, if $C$ concedes to one peripheral at
some calendar time $t\geq0$, then $C$ concedes to the other peripheral immediately as well (i.e.\ at calendar time $t+$).
Consequently, in the no-concession subgame it is without loss to restrict attention to strategies in which
$\cdf_C^A(t)=\cdf_C^B(t)$ for all $t\geq0$.
\elemma

\bprf
Fix a PBE $(\sigma,z)$. Consider any public history at which both negotiations are still contested and at calendar time $t\geq0$
player $C$ concedes to $A$. Since behavioral types never concede, this action reveals that $C$ is rational; hence at the
post-concession history (calendar time $t$) the posterior belief assigns probability one to $C$ being rational, i.e.\ $z_C(t+)=0$.

From time $t+$ onward, the only remaining contested negotiation is between $B$ and $C$, and the continuation game between $B$ and $C$ is exactly the
bilateral AG war of attrition with initial reputations $(z_B(t+),z_C(t+))=(z_B(t+),0)$. Because $z_B(t+)\in(0,1)$, player $C$
has strictly lower reputation than $B$. In the unique AG equilibrium, the lower-reputation player concedes at continuation time
$0$ with probability $1-\frac{z_C(t+)}{z_B(t+)}=1$. Therefore, sequential rationality in the continuation subgame implies that
$C$ concedes to $B$ immediately at calendar time $t+$.\footnote{Hence, if $C$ concedes with positive probability to $A$ at $t=0$, $B$ will not concede with positive probability to $C$ at $t=0$.}

The same argument applies if $C$ concedes to $B$ first. Hence, whenever $C$ concedes at a positive calendar time while both
negotiations are still contested, she concedes on both negotiations immediately. This justifies representing $C$'s no-concession strategy by
a single cdf $\cdf_C$ for simultaneous concessions (setting $\cdf_C(t):=\cdf_C^A(t)=\cdf_C^B(t)$ for $t\geq0$).
\eprf
Because conceding in one dispute fully reveals $C$’s type and forces immediate capitulation in the other, any concession by $C$ is effectively a concession on both pies. Henceforth we impose $\cdf_C^A(t)=\cdf_C^B(t)$ for all $t\geq0$ in the no-concession subgame and write $\cdf_C$ for this common
cdf on $[0,\infty)$. Accordingly, for $t\geq0$, $\cdf_C(t)$ denotes the probability that $C$ has conceded (and therefore, by the lemma,
has conceded to both $A$ and $B$) by time $t$, conditional on no concession strictly before $t$.

\paragraph{Regularity of equilibrium strategies.}
We restrict attention to equilibria in which, in the no-concession subgame, each $\cdf_i$
may have
an atom at $t=0$ but has no atoms on $(0,\infty)$.\footnote{One can show that the PBE we characterize is the unique PBE with finitely many atoms. A proof is available on the authors' websites.} Formally, for each relevant cdf $F$ we assume:
(i) $F$ is right-continuous and nondecreasing with $F(0-)=0$;
(ii) $F$ is absolutely continuous on $(0,\infty)$ with density $f(\cdot)=F'(\cdot)$;
(iii) $f$ is piecewise continuous on $(0,\infty)$.
Accordingly, the hazard rate $\rate(t):=\frac{f(t)}{1-F(t)}$ is well-defined a.e.\ on $\{t>0:F(t)<1\}$ and is piecewise continuous.

This regularity condition is imposed on the equilibrium profile being characterized; the PBE definition permits deviations
that do not satisfy these regularity properties.

\paragraph{Posterior dynamics on the no-concession path.}
Along public histories with no concession up to calendar time $t$, let $z_i(t)$ denote the
posterior probability that player $i$ is behavioral.
In the no-concession subgame, Bayes' rule implies for $t>0$
\[
z_i(t)=\frac{z_i(0)}{1-\cdf_i(t)}\qquad\text{whenever }\cdf_i(t)<1.
\]
At $t=0^+$, $z_i(0^+)=\frac{z_i(0)}{1-F_i(0)}$.

\section{Benchmark: Two separate AG interactions}\label{Section: benchmark}

We recall the equilibrium outcome in the AG benchmark without
reputational spillovers. In this benchmark, the two negotiations are \emph{independent}:
the history of play and the identity/timing of concessions in the negotiation between $A$
and $C$ are observed only by $A$ and $C$ (and not by $B$), and analogously for the
negotiation between $B$ and $C$. Equivalently, beliefs about $C$'s type do not spill over
across negotiations, so each negotiation is a separate bilateral reputational war of attrition
\`a la AG with a single commitment type.

Fix a bilateral negotiation over surplus $\pi_{ij}>0$ between players $i$ and $j$, with
initial reputations $z_i(0),z_j(0)\in(0,1)$. Let $F^{AG}_{ji}(t)$ denote the (unconditional)
probability that $j$ has conceded to $i$ by time $t$ in the unique AG equilibrium.
In each bilateral negotiation, at most one player concedes with positive probability at
time $0$. Player $j$ concedes with positive probability to player $i$ at time $0$ if and only if
$z_j(0)<z_i(0)$. In particular,
\begin{align}
F^{AG}_{ji}(0)=\max\left\{1-\frac{z_j(0)}{z_i(0)},\,0\right\}.
\tag{Atom:AG}\label{Equation: atom AG}
\end{align}
When $z_j(0)<z_i(0)$, we refer to $j$ as the \emph{weak} player and to $i$ as the \emph{strong}
player.

After time $0$, players concede at a constant rate that makes the opponent indifferent
between waiting and conceding. Specifically, player $i$ is indifferent if
\[
r(1-\alpha)=(2\alpha-1)\,\frac{f^{AG}_{ji}(t)}{1-F^{AG}_{ji}(t)},
\]
where $\frac{f^{AG}_{ji}(t)}{1-F^{AG}_{ji}(t)}$ is $j$'s hazard rate of conceding at time $t$.
Hence, after time $0$, both players concede at the constant AG hazard
\[
\lambda^{AG}=\frac{r(1-\alpha)}{2\alpha-1}.
\]

There is a finite terminal time $T^{AG}_{ij}$ by which the posterior probability of being
a commitment type reaches $1$ and concessions stop. Let $T^{AG}_i$ denote the time at which
player $i$ is believed committed with probability $1$ (conditional on not conceding with
positive probability at time $0$). Then
\[
T^{AG}_{ij}=\min\{T^{AG}_i,T^{AG}_j\},
\qquad
T^{AG}_i=-\frac{1}{\lambda^{AG}}\log z_i(0).
\]

Given this equilibrium play, player $i$'s expected payoff in the bilateral negotiation over
surplus $\pi_{ij}$ is
\[
v^{AG}_{ij}
= \pi_{ij}\Big((1-\alpha)+(2\alpha-1)\,F^{AG}_{ji}(0)\Big).
\]
Thus, in the bilateral benchmark all payoff-relevant asymmetry is absorbed by a single time-0 adjustment; after that, concession hazards are stationary at $\lambda^{AG}$.

In summary, the AG bilateral benchmark yields two key payoff properties:
the stronger player receives strictly more than $(1-\alpha)\pi_{ij}$, while the weaker
player receives exactly $(1-\alpha)\pi_{ij}$. We emphasize these properties because
reputational spillovers in the three-player game overturn them, leading to qualitatively
different concession behavior and, consequently, different payoffs.

\section{Analysis and Main Result}\label{Section: analysis and main result}

This section characterizes equilibrium behavior and payoffs in the public, three-player
bargaining game. The key friction is that player $C$ has a \emph{global} type, so learning
about $C$ in one dispute necessarily carries over to the other. In particular, whenever both
negotiations are still unresolved, any concession by $C$ reveals that she is rational and
therefore triggers immediate resolution of \emph{both} disputes (in equilibrium, $C$ never
concedes in only one dispute at a positive time). After the first concession, the game reduces
to a standard bilateral AG reputational war of attrition in the remaining dispute.

We therefore focus on the \emph{no-concession subgame}, i.e.\ the path on which neither dispute
has ended yet. Along this path, each player's posterior reputation $z_i(t)$ drifts upward
because not conceding is evidence of being the behavioral (commitment) type. Two basic
discipline results sharply restrict equilibrium dynamics.

First, there is an \emph{endogenous deadline}: there exists a finite time $T<\infty$ such that,
conditional on no concession up to $T$, all remaining players are believed behavioral with
probability one. Equivalently, all rational types exhaust their concession probability mass by
$T$.

Second, before $T$ the game cannot feature prolonged ``pauses.'' On any interval strictly
before $T$, at most one player can be inactive (otherwise the remaining active player can
profitably shift concession mass earlier and save discounting). Moreover, once a player starts
conceding with positive density, she does not stop before $T$ (otherwise conceding right
before a pause is strictly dominated by waiting an instant).

Together, these restrictions imply a sharp phase structure. There may be an initial phase in
which exactly one peripheral is inactive; after at most that initial phase, all three players
are active, posteriors align, and everyone concedes with the constant bilateral AG
hazard until the common deadline. All departures from two independent AG
negotiations are driven by a single question: \emph{is a peripheral initially reputable enough
to remain inactive at the outset?} If not, beliefs reconcile immediately at $t=0^+$ and the
public-linkage of negotiations is payoff-irrelevant. If instead a peripheral is initially
dominant, equilibrium begins with a two-player phase in which that peripheral waits while
$C$ bargains with the other peripheral; because any concession by $C$ is effectively global,
this initial phase distorts early concession incentives and forces an immediate ``belief
reconciliation'' through a time-zero concession.

A notable implication of the indifference conditions is piecewise stationarity: after at most one initial two-player phase (and a possible time-zero adjustment), concession hazards are constant and coincide with the bilateral AG hazard.

To state the equilibrium precisely, define for each player $i\in\{A,B,C\}$ the time at which
$i$'s posterior reaches one along the no-concession path,
\[
T_i \;:=\; \inf\{t\ge 0:\ z_i(t)=1\}, \qquad \inf\varnothing=\infty.
\]
We say that player $i$ is \emph{active} on an interval $I\subset(0,\infty)$ if $F_i$ is strictly
increasing on $I$ (equivalently $\lambda_i(t)>0$ for a.e.\ $t\in I$), and \emph{inactive} on $I$
if $F_i$ is constant on $I$ (equivalently $\lambda_i(t)=0$ for a.e.\ $t\in I$). Finally, define
the activation time
\[
t_i \;:=\; \inf\{t>0:\ F_i(t)>F_i(0)\}, \qquad \inf\varnothing=\infty.
\]

\begin{proposition}[Equilibrium structure]\label{prop:structure}
There exists a unique equilibrium. Assume without loss of generality that $z_A(0)\ge z_B(0)$.
Along the no-concession path:

\begin{enumerate}
\item[(i)]
There exists $T<\infty$ such that $T_i=T$ for all $i\in\{A,B,C\}$.

\item[(ii)]
If $z_A(0)>\max\{z_B(0),z_C(0)\}$, then $t_A\in(0,T)$ and $t_B=t_C=0$ with 

\[
\lambda_B(t)=(1+\pi_{AC})\lambda^{AG}, \qquad \lambda_C(t)=\lambda^{AG}
\quad \text{for a.e. } t\in(0,t_A).
\]
The identity and size of the time-$0$ atom (at most one of $B$ and $C$ places mass at $0$) are uniquely pinned down
by the posterior-alignment condition at the activation time:
\[
z_B(t_A) = z_C(t_A) = z_A(0).
\]
For a.e.\ $t\in(t_A,T)$, all three players concede with hazard $\lambda^{AG}$.

\item[(iii)]
If $z_A(0)\le \max\{z_B(0),z_C(0)\}$, $t_i=0$ for all $i\in\{A,B,C\}$. Any time-$0$ atoms are uniquely pinned down by immediate posterior alignment:
\[
z_A(0^+) = z_B(0^+) = z_C(0^+),
\]
and for a.e.\ $t\in(0,T)$ all three players concede with hazard $\lambda^{AG}$.
\end{enumerate}
\end{proposition}

\begin{figure}[ht!]
\centering
\begin{minipage}{0.45\textwidth}
\centering
\begin{tikzpicture}
  \begin{axis}[
    xlabel={$t$},
    xlabel style={at={(ticklabel* cs:1)}, anchor=west, yshift=5pt},
    ylabel={Pr[behavioral]},
    ymin=0, ymax=1.12,
    xmin=0, xmax=1.3,
    samples=200,
    xtick={0.9242},
    xticklabels={$T$},
    ytick={0.1, 0.3, 0.5, 1},
    yticklabels={$z_B(0)$, $z_A(0)$, $z_C(0)$, $1$},
    axis lines=left,
    width=\textwidth,
    height=0.9\textwidth,
  ]

  \addplot[black, very thick, domain=0:0.9242]
    {0.5*exp(0.75*x)};
  \addplot[black, very thick, domain=0.9242:1.3] {1};
  \draw[dotted] (axis cs:0.9242,0) -- (axis cs:0.9242,1);
  \addplot[black, dotted, domain=0:0.9242] {1};

  \draw[->, green!50!black, thick]
    (axis cs:0,0.3)
    .. controls (axis cs:0.08,0.34) and (axis cs:0.08,0.46) ..
    (axis cs:0,0.5);

  \draw[->, blue, thick]
    (axis cs:0,0.1)
    .. controls (axis cs:0.12,0.2) and (axis cs:0.12,0.4) ..
    (axis cs:0,0.5);

  \end{axis}
\end{tikzpicture}
\vspace{-0.4cm}
\subcaption{\shortstack{$z_C > z_A > z_B$\\[2pt]($A$ and $B$ concede with atoms)}}
\end{minipage}\hfill
\begin{minipage}{0.45\textwidth}
\centering
\begin{tikzpicture}
  \begin{axis}[
    xlabel={$t$},
    xlabel style={at={(ticklabel* cs:1)}, anchor=west, yshift=5pt},
    ylabel={Pr[behavioral]},
    ymin=0, ymax=1.12,
    xmin=0, xmax=1.6,
    samples=200,
    xtick={0.4072, 1.3314},
    xticklabels={$t^*$, $T$},
    ytick={0.1, 0.2, 0.5, 1},
    yticklabels={$z_C(0)$, $z_B(0)$, $z_A(0)$, $1$},
    axis lines=left,
    width=\textwidth,
    height=0.9\textwidth,
  ]

  \addplot[blue, very thick, domain=0:0.4072] {0.2*exp(2.25*x)};
  \addplot[red, very thick, domain=0:0.4072] {0.3684*exp(0.75*x)};
  \addplot[black, very thick, domain=0.4072:1.3314]
    {0.5*exp(0.75*(x - 0.4072))};
  \addplot[black, very thick, domain=1.3314:1.6] {1};

  \draw[dotted] (axis cs:0.4072,0) -- (axis cs:0.4072,0.5);
  \draw[dotted] (axis cs:1.3314,0) -- (axis cs:1.3314,1);
  \addplot[orange, very thick, domain=0:0.4072] {0.5};
  \addplot[black, dotted, domain=0:1.3314] {1};

  \draw[->, red, thick]
    (axis cs:0,0.1)
    .. controls (axis cs:0.10,0.17) and (axis cs:0.10,0.30) ..
    (axis cs:0,0.3684);

  \end{axis}
\end{tikzpicture}
\vspace{-0.4cm}
\subcaption{\shortstack{$z_A > z_B > z_C$\\[2pt]($C$ concedes with atom)}}
\end{minipage}

\vspace{0.5cm}

\begin{minipage}{0.45\textwidth}
\centering
\begin{tikzpicture}
  \begin{axis}[
    xlabel={$t$},
    xlabel style={at={(ticklabel* cs:1)}, anchor=west, yshift=5pt},
    ylabel={Pr[behavioral]},
    ymin=0, ymax=1.12,
    xmin=0, xmax=1.9,
    samples=200,
    xtick={0.7153, 1.6395},
    xticklabels={$t^*$, $T$},
    ytick={0.1, 0.2, 0.5, 1},
    yticklabels={$z_B(0)$, $z_C(0)$, $z_A(0)$, $1$},
    axis lines=left,
    width=\textwidth,
    height=0.9\textwidth,
  ]

  \addplot[blue, very thick, domain=0:0.7153] {0.1*exp(2.25*x)};
  \addplot[red, very thick, domain=0:0.7153] {0.2925*exp(0.75*x)};
  \addplot[black, very thick, domain=0.7153:1.6395]
    {0.5*exp(0.75*(x - 0.7153))};
  \addplot[black, very thick, domain=1.6395:1.9] {1};

  \draw[dotted] (axis cs:0.7153,0) -- (axis cs:0.7153,0.5);
  \draw[dotted] (axis cs:1.6395,0) -- (axis cs:1.6395,1);
  \addplot[orange, very thick, domain=0:0.7153] {0.5};
  \addplot[black, dotted, domain=0:1.6395] {1};

  \draw[->, red, thick]
    (axis cs:0,0.2)
    .. controls (axis cs:0.08,0.22) and (axis cs:0.08,0.27) ..
    (axis cs:0,0.2925);

  \end{axis}
\end{tikzpicture}
\vspace{-0.4cm}
\subcaption{\shortstack{$z_A > z_C > z_B$\\[2pt]($C$ concedes with atom)}}
\end{minipage}\hfill
\begin{minipage}{0.45\textwidth}
\centering
\begin{tikzpicture}
  \begin{axis}[
    xlabel={$t$},
    xlabel style={at={(ticklabel* cs:1)}, anchor=west, yshift=5pt},
    ylabel={Pr[behavioral]},
    ymin=0, ymax=1.12,
    xmin=0, xmax=1.5,
    samples=200,
    xtick={0.2975, 1.2217},
    xticklabels={$t^*$, $T$},
    ytick={0.1, 0.4, 0.5, 1},
    yticklabels={$z_B(0)$, $z_C(0)$, $z_A(0)$, $1$},
    axis lines=left,
    width=\textwidth,
    height=0.9\textwidth,
  ]

  \addplot[blue, very thick, domain=0:0.2975] {0.256*exp(2.25*x)};
  \addplot[red, very thick, domain=0:0.2975] {0.4*exp(0.75*x)};
  \addplot[black, very thick, domain=0.2975:1.2217]
    {0.5*exp(0.75*(x - 0.2975))};
  \addplot[black, very thick, domain=1.2217:1.5] {1};

  \draw[dotted] (axis cs:0.2975,0) -- (axis cs:0.2975,0.5);
  \draw[dotted] (axis cs:1.2217,0) -- (axis cs:1.2217,1);
  \addplot[orange, very thick, domain=0:0.2975] {0.5};
  \addplot[black, dotted, domain=0:1.2217] {1};

  \draw[->, blue, thick]
    (axis cs:0,0.1)
    .. controls (axis cs:0.06,0.14) and (axis cs:0.06,0.22) ..
    (axis cs:0,0.256);

  \end{axis}
\end{tikzpicture}
\vspace{-0.4cm}
\subcaption{\shortstack{$z_A > z_C > z_B$\\[2pt]($B$ concedes with atom)}}
\end{minipage}
\caption{Equilibrium posterior dynamics under four configurations. }\label{Figure: unique equilibrium, four cases}
\end{figure}
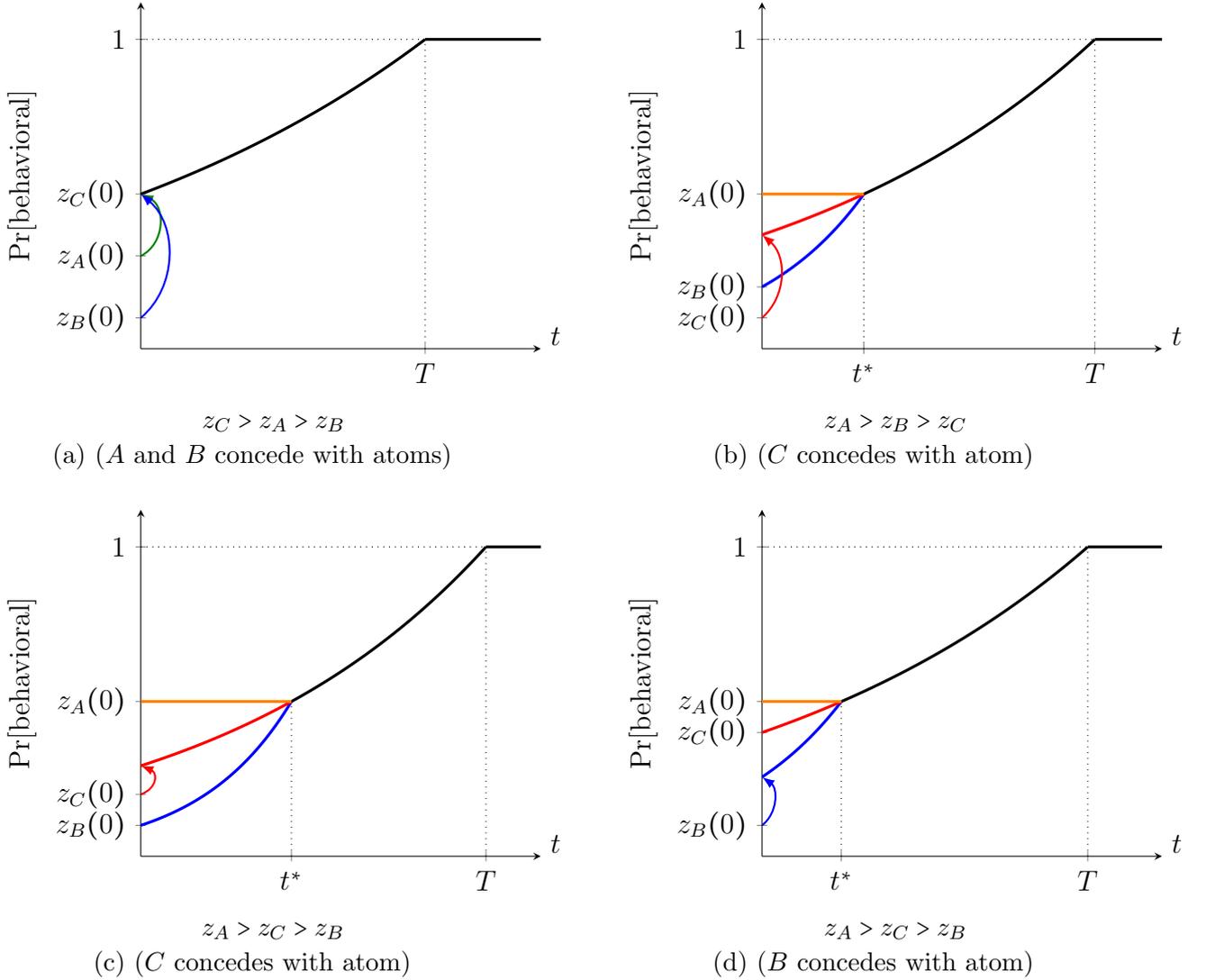

Proposition \ref{prop:structure} leaves one object to be determined: 
when a peripheral is strictly dominant, which player concedes at time $0$, and with what probability? We now show that the answer follows  from the requirement of posterior alignment at the moment the last player becomes active. A complete version of Proposition~\ref{prop:structure} is presented in Proposition~\ref{prop:full} in the appendix. 

Suppose $z_A(0)>\max\{z_B(0),z_C(0)\}$, so that $A$ is initially inactive 
while $B$ and $C$ are active on $(0,t_A)$. On that interval, Proposition 
\ref{prop:structure} implies constant hazards
\[
\lambda_C=\lambda^{AG}, 
\qquad 
\lambda_B=(1+\pi_{AC})\lambda^{AG}.
\]
Hence along the no-concession path,
\[
z_C(t)=z_C(0^+)e^{\lambda^{AG} t},
\qquad
z_B(t)=z_B(0^+)e^{(1+\pi_{AC})\lambda^{AG} t},
\qquad
z_A(t)=z_A(0).
\]
Because joint activity forces posterior alignment, the activation time $t_A$ 
must satisfy
\[
z_B(t_A)=z_C(t_A)=z_A(0).
\tag{Alignment}
\]
First consider the hypothetical case with no time-zero atoms:
$z_i(0^+)=z_i(0)$ for $i\in\{B,C\}$. 
Define the catch-up times
\[
\tilde{t}_C=\frac{1}{\lambda^{AG}}
\log\!\left(\frac{z_A(0)}{z_C(0)}\right),
\qquad
\tilde{t}_B=\frac{1}{(1+\pi_{AC})\lambda^{AG}}
\log\!\left(\frac{z_A(0)}{z_B(0)}\right).
\]

These are the times at which $C$ and $B$, respectively, would reach $z_A(0)$ 
under the initial-phase hazards. In other words, these times compare the endogenous growth speeds of reputations under the initial-phase hazard rates. If $\tilde{t}_C=\tilde{t}_B$, no time-zero atom is required: both posteriors reach 
$z_A(0)$ simultaneously. If instead $\tilde{t}_C>\tilde{t}_B$, then absent atoms $C$ would reach $z_A(0)$ strictly 
later than $B$. To satisfy the alignment condition, $C$ must therefore 
receive an upward jump in her posterior at time $0$, implemented by a 
time-zero concession atom. Conversely, if $\tilde{t}_B>\tilde{t}_C$, then $B$ is the laggard and must concede 
with positive probability at time $0$.

Consider the case $\tilde{t}_C>\tilde{t}_B$, so that $C$ must concede at $0$.
Since $C$ places mass $F_C(0)$ at time $0$, her posterior jumps to
\[
z_C(0^+)=\frac{z_C(0)}{1-F_C(0)}.
\]
The activation time is then determined by $B$'s catch-up time,
\[
t_A=\frac{1}{(1+\pi_{AC})\lambda^{AG}}
\log\!\left(\frac{z_A(0)}{z_B(0)}\right).
\]
Imposing alignment $z_C(t_A)=z_A(0)$ yields
\[
z_C(0^+)=z_A(0)
\Big(\frac{z_B(0)}{z_A(0)}\Big)^{\frac{1}{1+\pi_{AC}}}
=
z_A(0)^{\frac{\pi_{AC}}{1+\pi_{AC}}}
z_B(0)^{\frac{1}{1+\pi_{AC}}}.
\]
Hence
\[
F_C(0)
=
1-
\frac{z_C(0)}
{z_A(0)^{\frac{\pi_{AC}}{1+\pi_{AC}}}
z_B(0)^{\frac{1}{1+\pi_{AC}}}}.
\]

An analogous calculation yields the formula in the case where $B$ 
is the laggard. The time-zero atom is therefore assigned to whichever active player would 
otherwise ``arrive late'' to $z_A(0)$ under the initial hazard rates.
Because $B$ concedes at the accelerated rate 
$(1+\pi_{AC})\lambda^{AG}$, increasing $\pi_{AC}$ reduces $\tilde{t}_B$ and 
makes it more likely that $C$ is the laggard. Thus sufficiently high 
stakes in the $A$--$C$ negotiation shift the time-zero concession 
onto the center.

Let $t^*$ denote the time at which posteriors first align and all three players are active (so $t^*=t_A$ in Proposition~\ref{prop:structure}(ii), and $t^*=0$ in Proposition~\ref{prop:structure}(iii). Figure~\ref{Figure: unique equilibrium, four cases} illustrates the posterior paths described in Proposition~\ref{prop:structure}. Each panel plots posterior commitment probabilities $z_i(t)$ along the no-concession path. Vertical jumps at $t=0$ correspond to time-$0$ concession atoms, which raise the non-conceding player's posterior to $z_i(0+)=z_i(0)/(1-F_i(0))$. In panels (b)-(d), $A$ is initially inactive, while $B$ and $C$ are active with constant hazard rates $\lambda_B=(1+\pi_{AC})\lambda^{AG}$ and $\lambda_C=\lambda^{AG}$, so $z_B$ grows faster than $z_C$ until the activation time $t^*$ when posteriors align at $z_A(0)$. Panels (c) and (d) differ by which active player is the laggard: if $\tilde{t}_C>\tilde{t}_B$, then $C$ bears the time-$0$ atom; if $\tilde{t}_B>\tilde{t}_C$, then $B$ bears the time-$0$ atom. After alignment at $t^*$, all players concede at hazard $\lambda^{AG}$ until the common terminal time $T$.

Proposition~\ref{prop:structure} pins down the equilibrium path up to a single remaining
degree of freedom: the identity and size of any time--0 concession atom. This atom is needed
to reconcile the (initially asymmetric) posterior growth rates in the initial phase with the
requirement that posteriors coincide when the last player becomes active. As discussed
above, the posterior-alignment condition at the activation time uniquely selects this atom.

Once this object is fixed, equilibrium behavior is fully characterized. After at most an
initial two-player phase, all active players concede at the bilateral AG hazard
$\lambda^{AG}$ until a common terminal time. This phase structure makes transparent where
equilibrium payoffs can differ from the benchmark of two independent bilateral
\cite{abreu2000bargaining} negotiations: since the no-concession dynamics coincide with
$\lambda^{AG}$ after the initial adjustment, any payoff wedge must be generated entirely by
that initial adjustment—most notably by the time--0 atom required for posterior alignment
when a peripheral is initially dominant. We now formalize this comparison.

Before stating Proposition~\ref{Proposition: unique equilibrium payoffs}, let us define expected payoffs. Fix a strategy profile $\sigma$. For each player $i\in\{A,B,C\}$, let $v_i(\sigma)$ denote
the \emph{time-$0$} expected payoff of $i$'s \emph{rational type} under $\sigma$. We write
$v_i^*:=v_i(\sigma^*)$ for the equilibrium payoff. To connect payoffs to the indifference conditions used in the Appendix, it is convenient to define the payoff to
a peripheral from conceding at a deterministic time. Fix $i\in\{A,B\}$, let $k$ denote the other
peripheral, and consider the plan:
\begin{quote}
``If no one concedes strictly before $t$, concede to $C$ at time $t$. If $k$ concedes at some $y<t$ before $C$
concedes, then from calendar time $y$ onward play the unique AG equilibrium in the remaining
$i$--$C$ negotiation.''
\end{quote}
Let $U_i(t;\sigma_{-i})$ be the expected payoff at time $0$ from this plan, and write
$F_m(t-):=\lim_{s\uparrow t}F_m(s)$ for left limits. Then a first-event decomposition yields
\begin{align}
U_i(t;\sigma_{-i})
&=\alpha\pi_{iC}\int_{[0,t)} e^{-ry}\bigl(1-F_k(y-)\bigr)\,dF_C(y)
\;+\;\int_{[0,t)} e^{-ry}\bigl(1-F_C(y)\bigr)\,V_{iC}^{AG}\!\bigl(z_i(y),z_C(y)\bigr)dF_k(y)\,
\notag\\
&\quad\;+\;(1-\alpha)\pi_{iC}e^{-rt}\bigl(1-F_C(t)\bigr)\bigl(1-F_k(t-)\bigr),
\label{eq:Ui_plan_payoff}
\end{align}
where we recall that $V_{iC}^{AG}(z_i,z_C)$ denotes $i$'s continuation payoff in the induced bilateral AG game starting at calendar time $y$
when the remaining negotiation starts with reputations $(z_i,z_C)$.\footnote{All integrals are Lebesgue–Stieltjes, so they incorporate any atoms of $F_C$ and $F_k$ at time $t=0$.} 
The three terms in \eqref{eq:Ui_plan_payoff} correspond, respectively, to: (i) $C$ concedes before $t$ and before
$k$ concedes; (ii) $k$ concedes before $t$ and before $C$ concedes, inducing the AG continuation in $i$--$C$;
and
(iii) no one concedes before $t$, so $i$ concedes at $t$.\footnote{Under the regularity restriction we impose on the \emph{equilibrium profile},
there is no simultaneous-concession term for $t>0$.}

In the equilibrium characterized below, once a peripheral becomes active he is indifferent over the times in the
support of his concession distribution, so his equilibrium payoff, $v^*_i$, can be recovered from \eqref{eq:Ui_plan_payoff}
by evaluating $U_i(\cdot;\sigma_{-i}^*)$ at any such time (e.g.\ at $t_i:=\inf\{t\ge 0:F_i^*(t)>0\}$).

An analogous decomposition applies to $C$.\footnote{Recall that by Lemma~\ref{Lemma: C concedes to both once she concedes to one}, it is without loss
on the equilibrium path to treat $C$'s concession in the no-concession subgame as simultaneous across negotiations.} 
Define $U_C(t;\sigma_{-C})$ as the time-$0$ payoff to $C$'s rational type from the plan:
\begin{quote}
``Do not concede before $t$; if no one concedes before $t$, concede to both $A$ and $B$ at time $t$; if $A$ (resp.\ $B$)
concedes first at some $y<t$, then from $y$ onward play the unique AG equilibrium in the remaining $B$--$C$ (resp.\ $A$--$C$)
negotiation.'' \end{quote}

Recall that $V_{CB}^{AG}(z_C,z_B)$ (resp.\ $V_{CA}^{AG}(z_C,z_A)$) denotes $C$'s continuation payoff in the induced AG game, i.e., the remaining
$B$--$C$ (resp.\ $A$--$C$) negotiation when it starts with reputations $(z_C,z_B)$ (resp.\ $(z_C,z_A)$). Then
\begin{align}
U_C(t;\sigma_{-C})
&=\alpha(\pi_{AC}+1)\,F_A(0)F_B(0)
\notag\\
&\quad+\int_{[0,t)} e^{-ry}\bigl(1-F_B(y)\bigr)\,\Big(\alpha\pi_{AC}+V_{CB}^{AG}\!\bigl(z_C(y),z_B(y)\bigr)\Big)dF_A(y)
\notag\\
&\quad+\int_{[0,t)} e^{-ry}\bigl(1-F_A(y)\bigr)\,\Big(\alpha+V_{CA}^{AG}\!\bigl(z_C(y),z_A(y)\bigr)\Big)dF_B(y)
\notag\\
&\quad+(1-\alpha)(\pi_{AC}+1)\,e^{-rt}\bigl(1-F_A(t)\bigr)\bigl(1-F_B(t)\bigr).
\label{eq:UC_plan_payoff}
\end{align}

The time-$0$ concession events are partitioned as follows. The explicit
term $F_A(0)F_B(0)$ isolates the event that both peripherals concede at $0$.
Because the Lebesgue--Stieltjes integrals are taken over $[0,t)$, they also pick
up any atoms at $y=0$; however, the factor $1-F_B(y)$ (resp.\ $1-F_A(y)$) implies
that the atom of $dF_A$ at $0$ (resp.\ $dF_B$ at $0$) contributes only on the
one-sided event that $A$ concedes at $0$ while $B$ does not (resp.\ that $B$
concedes at $0$ while $A$ does not).\footnote{Thus simultaneous time-$0$ concessions are
counted only in the first term.}

We can now state Proposition~\ref{Proposition: unique equilibrium payoffs} which compares each player's equilibrium expected payoff under public negotiations---where reputational spillovers are present---with the benchmark of two independent bilateral negotiations \`a la AG from Section~\ref{Section: benchmark}.

\bprop[Equilibrium payoffs]\label{Proposition: unique equilibrium payoffs}
Assume without loss of generality that $z_A(0)\ge z_B(0)$. In the unique equilibrium:
\begin{enumerate}
\item If $z_A(0)\le \max\{z_B(0),z_C(0)\}$, then
\[
v_i^*=v^{AG}_{i,C}\ \text{for}\ i\in\{A,B\},\qquad 
v_C^*=v^{AG}_{C,A}+v^{AG}_{C,B}.
\]
\item If $z_A(0)>z_B(0)\geq z_C(0) $, then
\[
v_A^*<v^{AG}_{A,C},\qquad v_B^*>v^{AG}_{B,C},\qquad 
v_C^*=v^{AG}_{C,A}+v^{AG}_{C,B}.
\]
\item  If $z_A(0)>z_C(0)>z_B(0)$, then
\[
v_A^*<v^{AG}_{A,C},\qquad v_B^*\ge v^{AG}_{B,C},\qquad 
v_C^*<v^{AG}_{C,A}+v^{AG}_{C,B}.
\]
\end{enumerate}
\eprop

Proposition~\ref{Proposition: unique equilibrium payoffs} is the payoff counterpart of the
phase structure in Proposition~\ref{prop:structure}. When $z_A(0)\le \max\{z_B(0),z_C(0)\}$,
no peripheral is \emph{uniquely} most reputable. In this region Proposition~\ref{prop:structure}(iii)
implies immediate posterior alignment at $0^+$, and from that point onward the no-concession
dynamics coincide with the bilateral AG hazard $\lambda^{AG}$. Hence reputational
spillovers are payoff-irrelevant: each player’s expected payoff coincides with the bilateral
AG benchmark in each dispute.

Spillovers are payoff-relevant if and only if a peripheral is initially dominant,
$z_A(0)>\max\{z_B(0),z_C(0)\}$. Then Proposition~\ref{prop:structure}(ii) implies an initial
two-player phase in which the dominant peripheral $A$ optimally remains inactive while $B$
and $C$ ``bargain in earnest.'' Because $C$'s type is global, any concession by $C$ is effectively
\emph{global}---it resolves both negotiations at once and reveals $C$'s flexibility. This makes
conceding comparatively more costly for $C$ in the initial phase and forces an initial
reallocation of concession probability (including, when required by posterior alignment,
a time--0 atom pinned down by the condition that posteriors coincide when $A$ becomes active).

The payoff implications are then systematic and mirror the two subregions in the dominant-peripheral case.
If $z_A(0)>z_B(0)\ge z_C(0)$, spillovers operate purely through redistribution across the
peripherals: the dominant peripheral $A$ is strictly worse off than in the bilateral benchmark,
while the weaker peripheral $B$ is strictly better off, and the center's aggregate payoff remains
exactly at the sum of her bilateral AG payoffs. If instead $z_A(0)>z_C(0)>z_B(0)$, spillovers are
most consequential: the center is \emph{strictly} worse off than in the bilateral benchmark, the
dominant peripheral again loses, and the weakest peripheral weakly benefits. In particular, the
bilateral ``toughness pays'' logic is overturned precisely in the region where belief consistency
forces a nontrivial initial adjustment.

Figure~\ref{Figure: comparison AG} provides a numerical illustration of Proposition~\ref{Proposition: unique equilibrium payoffs} by contrasting the equilibrium posterior dynamics with those under two independent AG negotiations. The figure makes clear that the payoff wedge is generated entirely by the initial adjustment (including any time–0 atom) required for posterior alignment.

\begin{figure}[ht]
\centering
\begin{minipage}{0.45\textwidth}
\centering
\begin{tikzpicture}
  \begin{axis}[
    xlabel={$t$},
    xlabel style={at={(ticklabel* cs:1)}, anchor=west, yshift=5pt},
    ylabel={Pr[behavioral]},
    ymin=0, ymax=1.12,
    xmin=0, xmax=1.8,
    samples=200,
    xtick={0.796, 1.477},
    xticklabels={$t^*$, $T$},
    ytick={0.1, 0.2, 0.6, 1},
    yticklabels={$z_B(0)$, $z_C(0)$, $z_A(0)$, $1$},
    major tick length=4pt,
    ytick style={thick},
    axis lines=left,
    width=\textwidth,
    height=0.9\textwidth,
  ]

  \addplot[blue, very thick, domain=0:0.796] {0.1*exp(2.25*x)};

  \addplot[red, very thick, domain=0:0.796] {0.330*exp(0.75*x)};

  \addplot[black, very thick, domain=0.796:1.477]
    {0.6*exp(0.75*(x - 0.796))};

  \addplot[black, very thick, domain=1.477:1.8] {1};

  \draw[dotted] (axis cs:0.796,0) -- (axis cs:0.796,0.6);
  \draw[dotted] (axis cs:1.477,0) -- (axis cs:1.477,1);

  \addplot[black, dotted, domain=0:0.796] {0.6};
  \addplot[black, dotted, domain=0:1.477] {1};

  \draw[->, red, thick]
    (axis cs:0,0.2)
    .. controls (axis cs:0.10,0.23) and (axis cs:0.10,0.30) ..
    (axis cs:0,0.330);

  \end{axis}
\end{tikzpicture}
\vspace{-0.4cm}
\subcaption{Multilateral bargaining}\label{Figure: comparison multilateral}
\end{minipage}\hfill
\begin{minipage}{0.45\textwidth}
\centering
\begin{tikzpicture}
  \begin{axis}[
    xlabel={$t$},
    xlabel style={at={(ticklabel* cs:1)}, anchor=west, yshift=5pt},
    ylabel={Pr[behavioral]},
    ymin=0, ymax=1.12,
    xmin=0, xmax=2.6,
    samples=200,
    xtick={0.681, 2.145},
    xticklabels={$T^{AC}_{AG}$, $T^{BC}_{AG}$},
    ytick={0.1, 0.2, 0.6, 1},
    yticklabels={$z_B(0)$, $z_C(0)$, $z_A(0)$, $1$},
    major tick length=4pt,
    ytick style={thick},
    axis lines=left,
    width=\textwidth,
    height=0.9\textwidth,
  ]

  \addplot[red, very thick, domain=0:0.681]
    {0.6*exp(0.75*x)};

  \addplot[red, very thick, domain=0.681:2.6] {1};

  \addplot[blue, very thick, domain=0:2.145]
    {0.2*exp(0.75*x)};

  \addplot[blue, very thick, domain=2.145:2.6] {1};

  \draw[dotted] (axis cs:0.681,0) -- (axis cs:0.681,1);
  \draw[dotted] (axis cs:2.145,0) -- (axis cs:2.145,1);

  \addplot[black, dotted, domain=0:2.145] {1};

  \draw[->, red, thick]
    (axis cs:0,0.2)
    .. controls (axis cs:0.12,0.30) and (axis cs:0.12,0.50) ..
    (axis cs:0,0.6);

  \draw[->, blue, thick]
    (axis cs:0,0.1)
    .. controls (axis cs:0.10,0.12) and (axis cs:0.10,0.18) ..
    (axis cs:0,0.2);

  \end{axis}
\end{tikzpicture}
\vspace{-0.4cm}
\subcaption{Bilateral bargaining (AG benchmark)}\label{Figure: comparison bilateral}
\end{minipage}
\caption{Comparison of concession behavior: multilateral bargaining (left) vs.\ bilateral bargaining (right).}\label{Figure: comparison AG}
\end{figure}
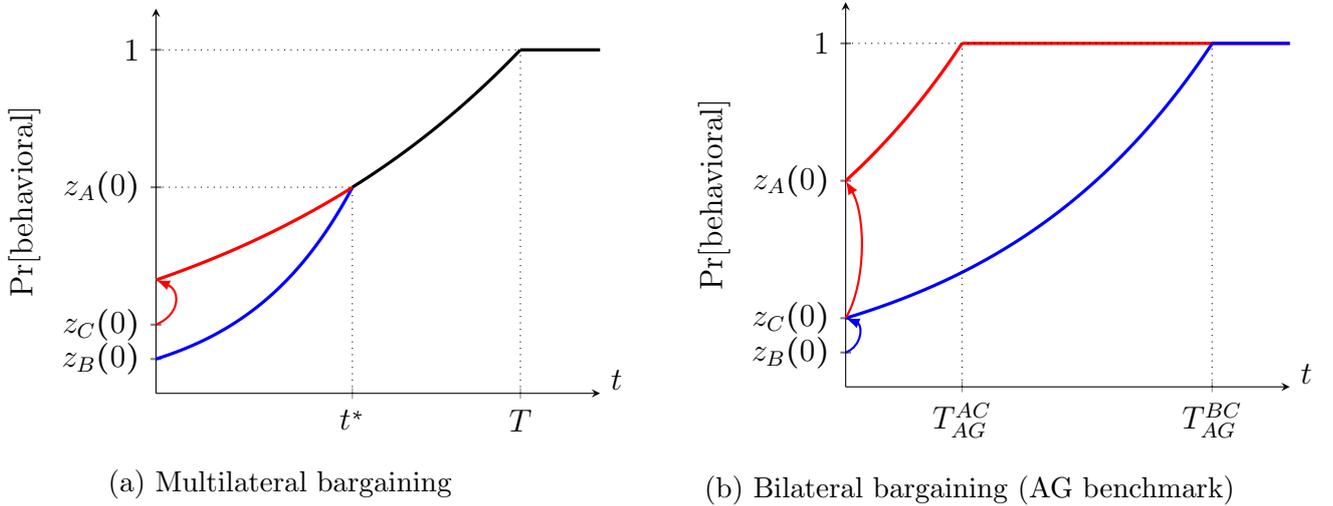

Specifically, Figure~\ref{Figure: comparison AG} illustrates the equilibrium concession dynamics with reputational spillovers and contrasts it to the bilateral AG benchmark.
In the multilateral game (left panel), player~C (in red) concedes with positive probability to both A and B at $t=0$, producing an upward jump in $C$’s posterior reputation at time $0$ (to $z_C(0+)$) to maintain consistent posterior beliefs across the two negotiations. Thereafter, the initial two-player phase begins: $B$ concedes at a hazard rate of $2\rate^{AG}$ (in blue) and $C$ concedes at a rate of $\rate^{AG}$. Given these rates of concession and the initial atom by $C$, they reach $A$'s prior $z_A(0)$ at time $t^*=t_A$. Once $A$'s prior is reached, all three players concede at a rate $\rate^{AG}$ up to time $T$, where all three players are believed to be committed with probability $1$. 

The right panel considers the corresponding bilateral benchmarks. In the bilateral negotiation between $A$ and $C$, $C$ is initially the weaker player ($z_C(0)<z_A(0)$) and therefore concedes with an atom to $A$ at time $0$, raising the common posterior to $z_A(0)$. Thereafter, both players concede at a rate $\rate^{AG}$ until $T_{AC}^{AG}$.
In the bilateral negotiation between $B$ and $C$, $B$ is initially the weaker player ($z_B(0)<z_C(0)$) and therefore concedes with an atom to $C$ at time $0$, raising the common posterior to $z_C(0)$. Thereafter, both players concede at a rate $\rate^{AG}$ until $T_{BC}^{AG}$. 

Two contrasts with the multilateral equilibrium are immediate. First, in the bilateral benchmark the weaker player in each negotiation is the one who concedes with positive probability at time $0$ (here $C$ in the negotiation between $A$ and $C$ and $B$ in the negotiation between $B$ and $C$), whereas in the multilateral environment $C$ concedes with an atom to both players even though $z_C(0)>z_B(0)$. Second, the multilateral equilibrium exhibits a longer common terminal time than the bilateral benchmark between $A$ and $C$ ($T>T_{AC}^{AG}$) but a shorter terminal time than the bilateral benchmark between $B$ and $C$ ($T<T_{BC}^{AG}$), reflecting the accelerated concession rate of the weak peripheral during the initial two--player phase.

Propositions~\ref{prop:structure} and \ref{Proposition: unique equilibrium payoffs} imply that any departure from the benchmark of two independent AG negotiations is generated entirely by the initial adjustment required for posterior alignment—after alignment, concession hazards coincide with the bilateral AG rate and the subsequent dynamics are identical. In the dominant-peripheral region, this adjustment takes the form of a single time‑0 atom borne by the player who would otherwise ‘arrive late’ to the dominant peripheral’s reputation under the initial-phase hazards. The closed-form atom pinned down by the alignment condition therefore delivers immediate comparative statics in the stakes parameter $\pi_{AC}$. If $B$ is the bearer of the atom, $C$'s ``reputation rent'' from $B$'s immediate concession shrinks as $\pi_{AC}$ increases. If $C$ is the bearer of the atom, the spillover component of the peripherals' equilibrium payoffs is increasing in $\pi_{AC}$: a higher $\pi_{AC}$ raises the probability that $C$ concedes immediately at time $0$, which benefits both peripherals. 

\section{A vanishing-prior limit with fixed relative reputations}
\label{sec:vanishing}

A natural question is whether the payoff reversal in Proposition~\ref{Proposition: unique equilibrium payoffs}
is an artifact of treating commitment probabilities as fixed primitives or whether it
survives when those probabilities become small. To address this, we study a
vanishing-prior limit along sequences for which all three priors converge to zero while
relative reputations remain fixed.

Throughout this section, we set $\pi_{AC}=\pi_{BC}=1$ and focus on the region
$z_A(0)>z_C(0)>z_B(0)$, where spillovers are most consequential
(Proposition~\ref{Proposition: unique equilibrium payoffs}, case~3). Fix constants $\kappa_A,\kappa_B>1$ and,
for each sufficiently small $\varepsilon>0$, define
\[
z_B^\varepsilon(0)=\varepsilon,\qquad
z_C^\varepsilon(0)=\kappa_B\varepsilon,\qquad
z_A^\varepsilon(0)=\kappa_A\kappa_B\varepsilon.
\]
Thus,
\[
\frac{z_A^\varepsilon(0)}{z_C^\varepsilon(0)}=\kappa_A,
\qquad
\frac{z_C^\varepsilon(0)}{z_B^\varepsilon(0)}=\kappa_B
\]
for every $\varepsilon$, while all three commitment probabilities vanish as
$\varepsilon\downarrow0$. Let $v_i^*(\varepsilon)$ denote player $i$'s equilibrium payoff
in the multilateral game under these priors, and let $v_{ij}^{AG}(\varepsilon)$ denote
player $i$'s payoff in the corresponding bilateral AG benchmark against $j$.

\begin{proposition}[Vanishing-prior limit]
\label{prop:vanishing}
As $\varepsilon\downarrow0$,
\begin{align*}
v_{CA}^{AG}(\varepsilon)+v_{CB}^{AG}(\varepsilon)-v_C^*(\varepsilon)
&\longrightarrow
(2\alpha-1)\frac{\min\{\kappa_A,\kappa_B\}-1}{\kappa_B}>0, \\
v_{AC}^{AG}(\varepsilon)
&> v_A^*(\varepsilon)\qquad \text{for every }\varepsilon>0,\\
v_B^*(\varepsilon)
&\ge v_{BC}^{AG}(\varepsilon)\qquad \text{for every }\varepsilon>0,
\end{align*}
with strict inequality in the last display if and only if $\kappa_A>\kappa_B$.
\end{proposition}

The limit in Proposition~\ref{prop:vanishing} is directional: it depends on the fixed
ratio pair $(\kappa_A,\kappa_B)$. When $\kappa_A>\kappa_B$, player $C$ is the laggard
under the initial-phase hazards and therefore bears the time-zero atom, which benefits
both peripherals. When $\kappa_A\le \kappa_B$, player $B$ is the laggard and concedes
at time zero instead. Even in that case, the center remains strictly worse off than under
two independent bilateral AG negotiations, because the time-zero adjustment is smaller
than the bilateral concession that $B$ would make to $C$ in isolation.

\section{Discussion and robustness}\label{Section: Extensions}

This section discusses some extensions that relax key modeling assumptions of the baseline such as the the simultaneity of negotiations, full observability of concessions, or that there are only two peripherals, and illustrates how the mechanism identified above operates in these environments.

\paragraph{Sequential negotiations:} A maintained assumption of the baseline model is that both negotiations are simultaneous. In Section~\ref{Sec: sequential} of the Online Appendix, we study a sequential environment in which $A$ and $C$ bargain first, and the negotiation between $B$ and $C$ begins only once the first dispute is resolved. Proposition~\ref{prop:SequentialUnique} shows that the equilibrium characterization carries over: there is a unique equilibrium with a common terminal time and piecewise-constant hazard rates, and the time-$0$ atom
(if any) is uniquely pinned down by a boundary condition on posteriors. The
key finding is that the central player's disadvantage is not an artifact of
simultaneity. In fact, the sequential environment sharpens it: even when all three players have \emph{identical} initial reputations, player $C$ may be forced to concede with a positive atom at time $0$, and is strictly worse off than in the bilateral benchmark. The reason is that the prospect of a future negotiation with $B$ enters $C$'s indifference condition in the first stage. This raises the effective cost to $C$ of holding firm against $A$. This is because conceding to $A$ is now less costly since it triggers the second-stage game in which $C$ retains her reputation. As a result, $A$ concedes at a strictly higher rate than $\lambda^{AG}$, his posterior grows faster than $C$'s, and belief alignment at the terminal time can require an upfront concession by $C$ even when no asymmetry in initial reputations would have compelled one in the bilateral benchmark.

\paragraph{Partial observability.}
A second assumption we relax is full observability of concessions. In
Section~\ref{Section: partial} of the Online Appendix, we study the case in
which the uninvolved peripheral observes only the timing of an agreement, not
the identity of the conceding party. An agreement is then informative about
$C$'s type regardless of who conceded: since commitment types never concede,
observing a settlement causes the uninvolved peripheral to revise downward his
belief that $C$ is behavioral. This downward jump can make $C$ the perceived
weaker party in the remaining negotiation and force an immediate concession by
$C$ the instant the outside agreement is observed. Proposition~\ref{prop:partial} shows that this force operates even when all three players have identical initial reputations---a case where the baseline features no atoms and all players concedesymmetrically at $\lambda^{AG}$. Under partial observability, $C$ is forced to concede with a strictly positive atom at $t=0$ and at a rate strictly below $\lambda^{AG}$ thereafter, while $A$ and $B$ continue to concede at $\lambda^{AG}$ and earn strictly more than in the bilateral benchmark. As a consequence, even if $C$ is the strongest player initially, she may concede with an atom at time $0$ yielding a strictly lower payoff than the bilateral benchmark. Thus, partial observability amplifies rather than attenuates the central player's disadvantage.

\paragraph{Beyond the three-player star:} A natural question is whether the mechanism identified in the three-player star extends to richer star networks. Section~\ref{sec:online-four-player} in the Online Appendix derives the corresponding local indifference conditions for an equal-pie four-player star. Relative to the three-player case, these conditions depend on continuation values from the three-player equilibrium, so the implied hazard system is a system of ODEs leading to time-varying hazard rates.

Section~\ref{sec:online-four-player} also reports a numerical solution for a representative parameterization. In that example, the equilibrium exhibits a sequential activation pattern and preserves the same qualitative force as in the three-player model: the weakest peripheral benefits from spillovers, while the center is worse off relative to the bilateral benchmark. More broadly, the structure of the four-player system suggests that the methods developed here can be extended to obtain a fuller analytical characterization of the $n$-peripheral star. 

Together, these extensions reinforce the main message of the paper. Reputational spillovers arise from the combination of a global type and belief consistency across negotiations. Relaxing simultaneity, full observability, or adding more peripheral players does not eliminate these forces and can continue to make the central player worse off relative to the bilateral benchmark.

\end{spacing}
\appendix

\section{Appendix: Proofs}\label{Section: appendix}

\subsection{Notation and Preliminaries.}

Throughout the appendix we analyze the \emph{no-concession subgame}, i.e., the subgame in which
both negotiations between $A$ and $C$ and between $B$ and $C$ are still contested. In this subgame the public history is
summarized by calendar time $t\ge 0$.

\paragraph{Concession-time distributions.}
Fix an equilibrium profile in the no-concession subgame and let
$F=(F_A,F_B,F_C)$ denote the induced cdfs of concession times.
For each player $i\in N:=\{A,B,C\}$, $F_i(t)$ is the probability (under the equilibrium profile)
that player $i$ has conceded by time $t$ \emph{conditional on no earlier concession}. Since behavioral types never concede, $F_i(t)\le 1-z_i(0)$ for all $t$, and in the equilibrium
characterized below rational types concede with probability one along the no-concession path, so
$F_i(\infty)=1-z_i(0)$.

For $t>0$ write $F_i(t-):=\lim_{s\uparrow t}F_i(s)$.

\paragraph{Posteriors, hazard rates, and ``survival''.}
Let $z_i(t)$ denote the posterior probability that $i$ is behavioral conditional on
no concession up to time $t$. Bayes' rule implies that whenever $F_i(t)<1$,
\begin{equation}\label{eq:posterior_def_appendix}
z_i(t)=\frac{z_i(0)}{1-F_i(t)}.
\end{equation}
On $(0,\infty)$ let $f_i(t):=F_i'(t)$ denote the (a.e.) density and define the hazard rate
\[
\rate_i(t):=\frac{f_i(t)}{1-F_i(t)}\qquad (t>0,\ F_i(t)<1).
\]

Using \eqref{eq:posterior_def_appendix}, we will repeatedly use the identity
\begin{equation}\label{eq:posterior_growth_appendix}
z_i(t)=z_i(0+)\;\exp\;\;\!\Bigl(\int_0^t \rate_i(u)\,du\Bigr),
\qquad z_i(0+):=\frac{z_i(0)}{1-F_i(0)}.
\end{equation}

Since $\cdf_i\cd$ is continuous on $\real_+$, so is $z_i(t)$. Equipped with this, the players' indifference conditions are: 
\begin{align}
    \icA \tag{IC:A},\label{Equation:IC-A}\\
    \icB \tag{IC:B},\label{Equation:IC-B}\\
    \icC \tag{IC:C}, \label{Equation:IC-C}
\end{align}
where we recall that $g_i(\cdot)$ (and resp. $g_C^j$)  is probability of an immediate concession by player $i$ (resp. $C$) in the induced bilateral continuation if the other peripheral concedes now.

Moreover, the complementary-slackness implications are:
\begin{align*}
\rate_C(t)+\rate_B(t)g_C^A(t)>\rate^{AG} &\implies \rate_A(t)=0,\\
\rate_C(t)+\rate_A(t)g_C^B(t)>\rate^{AG} &\implies \rate_B(t)=0,\\
\rate_A(t)\big(\pi_{AC}+g_B(t)\big)+\rate_B(t)\big(1+\pi_{AC}g_A(t)\big)>(1+\pi_{AC})\rate^{AG}
&\implies \rate_C(t)=0.
\end{align*}

For $0\le s\le t$ define the survival function
\[
\Lambda_i(t,s):=\exp\;\;\!\Bigl(-\int_s^t \lambda_i(u)\,du\Bigr),\qquad \Rate_i(t):=\Lambda_i(t,0).
\]

\paragraph{Activation and terminal times.}
Let
\[
t_i:=\inf\{t>0:\ F_i(t)>F_i(0)\}
\]
denote the first time at which player $i$ concedes with positive density (i.e., after any atom at $t=0$),
and let
\[
T_i:=\inf\{t\ge 0:\ z_i(t)=1\}
\]
denote the time at which $i$ is believed behavioral with probability one on the no-concession path.

\bigskip

\paragraph{Proof roadmap.}
The proof of Proposition~\ref{prop:full} proceeds through lemmas that
progressively pin down feasible posterior dynamics:
(i) all players exhaust their concession mass at a common finite time $T$;
(ii) on any interval strictly before $T$, at most one player can be inactive (have constant $F_i$);
(iii) whenever $F_i$ is strictly increasing on an interval, player $i$'s local indifference holds a.e.\
on that interval;
(iv) once a player starts conceding, they do not stop before $T$;
(v) a collection of posterior-ordering lemmas rules out all strict orderings of $(z_A,z_B,z_C)$
after all three players are active, forcing $z_A=z_B=z_C$ on that region; and
(vi) the time-$0$ atom(s) are pinned down by requiring posterior alignment at the moment
the last player becomes active.

\subsection{Proposition~\ref{prop:full} and its proof}\label{Appendix: Proposition of unique equilibrium}

We now state and prove a complete characterization of equilibrium behavior in the no-concession
subgame, summarized in Proposition~\ref{prop:full}. For expositional clarity,
we maintain the labeling convention from the main text and assume $z_A(0)\ge z_B(0)$.

A key structural feature of the equilibrium is that whenever the most reputable peripheral initially
strictly dominates the other two players (i.e.\ $z_A(0)>\max\{z_B(0),z_C(0)\}$), that peripheral optimally remains \emph{inactive}
at the outset. Intuitively, as long as $A$ is more reputable than $B$ and $C$, $A$ can wait without jeopardizing
his bargaining position, while the other two players must adjust: because any concession by $C$ resolves
\emph{both} negotiations, the local incentives of the remaining active peripheral are distorted and he must
concede at a higher hazard rate than in a stand-alone AG interaction.

Formally, in this region equilibrium begins with an \emph{initial two-player phase} on an interval $(0,t_A)$
during which $A$ does not concede, while $B$ and $C$ concede with constant hazard rates
\[
\rate_B=(1+\pi_{AC})\rate^{AG}\quad\text{and}\quad \rate_C=\rate^{AG},
\qquad\text{where}\qquad
\rate^{AG}=\frac{r(1-\alpha)}{2\alpha-1}.
\]
The phase ends at the \emph{activation time} $t_A$, defined as the first time at which the last inactive player
($A$) becomes willing to concede with positive density. At that moment all three players must be jointly
active; by the posterior-ordering lemmas proved below, joint activity forces \emph{posterior alignment}.
Hence the identity and size of any time-$0$ atom(s) are pinned down by the requirement that the posteriors
of the two active players reach $z_A(0)$ \emph{simultaneously} at the activation time:
\[
z_B(t_A)=z_C(t_A)=z_A(0).
\]
Equivalently, the time-$0$ atom is chosen so that, under the initial-phase hazard rates, $B$ and $C$ ``catch up''
to $A$'s prior reputation at the same calendar time.

To express this condition, define the ``catch-up times absent atoms''
\[
\tilde{t}_C:=\frac{1}{\rate^{AG}}\log\;\;\!\Bigl(\frac{z_A(0)}{z_C(0)}\Bigr),
\qquad
\tilde{t}_B:=\frac{1}{(1+\pi_{AC})\rate^{AG}}\log\;\;\!\Bigl(\frac{z_A(0)}{z_B(0)}\Bigr).
\]
Comparing $\tilde{t}_C$ and $\tilde{t}_B$ determines which player (if any) must place an atom at $t=0$ to ensure
posterior alignment at activation: if $\tilde{t}_C\ge \tilde{t}_B$, $C$ requires a (possibly zero) time-$0$ atom;
if $\tilde{t}_B>\tilde{t}_C$, $B$ requires a time-$0$ atom; and if $\tilde{t}_B=\tilde{t}_C$, no atom is required.
After $t_A$, all three players concede with the common constant hazard $\rate^{AG}$ until the common
terminal time at which posteriors reach $1$

Proposition~\ref{prop:full} below presents a more complete version of Proposition~\ref{prop:structure} that describes the equilibrium dynamics along with initial atoms.

\medskip
\renewcommand{\theHproposition}{A.\arabic{proposition}}
\setcounter{proposition}{0}
\renewcommand{\theproposition}{A.\arabic{proposition}}
\begin{proposition}[Full version of Proposition 1]\label{prop:full}
There is a unique equilibrium.
Assume, without loss of generality, $z_A(0)\ge z_B(0)$. The equilibrium concession behavior on the no-concession path is:

\begin{enumerate}
\item There exists a finite time $T<\infty$ such that $T_i=T$ for all $i\in\{A,B,C\}$.
\item Once all three players concede with positive density, their hazards coincide and are constant:
\[
\rate_A(t)=\rate_B(t)=\rate_C(t)=\rate^{AG}\quad\text{for a.e.\ such }t.
\]

\item If $z_C(0)\ge z_A(0)$ (the center is initially at least as reputable as both peripherals), then $F_C(0)=0$ and the peripherals place time-$0$ mass so that posteriors align immediately:
\[
F_i(0)=\max\left\{1-\frac{z_i(0)}{z_C(0)},\,0\right\}\quad (i\in\{A,B\}),
\qquad\text{hence}\qquad
z_A(0+)=z_B(0+)=z_C(0).
\]
All three players start conceding at time $0$ with hazard rate $\rate^{AG}$.

\item If $z_A(0)=z_B(0)>z_C(0)$, then $F_C(0)=1-\frac{z_C(0)}{z_A(0)}$ and $F_A(0)=F_B(0)=0$, so that
$z_C(0+)=z_A(0)=z_B(0)$. All three players then concede from time $0$ with hazard rate $\rate^{AG}$.

\item If $z_A(0)=z_B(0)=z_C(0)$,
then $F_A(0)=F_B(0)=F_C(0)=0$ and all three players concede from time $0$ with hazard rate $\rate^{AG}$.

\item If $z_A(0)>z_B(0)\ge z_C(0)$,
then $F_A(0)=F_B(0)=0$ and $C$ concedes at $t=0$ with probability
\[
F_C(0)=1-\frac{z_C(0)}{z_A(0)^{\frac{\pi_{AC}}{1+\pi_{AC}}}\,z_B(0)^{\frac{1}{1+\pi_{AC}}}}.
\]
Moreover, there is an initial phase $(0,t_A)$ in which $A$ is inactive and $(B,C)$ are active with
hazard rates $(1+\pi_{AC})\rate^{AG}$ and $\rate^{AG}$, respectively, until their posteriors reach $z_A(0)$ at $t_A=\tilde{t}_B$;
afterwards all three concede with hazard $\rate^{AG}$.

\item If $z_A(0)>z_C(0)>z_B(0)$, then $A$ is initially inactive and $(B,C)$ are active in an initial phase as in item~6.
At most one of $B$ and $C$ has a time-$0$ atom, determined by comparing $\tilde{t}_C$ and $\tilde{t}_B$:
\begin{enumerate}[(a)]
\item If $\tilde{t}_C>\tilde{t}_B$, then $C$ has the atom and $B$ does not:
\[
F_C(0)=1-\frac{z_C(0)}{z_A(0)^{\frac{\pi_{AC}}{1+\pi_{AC}}}\,z_B(0)^{\frac{1}{1+\pi_{AC}}}},
\qquad
F_A(0)=F_B(0)=0.
\]
\item If $\tilde{t}_B>\tilde{t}_C$, then $B$ has the atom and $C$ does not:
\[
\cdf_B(0) = 1-\frac{z_B(0)\,z_A(0)^{\pi_{AC}}}{z_C(0)^{1+\pi_{AC}}},
\qquad
F_A(0)=F_C(0)=0.
\]
\end{enumerate}
In either subcase, posteriors align when $A$ becomes active at $t_A=\min\{\tilde{t}_B,\tilde{t}_C\}$, and thereafter all three concede with hazard rate $\rate^{AG}$.
\end{enumerate}
\end{proposition}

\blemma\label{Lemma: All WoAs end at the same time} If $\cdf$ is an equilibrium, then $T_i = T < \infty$ for all $i\in\players$.\elemma

\bprf 
The proof proceeds in three steps. We first show that if $T_i<\infty$ for some $i\in\{A,B,C\}$, then $T_j<\infty$, for all $j\neq i$.  We then show that if $T_i<\infty$ for all $i$, then $T_i=T$ for all $i$. Finally, we show that $T_i<\infty$ for all $i$. \bigskip

Suppose first $T_C<\infty$. Then, for any $t>T_C$, $A$ and $B$ strictly prefer to concede immediately. Hence, \begin{equation}T_C\geq\max\{T_A,T_B\}.\label{eq:TC_geq}\end{equation} Suppose next $T_i<\infty$ for some $i\in\{A,B\}$. If $i$ conceded by $T_i$, the continuation reduces to the bilateral AG interaction between $j$ and $C$, and hence, $T_j=T_C<\infty$. By definition of $T_i$, $i$ never concedes after $T_i$ on the ``no-concession path.'' Then at any time $t>T_i$, the continuation reduces to the bilateral AG interaction between $j$ and $C$, but with the feature that a concession by $C$ terminates both negotiations. Hence, local indifference conditions in that continuation therefore imply that for any $t\in(T_A,T]$, the players' instantaneous rates of concession are given by:
\[
\rate_j(t)=
\begin{cases}
(1+\pi_{AC})\,\rate^{AG} & \text{if } j=B, \\[6pt]
\dfrac{1+\pi_{AC}}{\pi_{AC}}\,\rate^{AG} & \text{if } j=A,
\end{cases}
\;\;
\text{and}\;\;\rate_C(t)=\rate^{AG}.\]
Hence, $T_j=T_C<\infty.$

We next show that if $T_i<\infty$ for all $i$, terminal times must coincide. Suppose, toward a contradiction, that $T_i\neq T_j$ for some $i\neq j$. By \eqref{eq:TC_geq},
the only possible strict inequality is that one peripheral reaches $1$ earlier than the other two.
Thus, without loss, assume
\[
T_A < T_B = T_C=:T.
\]
Fix any $\Delta>0$ sufficiently small such that $T_A+\Delta<T$. By definition of $T_A$ as the \emph{first}
time $z_A(t)=1$, player $A$ concedes with positive probability on every interval $(t,T_A]$ with $t<T_A$.
In particular, $A$ must be willing to concede arbitrarily close to $T_A$.

We now show that, conditional on reaching $T_A$ with no concession, $A$ strictly prefers to wait an additional
$\Delta$ rather than concede at $T_A$, contradicting the preceding paragraph.

Conditional on no concession up to $T_A$, player $A$'s concession at $T_A$ yields payoff $(1-\alpha)\pi_{AC}$. Consider instead the deviation: do not concede on $(T_A,T_A+\Delta]$ and
(if still no concession has occurred by $T_A+\Delta$) concede at $T_A+\Delta$. On $(T_A,T)$, player $A$ is believed
behavioral with probability one (since $z_A=1$), hence $A$ never concedes on-path after $T_A$; the continuation
reduces to the two-player interaction between $B$ and $C$, but with the key feature that a concession by $C$
terminates \emph{both} negotiations. The local indifference conditions in that continuation therefore imply
that on $(T_A,T)$ the hazard rates are constant and satisfy
\[
\rate_C(t)=\rate^{AG}\quad\text{and}\quad \rate_B(t)=(1+\pi_{AC})\rate^{AG}
\qquad\text{for a.e.\ }t\in(T_A,T).
\]
Moreover, since $z_A(T_A)=1$, the AG continuation following a concession by $B$ at time $t\in(T_A,T)$ features an
immediate concession by $C$ to $A$ with strictly positive probability $g_C^A(t)>0$. By continuity of posteriors
(and of the AG atom as a function of posteriors), there exists $\underline g>0$ such that
\[
g_C^A(t)\ge \underline g\quad\text{for all }t\in[T_A,T_A+\Delta].
\]

Let $\mathcal E$ be the event that within $(T_A,T_A+\Delta]$ either (a) $C$ concedes, or (b) $B$ concedes and, in
the induced continuation, $C$ concedes to $A$ immediately with probability at least $\underline g$.
Under the hazard rate representation, conditional on no concession up to $T_A$, the probability that $C$ concedes in
$(T_A,T_A+\Delta]$ is
\[
\int_{T_A}^{T_A+\Delta}\rate_B(s,T_A)\rate_C(s,T_A)\,\rate_C(s)\,ds,
\]
and the probability that $B$ concedes in $(T_A,T_A+\Delta]$ is
\[
\int_{T_A}^{T_A+\Delta}\rate_B(s,T_A)\rate_C(s,T_A)\,\rate_B(s)\,ds.
\]
Using $\rate_i(s,T_A)\ge 1-\int_{T_A}^s\rate_i(u)\,du$ and constancy of $\rate_B,\rate_C$ on this interval,
we obtain the lower bound
\[
\Pr(\mathcal E)\ \ge\ \Delta\big(\rate_C + \underline g\,\rate_B\big) - o(\Delta),
\]
where $o(\Delta)/\Delta\to 0$ as $\Delta\downarrow 0$.

Under the deviation, conditional on reaching $T_A$,
\begin{itemize}
\item on $\mathcal E$, player $A$ is conceded to and obtains payoff at least $\alpha\pi_{AC}$;
\item on $\mathcal E^c$, $A$ concedes at $T_A+\Delta$ and obtains $(1-\alpha)\pi_{AC}$.
\end{itemize}
Discounting only \emph{reduces} payoffs under the deviation relative to evaluating them at $T_A$, so a lower bound
on the deviation payoff is obtained by discounting \emph{all} payoffs by $e^{-r\Delta}$.
Hence, conditional on reaching $T_A$, the deviation payoff is at least
\[
e^{-r\Delta}\pi_{AC}\Big((1-\alpha)\Pr(\mathcal E^c)+\alpha\Pr(\mathcal E)\Big)
= e^{-r\Delta}\pi_{AC}\Big((1-\alpha)+(2\alpha-1)\Pr(\mathcal E)\Big).
\]
Subtracting the immediate-concession payoff $(1-\alpha)\pi_{AC}$ and using $e^{-r\Delta}\ge 1-r\Delta$ yields
\[
\text{gain from waiting }\ge \pi_{AC}\Big((2\alpha-1)\Pr(\mathcal E) - r\Delta(1-\alpha)\Big).
\]
Using the lower bound $\Pr(\mathcal E)\ge \Delta(\rate_C+\underline g\,\rate_B) - o(\Delta)$ and
$\rate_C=\rate^{AG}$ gives
\[
(2\alpha-1)\Pr(\mathcal E) - r\Delta(1-\alpha)
\ \ge\ \Delta(2\alpha-1)\underline g\,\rate_B - o(\Delta)
\ =\ \Delta\underline g(2\alpha-1)(1+\pi_{AC})\rate^{AG} - o(\Delta).
\]
For $\Delta>0$ sufficiently small, the right-hand side is strictly positive. Hence $A$ strictly prefers to wait,
contradicting that $A$ concedes with positive probability arbitrarily close to $T_A$. Therefore, if at least one terminal time is finite, the terminal times
must coincide.

Suppose finally that $T_A=T_B=T_C=\infty$. Fix $\varepsilon\in(0,z_C(0))$ and define
\[
t_2:=\inf\{t\ge 0:\ F_C(t)\ge 1-z_C(0)-2\varepsilon\},\qquad
t_1:=\inf\{t\ge 0:\ F_C(t)\ge 1-z_C(0)-\varepsilon\}.
\]
Continuity of $F_C$ implies $t_2<t_1<\infty$. By construction,
\[
F_C(t_1)-F_C(t_2)\le \varepsilon.
\]
Consider player $A$'s continuation problem at time $t_2$ conditional on no concession up to $t_2$.
If $A$ concedes immediately at $t_2$, he obtains $(1-\alpha)\pi_{AC}$ (evaluated at $t_2$).
If instead $A$ commits to wait until $t_1$ and then (if necessary) concede, then the most favorable outcomes for $A$
over $(t_2,t_1]$ arise from being conceded to by $C$; this can happen with (conditional) probability at most
$\varepsilon$ by the preceding inequality. Therefore $A$'s continuation payoff from waiting until $t_1$ is bounded above by
\[
\varepsilon\alpha\pi_{AC} + e^{-r(t_1-t_2)}\Big(\varepsilon\alpha\pi_{AC}+(1-\varepsilon)(1-\alpha)\pi_{AC}\Big).
\]
For $\varepsilon>0$ sufficiently small, this upper bound is strictly less than $(1-\alpha)\pi_{AC}$.
Thus $A$ strictly prefers conceding at $t_2$ to waiting until $t_1$, contradicting equilibrium.
Hence $T<\infty$.
\eprf

\blemma\label{Lemma: at most one CDF constant} Let $0<t'<t<T$. If $F_i(t)=F_i(t')$ for some $i\in N$, then for each $j\in N\setminus\{i\}$,
\[
F_j(t)>F_j(t').
\]
Equivalently, on any interval strictly before $T$, at most one player can be inactive. \elemma

\bprf

Fix $0<t'<t<T$ and suppose $F_i(t)=F_i(t')$. 

Suppose there exists $j\neq i$ with $F_j(t)=F_j(t')$ and suppose the remaining player $k\in N\setminus\{i,j\}$
assigns positive concession probability on $(t',t]$. Consider $k$'s strategy restricted to histories in which no concession occurs before $t'$.
Holding fixed all play outside $(t',t]$, define a deviation for player $k$ that shifts \emph{all} concession probability
that $k$ assigns to times in $(t',t]$ to an atom at time $t'$. Because players $i$ and $j$ are inactive on $(t',t]$,
this deviation does not alter the probability distribution over \emph{who concedes first} conditional on reaching $t'$;
it only (weakly) reduces the calendar time of agreement. Since payoffs are discounted in calendar time and the
division conditional on who concedes is unchanged, the deviation strictly increases $k$'s expected payoff whenever
$k$ concedes with positive probability on $(t',t]$. This contradicts optimality.

If all three $F_A,F_B,F_C$ were constant on $[t',t]$, then since $t<T$ there exists some $t''>t$ at which some player
concedes with positive probability. The same time-shifting argument (moving that concession probability to an atom
at $t'$) yields a strict improvement, again contradicting equilibrium.

Therefore, if $F_i$ is constant on $[t',t]$, both remaining players must strictly increase on that interval.

\eprf

\blemma\label{Lemma: IC hold if CDFs increasing}Suppose that $\cdf_i\cd$ is strictly increasing on $[t_1,t_2]$. Then, the indifference condition for player $i$ holds at a.e. $t \in (t_1,t_2)$. \elemma
\bprf

\noindent\textbf{Case 1: $i=A$.}

Define 
\begin{align*}\begin{split}
&S_A(s;t_1)
:= \rate_B(s,t_1)\,\rate_C(s,t_1)
 = \exp\!\left(-\int_{t_1}^{s}\big(\rate_B(u)+\rate_C(u)\big)\,du\right),\\
&R_A(s)
:= \rate_C(s)\,\alpha
   + \rate_B(s)\Big(g_C^A(s)\alpha + (1-g_C^A(s))(1-\alpha)\Big),
   \end{split}
\end{align*}

where we can interpret $S_A(s;t_1)$ as the joint survival from $t_1$ to $s$ (i.e., no concession by $B$ or $C$ on $(t_1,s]$ and $R_A(s)$ as the instantaneous expected “reward rate” (in units of $\pi_{AC}$) from the first event at time $s$. Fix $t_1$ and define $u_A(t)$ as $A$'s continuation payoff at time $t_1$ from conceding at
exactly time $t\in[t_1,t_2]$, conditional on no concession strictly before $t$:

\begin{align*}
\frac{u_A(t)}{\pi_{AC}}
&= \int_{t_1}^{t} \exp\!\big(-r(s-t_1)\big)\, S_A(s;t_1)\, R_A(s)\, ds
   \;+\; \exp\!\big(-r(t-t_1)\big)\, S_A(t;t_1)\,(1-\alpha),
\qquad t\in[t_1,t_2].
\end{align*}

By piecewise continuity of $\rate_B(s),\rate_C(s)$ and boundedness of the integrands, $u_A(\cdot)$ is continuous on
$[t_1,t_2]$.

Let $M:=\arg\max_{w\in[t_1,t_2]}u_A(w)$ be the set of maximizers. If $A$ assigns positive concession density at some
$t\in(t_1,t_2)$, then $t\in M$ by optimality. Since $F_A$ is strictly increasing on $[t_1,t_2]$, its support is dense
in $(t_1,t_2)$, hence $M$ is dense in $(t_1,t_2)$. A continuous function that attains its maximum on a dense set
must be constant; therefore $u_A(t)$ is constant on $[t_1,t_2]$.

Since $u_A$ is absolutely continuous (indeed differentiable a.e.) with derivative obtained by the fundamental theorem
of calculus, differentiating $u_A(t)$ a.e.\ on $(t_1,t_2)$ and setting $u_A'(t)=0$ yields \eqref{Equation:IC-A} a.e.\ on $(t_1,t_2)$.

$(t_1,t_2)$.

\noindent\textbf{Case 2: $i=C$.}
Fix $t_1$. For each $t\in[t_1,t_2]$, let $u_C(t)$ denote $C$'s continuation payoff
\emph{evaluated at time $t_1$ and conditional on no concession up to $t_1$},
when $C$ follows the strategy: “concede at time $t$ if no player concedes before $t$.”

Define similarly

\begin{align*}
S_C(s;t_1)
&:= \rate_A(s,t_1)\,\rate_B(s,t_1)
 = \exp\!\left(-\int_{t_1}^{s}\big(\rate_A(u)+\rate_B(u)\big)\,du\right),\\
V_C^{A}(s)
&:= \alpha\,\pi_{AC}
   \;+\; \Big(g_B(s)\alpha + (1-g_B(s))(1-\alpha)\Big),\\
V_C^{B}(s)
&:= \alpha
   \;+\; \pi_{AC}\Big(g_A(s)\alpha + (1-g_A(s))(1-\alpha)\Big),\\
R_C(s)
&:= \rate_A(s)\,V_C^{A}(s) + \rate_B(s)\,V_C^{B}(s).
\end{align*}

 We can then write $C$'s continuation payoff from following this strategy (to concede at $t$ if nobody else has by then) as:
\begin{align*}
u_C(t)
&= \int_{t_1}^{t} e^{-r(s-t_1)}\, S_C(s;t_1)\, R_C(s)\, ds
   \;+\; e^{-r(t-t_1)}\, S_C(t;t_1)\,(1+\pi_{AC})(1-\alpha),
\qquad t\in[t_1,t_2].
\end{align*}

The same density-of-support argument implies $u_C(\cdot)$ is constant on $[t_1,t_2]$, hence $u_C'(t)=0$ a.e.\ and
therefore \eqref{Equation:IC-C} holds a.e.\ on $(t_1,t_2)$.

\eprf 

\blemma\label{Lemma: C doesn't stop conceding} If $\cdf_C(t) > 0$, then $\cdf_C\cd$ is strictly increasing on $[t,T]$. Therefore, \eqref{Equation:IC-C} holds for a.e. $s \in [t, T]$. 
\elemma

\bprf
 Assume $F_C(t)>0$ for some $t<T$ and suppose, toward a contradiction, that $F_C$ is not strictly increasing on $[t,T]$.
Then there exist $t<t'<T$ such that $F_C(t)=F_C(t')>0$.

Let
\[
\tau:=\inf\{s\le t:\ F_C(s)=F_C(t')\}
\]
be the left endpoint of the first plateau at level $F_C(t')$. Suppose for now $\tau>0$. By construction, $F_C$ is constant on $[\tau,t']$ and
strictly increasing on $(\tau-\varepsilon,\tau)$ for every sufficiently small $\varepsilon>0$ (by continuity).

\medskip
\noindent\textbf{Step 1: on $(\tau,t')$, both peripherals are active and satisfy their ICs.}
Because $F_C$ is constant on $[\tau,t']$, Lemma~\ref{Lemma: at most one CDF constant} implies that both $F_A$ and $F_B$ are
strictly increasing on $[\tau,t']$. Hence, by Lemma~\ref{Lemma: IC hold if CDFs increasing}, \eqref{Equation:IC-A} and \eqref{Equation:IC-B}
hold for a.e.\ $s\in(\tau,t')$. Since $\rate_C(s)=0$ a.e.\ on $(\tau,t')$, these equalities imply
\[
\rate_B(s)g_C^A(s)=\rate^{AG}\quad\text{and}\quad \rate_A(s)g_C^B(s)=\rate^{AG}
\qquad\text{for a.e.\ }s\in(\tau,t').
\]
In particular, $g_C^A(s),g_C^B(s)>0$, and therefore, $\rate_A(s),\rate_B(s)> \rate^{AG}$ a.e.\ on $(\tau,t')$.

\medskip
\noindent\textbf{Step 2: $C$ strictly prefers not to concede at $\tau$, contradicting that $C$ concedes before $\tau$.}

Let $V_C(\tau)$ denote $C$'s equilibrium continuation value at time $\tau$
conditional on no concession up to $\tau$.
Note that if $C$ concedes immediately at $\tau$, her payoff equals
\[
u_C(\tau,\tau)=(1+\pi_{AC})(1-\alpha).
\]
We show that $V_C(\tau)>(1+\pi_{AC})(1-\alpha)$, implying that $C$ strictly prefers to wait at $\tau$.

Consider waiting for a small $\varepsilon>0$. Over $(\tau,\tau+\varepsilon]$, either $A$ concedes, or $B$ concedes,
or neither concedes. Using only the \emph{minimal} payoff $C$ can guarantee in each event gives a lower bound:
\begin{itemize}
\item if $A$ concedes in $(\tau,\tau+\varepsilon]$, then $C$ obtains at least $\alpha\pi_{AC}$ in the negotiation between $A$ and $C$ and at least
$(1-\alpha)$ in the negotiation between $B$ and $C$;
\item if $B$ concedes in $(\tau,\tau+\varepsilon]$, then $C$ obtains at least $\alpha$ in the negotiation between $B$ and $C$ and at least
$(1-\alpha)\pi_{AC}$ in the negotiation between $A$ and $C$;
\item if no concession occurs in $(\tau,\tau+\varepsilon]$, then at $\tau+\varepsilon$ player $C$ can still secure
$(1+\pi_{AC})(1-\alpha)$ by conceding immediately.
\end{itemize}
Discounting over $\varepsilon$ satisfies $e^{-r\varepsilon}\ge 1-r\varepsilon$. Therefore,
\begin{align*}
V_C(\tau)
&\ge \Big(\int_{\tau}^{\tau+\varepsilon}\rate_A(u)\,du\Big)\,(\alpha\pi_{AC}+1-\alpha)
 + \Big(\int_{\tau}^{\tau+\varepsilon}\rate_B(u)\,du\Big)\,(\alpha+\pi_{AC}(1-\alpha))\\
&\quad + (1-r\varepsilon)\Big(1-\int_{\tau}^{\tau+\varepsilon}\rate_A(u)\,du-\int_{\tau}^{\tau+\varepsilon}\rate_B(u)\,du\Big)
(1+\pi_{AC})(1-\alpha).
\end{align*}
Since $\rate_A(\cd),\rate_B(\cd) > \rate^{AG}$ for a.e.\ $s \in (\tau,\tau+\varepsilon]$, expanding the right-hand side and writing $P := (1+\pi_{AC})(1-\alpha)$ gives
\begin{align*}
V_C(\tau) &\geq P + (2\alpha-1)\!\left(\pi_{AC}\!\int_\tau^{\tau+\varepsilon}\!\rate_A(u)\,du + \int_\tau^{\tau+\varepsilon}\!\rate_B(u)\,du\right) - r\varepsilon\, P + O(\varepsilon^2)\\
&> P + (2\alpha-1)(1+\pi_{AC})\rate^{AG}\varepsilon - r\varepsilon\, P + O(\varepsilon^2) = P + O(\varepsilon^2),
\end{align*}
where the last equality uses $(2\alpha-1)\rate^{AG} = r(1-\alpha)$. For $\varepsilon>0$ sufficiently small, $V_C(\tau) > P = (1+\pi_{AC})(1-\alpha)$.

But $(1+\pi_{AC})(1-\alpha)$ is precisely the payoff $C$ obtains by conceding at $\tau$.
Therefore, $C$ strictly prefers to wait rather than concede at $\tau$. By continuity of payoffs in time,
$C$ also strictly prefers to wait rather than concede at any time in a left-neighborhood of $\tau$, contradicting that
$F_C$ is strictly increasing just before $\tau$.

This contradiction shows that no plateau can occur after $F_C$ becomes positive. Hence $F_C$ is strictly increasing on
$[t,T]$. 

Suppose $\tau=0$. Then $F_C(0)=F_C(t')>0$ and $F_C$ is constant on $(0,t']$, so $\rate_C(t)=0$ for a.e.\ $t\in(0,t']$.
Fix $\varepsilon\in(0,t']$ and set $t_0=\varepsilon/2$.
Since $F_C$ is constant on $[t_0,\varepsilon]$, Lemma~\ref{Lemma: at most one CDF constant} implies
$F_A(\varepsilon)>F_A(t_0)$ and $F_B(\varepsilon)>F_B(t_0)$, hence $\rate_A,\rate_B>0$ a.e.\ on $(0,\varepsilon)$.
By Lemma~\ref{Lemma: IC hold if CDFs increasing}, \eqref{Equation:IC-A} and \eqref{Equation:IC-B} hold a.e.\ on $(t_0,\varepsilon)$; with $\rate_C=0$ this implies
$\rate_B g_C^A=\rate_A g_C^B=\rate^{AG}$ and therefore $\rate_A,\rate_B>\rate^{AG}$ a.e.\ on $(t_0,\varepsilon)$.

Consider $C$'s deviation at time $0$: wait until time $\varepsilon$ and (if still unresolved) concede at $\varepsilon$.
Let $S_C(s;0):=\rate_A(s,0)\rate_B(s,0)$. The deviation payoff is bounded below by
\[
\int_0^\varepsilon e^{-rs}S_C(s;0)\Big[
\rate_A(s)(\alpha\pi_{AC}+1-\alpha)+\rate_B(s)(\alpha+\pi_{AC}(1-\alpha))
\Big]ds
+e^{-r\varepsilon}S_C(\varepsilon;0)(1+\pi_{AC})(1-\alpha),
\]
which exceeds $(1+\pi_{AC})(1-\alpha)$ for $\varepsilon>0$ sufficiently small since $\rate_A,\rate_B>\rate^{AG}$.
Thus $C$ strictly prefers waiting to conceding at time $0$, contradicting $F_C(0)>0$.

The final claim follows from Lemma~\ref{Lemma: IC hold if CDFs increasing}.

\eprf

\blemma\label{Lemma: A doesn't stop conceding once above C} 
Fix $i\in\{A,B\}$. If for some $t<T$,
\[
z_i(t)>z_i(0)\quad\text{and}\quad z_i(t)>z_C(t),
\]
then for every $t'>t$ such that $z_i(t')\ge z_C(t')$, we have $F_i(t')>F_i(t)$.
Equivalently, $\cdf_i$ cannot be constant on any interval $[t,t']\subset(0,T)$ with $z_i(t')\ge z_C(t')$.
\elemma

\bprf 

We prove the claim for $i=A$ (the case $i=B$ is symmetric). Assume $z_A(t)>z_A(0)$ and $z_A(t)>z_C(t)$.

 Suppose, toward a contradiction, that there exists $t'>t$ such that

\[
F_A(t')=F_A(t)
\qquad\text{and}\qquad
z_A(t')\ge z_C(t').
\]
Since $F_A$ is constant on $[t,t']$, we have $z_A(s)=z_A(t)$ for all $s\in[t,t']$.
Since $z_C(\cdot)$ is weakly increasing, $z_C(s)\le z_C(t')$ for all $s\in[t,t']$.
Therefore,
\begin{equation}
z_A(s)=z_A(t')\ge z_C(t')\ge z_C(s)
\qquad\text{for all }s\in[t,t'].
\tag{$\dagger$}\label{Equation: dagger}
\end{equation}

Let
\[
\tau:=\inf\{s\le t:\ F_A(s)=F_A(t)\}.
\]
Then $F_A$ is constant on $[\tau,t']$. Moreover, since $z_A(t)>z_A(0)$, we have $F_A(t)>0$ and (by continuity at
positive times) $\tau\in(0,t]$. Hence $F_A$ is not constant on any left-neighborhood of $\tau$.

Because $F_A$ is constant on $[\tau,t']$, Lemma~\ref{Lemma: at most one CDF constant} implies that both $F_B$ and $F_C$ are strictly increasing on
$[\tau,t']$. Hence $\rate_B,\rate_C>0$ a.e.\ on $(\tau,t')$, and by Lemma~\ref{Lemma: IC hold if CDFs increasing}, \eqref{Equation:IC-B} and
\eqref{Equation:IC-C} hold a.e.\ on $(\tau,t')$. Since $F_A$ is constant on $(\tau,t')$, we have $\rate_A=0$
a.e.\ on $(\tau,t')$, so \eqref{Equation:IC-B} yields
\begin{equation}
\rate_C(u)=\rate^{AG}\qquad\text{for a.e.\ }u\in(\tau,t').
\tag{$\ddagger$}\label{Equation: dagger2}
\end{equation}

Next, by ($\dagger$) applied on $[\tau,t]$, we have $z_A(u)>z_C(u)$ for all $u\in[\tau,t]$ (since $z_A(t)>z_C(t)$),
so in the AG continuation following $B$'s concession at time $u$ we have
\[
\atom_C^A(u)>0
\quad\text{and}\quad
g_A(u)=0
\qquad\text{for all }u\in[\tau,t].
\]
Using $\rate_A=0$ and $g_A=0$ in \eqref{Equation:IC-C} therefore implies
\[
\rate_B(u)=(1+\pi_{AC})\rate^{AG}
\qquad\text{for a.e.\ }u\in(\tau,t).
\tag{$\star$}
\]

By continuity of posteriors on $(0,T)$ and the fact that $g_C^A(u)=\max\{1-z_C(u)/z_A(u),0\}$ in the AG continuation,
there exist $\varepsilon>0$ and $\underline g>0$ such that $\tau+\varepsilon\le t$ and
\[
g_C^A(u)\ge \underline g\qquad\text{for all }u\in[\tau,\tau+\varepsilon].
\tag{$\star\star$}
\]

Now, let us first suppose that $\tau > 0$. Since $F_A$ is not constant on any left-neighborhood of $\tau$, there exists a sequence $s_n\uparrow\tau$
such that $\rate_A(s_n)>0$ (at points of continuity). Consider such an $s=s_n$ close enough to $\tau$.
Using ($\ddagger$), ($*$), and ($**$), the same $\varepsilon$-deviation argument as in Lemma~\ref{Lemma: All WoAs end at the same time}
implies that, conditional on reaching time $s$ with no concession, player $A$ strictly prefers to wait until
$s+\varepsilon$ (and concede then if necessary) rather than concede at time $s$:
over $(s,s+\varepsilon]$, the effective rate at which $A$ is conceded to is at least
$\rate_C+\rate_B\,\underline g>\rate^{AG}$, which yields a strictly positive first-order gain that dominates
the $O(\varepsilon)$ discounting loss. This contradicts $\rate_A(s)>0$.

Now suppose that $\tau = 0$. Therefore, $z_A(0+) > z_A(0)$. Thus, $\cdf_A(0) > 0$, i.e., there is a time $0$ atom from $A$. Following the previous argument, the effective rate at which $A$ is conceded to on $(0,\epsilon)$ is at least $\rate_C+\rate_B\, \underline g > \rate^{AG}$, which yields a strict incentive for $A$ to wait at $0$. Therefore, there cannot be an atom at $0$ by $A$ contradicting $z_A(t) > z_A(0)$. 

Therefore no such $t'>t$ can exist, and $F_A(t')>F_A(t)$ for all $t'>t$ with $z_A(t')\ge z_C(t')$.

\eprf

\blemma\label{Lemma: if zA > zB,zC then, A must not have conceded}
Let $i,j\in\{A,B\}$ with $i\neq j$. If for some $u>0$,
\[
z_i(u)>z_C(u)\ge z_j(u),
\]
then $z_i(u)=z_i(0)$.
\elemma

\bprf 
Without loss of generality, take $(i,j)=(A,B)$. Suppose $z_A(u)>z_C(u)\ge z_B(u)$ for some $u>0$.
Assume toward a contradiction that $z_A(u)>z_A(0)$.

Define
\[
\tau:=\sup\Big\{s\in[u,T):\ 1>z_A(t)>z_C(t)\ge z_B(t)\ \text{for all }t\in[u,s]\Big\}.
\]
By Lemma~\ref{Lemma: A doesn't stop conceding once above C}, $F_A$ is strictly increasing on $[u,\tau]$, hence
\eqref{Equation:IC-A} holds a.e.\ on $(u,\tau)$ by Lemma~\ref{Lemma: IC hold if CDFs increasing}.

We consider two cases for $\rate_B$ on $(u,\tau)$.

\medskip
\noindent\textbf{Case 1: $\rate_B>0$ on a set of positive measure in $(u,\tau)$.}
Since $\rate_B$ is piecewise continuous, there exists a point $s\in(u,\tau)$ at which $\rate_B$ is continuous
and $\rate_B(s)>0$, hence $\rate_B>0$ on some open interval $(s_1,s_2)\subset(u,\tau)$.
On $(s_1,s_2)$, $F_B$ is strictly increasing, so \eqref{Equation:IC-B} holds a.e.\ there.

On $(u,\tau)$ we have $z_B\le z_C$, hence the AG atom $g_C^B(\cdot)=0$ throughout, while $z_A>z_C$ implies
$g_C^A(\cdot)>0$ throughout. Therefore, on $(s_1,s_2)$, \eqref{Equation:IC-B} implies $\rate_C=\rate^{AG}$ a.e.,
whereas \eqref{Equation:IC-A} implies $\rate_C=\rate^{AG}-\rate_B g_C^A<\rate^{AG}$ a.e., a contradiction.

\medskip
\noindent\textbf{Case 2: $\rate_B=0$ a.e.\ on $(u,\tau)$.}
Then $F_B$ is constant on $(u,\tau)$, so Lemma~\ref{Lemma: at most one CDF constant} implies that both $F_A$ and $F_C$
are strictly increasing on $(u,\tau)$. Hence \eqref{Equation:IC-A} and \eqref{Equation:IC-C} hold a.e.\ on $(u,\tau)$.
Since $z_A>z_C$ on $(u,\tau)$, we have $g_A(\cdot)=0$ there, and thus \eqref{Equation:IC-C} together with $\rate_B=0$
implies $\rate_A=\frac{(1+\pi_{AC})}{\piac}\rate^{AG}$ a.e.\ on $(u,\tau)$, while \eqref{Equation:IC-A} implies
$\rate_C=\rate^{AG}$ a.e.\ on $(u,\tau)$.

Therefore $z_A$ grows strictly faster than $z_C$ on $(u,\tau)$, so the strict inequality $z_A>z_C$ persists at $\tau$
and (by continuity) beyond $\tau$ unless $z_A(\tau)=1$. If $z_A(\tau)=1$, then $T_A< T$ contradicting
Lemma~\ref{Lemma: All WoAs end at the same time}. If $z_A(\tau)<1$, then the strict inequality persists on a right-neighborhood of $\tau$,
contradicting the definition of $\tau$ as a supremum. Either way yields a contradiction.

Thus $z_A(u)>z_A(0)$ is impossible, so $z_A(u)=z_A(0)$.
\eprf

\blemma\label{Lemma: not possible to have za, zb>zc} For all $t>0$,
\[
\min\{z_A(t),z_B(t)\}\le z_C(t).
\]\elemma

\bprf 

Suppose, toward a contradiction, that $\min\{z_A(t),z_B(t)\}>z_C(t)$ for some $t>0$.
Define
\[
\tau:=\sup\Big\{s\ge t:\ 1>z_A(u),\ z_B(u)>z_C(u)\ \text{for all }u\in[t,s]\Big\}.
\]
On $[t,\tau]$, both peripherals are strictly more reputable than $C$, hence in the AG continuation after $C$ concedes
we have $g_A(u)=g_B(u)=0$ for all $u\in[t,\tau]$.

We claim $F_C$ cannot be constant on $(t,\tau)$. If $F_C$ were constant on $(t,\tau)$, then $z_C$ would be constant there,
while $z_A$ and $z_B$ are weakly increasing. In particular, the strict inequality $z_A,z_B>z_C$ would persist at $\tau$.
If $\max\{z_A(\tau),z_B(\tau)\}=1$, then Lemma~2 is violated. If instead
$1>z_A(\tau),z_B(\tau)>z_C(\tau)$, continuity implies the same strict inequality holds on a right-neighborhood of $\tau$,
contradicting the definition of $\tau$. Hence $F_C$ is not constant on $(t,\tau)$.

Therefore there exists $s\in(t,\tau)$ with $F_C(s)>F_C(t)$. By Lemma~\ref{Lemma: C doesn't stop conceding}, $F_C$ is strictly increasing on
$[s,T]$, so $\rate_C>0$ a.e.\ on $(s,\tau)$ and \eqref{Equation:IC-C} holds a.e.\ on $(s,\tau)$ by Lemma~\ref{Lemma: IC hold if CDFs increasing}.
Since $g_A=g_B=0$ on $(s,\tau)$, \eqref{Equation:IC-C} simplifies to
\[
\pi_{AC} \rate_A(u)+\rate_B(u)=(1+\pi_{AC})\rate^{AG}\qquad\text{for a.e.\ }u\in(s,\tau).
\]

By Lemma~\ref{Lemma: at most one CDF constant}, at least one of $F_A,F_B$ is strictly increasing on $(s,\tau)$.
We now show each possible activity pattern leads to a contradiction.

\medskip
\noindent\textbf{(i) Exactly one peripheral is active on $(s,\tau)$.}
Without loss, suppose $F_A$ is strictly increasing and $F_B$ is constant on $(s,\tau)$.
Then $\rate_B=0$ a.e.\ on $(s,\tau)$, so the simplified condition yields $\rate_A=(1+\pi_{AC})\rate^{AG}$ a.e.
Moreover, since $F_A$ is strictly increasing, \eqref{Equation:IC-A} holds a.e.\ on $(s,\tau)$, implying
$\rate_C=\rate^{AG}-\rate_B g_C^A=\rate^{AG}$ a.e.\ on $(s,\tau)$.
Hence $z_A$ grows strictly faster than $z_C$ on $(s,\tau)$, implying $z_A(\tau)>z_C(\tau)=z_B(\tau)$ and $z_A(\tau)>z_A(0)$.
This contradicts Lemma~\ref{Lemma: if zA > zB,zC then, A must not have conceded}.

\medskip
\noindent\textbf{(ii) Both peripherals are active on $(s,\tau)$.}
Then \eqref{Equation:IC-A} and \eqref{Equation:IC-B} hold a.e.\ on $(s,\tau)$, and because $z_A,z_B>z_C$ on $(s,\tau)$ we have
$g_C^A,g_C^B>0$ there. Hence \eqref{Equation:IC-A} and \eqref{Equation:IC-B} imply $\rate_C<\rate^{AG}$ a.e.\ on $(s,\tau)$,
while the simplified condition fixes $\pi_{AC}\rate_A+\rate_B=(1+\pi_{AC})\rate^{AG}$. Therefore at least one of
$\rate_A,\rate_B$ exceeds $\rate^{AG}$ on a set of positive measure, implying at least one of $z_A,z_B$ grows strictly
faster than $z_C$ on $(s,\tau)$. At $\tau$, maximality of $\tau$ forces $\min\{z_A(\tau),z_B(\tau)\}=z_C(\tau)$ while the
other peripheral strictly exceeds $z_C(\tau)$, and both peripherals have strictly increased from their priors on $(s,\tau)$.
This contradicts Lemma~\ref{Lemma: if zA > zB,zC then, A must not have conceded} (applied to the peripheral that strictly exceeds $z_C$ at $\tau$).

Both cases are impossible, so $\min\{z_A(t),z_B(t)\}>z_C(t)$ cannot occur.

\eprf

\blemma\label{Lemma: Not possible to have zC > zA,zB} For all $t>0$,
\[
z_C(t)\le \max\{z_A(t),z_B(t)\}.
\]
\elemma
\bprf
Suppose, toward a contradiction, that $z_C(t)>\max\{z_A(t),z_B(t)\}$ for some $t>0$.

\medskip
\noindent\textbf{Case 1: $z_C(t)=z_C(0)$.}
Then $\cdf_C$ is constant on $(0,t)$ and hence $\rate_C(s)=0$ for a.e.\ $s\in(0,t)$.
Since posteriors are weakly increasing in $s$, we have
\[
z_C(s)>\max\{z_A(s),z_B(s)\}\qquad\text{for all }s\in(0,t].
\]
By Lemma~\ref{Lemma: at most one CDF constant}, both $\cdf_A$ and $\cdf_B$ are strictly increasing on $(0,t)$.
Therefore, by Lemma~\ref{Lemma: IC hold if CDFs increasing}, \eqref{Equation:IC-A} and \eqref{Equation:IC-B}
hold for a.e.\ $s\in(0,t)$.

Fix any such $s\in(0,t)$. Since $z_C(s)>z_A(s)$ and $z_C(s)>z_B(s)$, player $C$ is the strong player in the induced
AG continuation following either peripheral's concession at time $s$. Hence
\[
\atom_C^A(s)=\atom_C^B(s)=0.
\]
Substituting into \eqref{Equation:IC-A} (or \eqref{Equation:IC-B}) yields $\rate_C(s)=\rate^{AG}$ for a.e.\ $s\in(0,t)$,
contradicting $\rate_C(s)=0$ there.

\medskip
\noindent\textbf{Case 2: $z_C(t)>z_C(0)$.}
Define
\[
\tau:=\sup\Big\{s\ge t:\ 1>z_C(u)>\max\{z_A(u),z_B(u)\}\ \text{for all }u\in[t,s]\Big\}.
\]
Then $\tau>t$ and, by continuity of posteriors on $(0,T)$,
\[
z_C(u)>\max\{z_A(u),z_B(u)\}\qquad\text{for all }u\in(t,\tau).
\]
Since $z_C(t)>z_C(0)$, we have $\cdf_C(t)>0$, and Lemma~\ref{Lemma: C doesn't stop conceding} implies that $\cdf_C$
is strictly increasing on $(t,\tau)$. Hence, by Lemma~\ref{Lemma: IC hold if CDFs increasing}, \eqref{Equation:IC-C}
holds for a.e.\ $u\in(t,\tau)$.

\smallskip
\noindent\emph{Step 1: at least one peripheral accumulates strictly less than $\rate^{AG}(\tau-t)$.}
Fix $u\in(t,\tau)$ such that \eqref{Equation:IC-C} holds. Since $z_C(u)>z_A(u)$ and $z_C(u)>z_B(u)$, in the AG continuation
induced by $C$'s concession at time $u$, both peripherals are weak against $C$, so
\[
\atom_A(u)>0\quad\text{and}\quad \atom_B(u)>0.
\]
Thus \eqref{Equation:IC-C} reads
\[
\rate_A(u)\big(\piac+\atom_B(u)\big)+\rate_B(u)\big(1+\piac\atom_A(u)\big)=(1+\piac)\rate^{AG},
\]
with $\piac+\atom_B(u)>\piac$ and $1+\piac\atom_A(u)>1$. Therefore, whenever $(\rate_A(u),\rate_B(u))\neq (0,0)$,
we have the strict inequality
\[
\piac\,\rate_A(u)+\rate_B(u) < (1+\piac)\rate^{AG}.
\]
By Lemma~\ref{Lemma: at most one CDF constant}, $\cdf_A$ and $\cdf_B$ cannot both be constant on any subinterval of $(t,\tau)$,
and since hazard rates are piecewise continuous, this implies $(\rate_A(u),\rate_B(u))\neq(0,0)$ for a.e.\ $u\in(t,\tau)$.
Integrating over $(t,\tau)$ yields
\[
\piac\int_t^\tau \rate_A(u)\,du+\int_t^\tau \rate_B(u)\,du < (1+\piac)\rate^{AG}(\tau-t),
\]
and hence
\[
\min\Big\{\int_t^\tau \rate_A(u)\,du,\ \int_t^\tau \rate_B(u)\,du\Big\}<\rate^{AG}(\tau-t).
\tag{$\dagger$}
\]

\smallskip
\noindent\emph{Step 2: $\rate_C=\rate^{AG}$ a.e.\ on $(t,\tau)$ and $z_C(\tau)>\min\{z_A(\tau),z_B(\tau)\}$.}
On $(t,\tau)$ we have $z_C>\max\{z_A,z_B\}$, so as in Case~1,
\[
\atom_C^A(u)=\atom_C^B(u)=0\qquad\text{for all }u\in(t,\tau).
\]
Therefore, whenever $\rate_A(u)>0$, \eqref{Equation:IC-A} implies $\rate_C(u)=\rate^{AG}$; whenever $\rate_B(u)>0$,
\eqref{Equation:IC-B} implies $\rate_C(u)=\rate^{AG}$. Again using Lemma~\ref{Lemma: at most one CDF constant} and
piecewise continuity, at least one of $\rate_A,\rate_B$ is positive a.e.\ on $(t,\tau)$, hence
\[
\rate_C(u)=\rate^{AG}\qquad\text{for a.e.\ }u\in(t,\tau).
\]
Thus $\int_t^\tau \rate_C(u)\,du=\rate^{AG}(\tau-t)$. Combining with $(\dagger)$ and the fact $z_C(t)>z_A(t),z_B(t)$
implies
\[
z_C(\tau)>\min\{z_A(\tau),z_B(\tau)\}.
\tag{$\ddagger$}
\]

\smallskip
\noindent\emph{Step 3: contradiction at $\tau$ and beyond.}
By definition of $\tau$ and continuity, either (i) $z_C(\tau)=1$, or (ii) $z_C(\tau)=\max\{z_A(\tau),z_B(\tau)\}$.
Case (i) is impossible because all posteriors reach $1$ only at the common terminal time $T$.
Hence we are in case (ii). Together with ($\ddagger$), this implies (without loss of generality)
\[
z_C(\tau)=z_A(\tau)>z_B(\tau).
\tag{$\star$}
\]
Since $z_C(u)>z_A(u)$ for all $u\in[t,\tau)$ while $z_C(\tau)=z_A(\tau)$, it follows that $z_A(\tau)>z_A(0)$.

We claim that
\[
z_C(u)\ge \max\{z_A(u),z_B(u)\}\qquad\text{for all }u\in[\tau,T).
\tag{$\star\star$}
\]
If $z_A(u)>z_C(u)$ for some $u\ge\tau$, then Lemma~\ref{Lemma: not possible to have za, zb>zc} implies $z_B(u)\le z_C(u)$,
so $z_A(u)>z_C(u)\ge z_B(u)$, and Lemma~\ref{Lemma: if zA > zB,zC then, A must not have conceded} forces $z_A(u)=z_A(0)$,
contradicting $z_A(u)\ge z_A(\tau)>z_A(0)$. If instead $z_B(u)>z_C(u)$ for some $u\ge\tau$, then similarly
$z_B(u)>z_C(u)\ge z_A(u)$ and Lemma~\ref{Lemma: if zA > zB,zC then, A must not have conceded} forces $z_B(u)=z_B(0)$,
but $z_B$ is nondecreasing and $z_B(\tau)<z_C(\tau)\le z_C(u)$, a contradiction. This proves ($**$).

On $(\tau,T)$ we therefore have $\atom_C^A=\atom_C^B=0$ and, by the same argument as in Step~2,
$\rate_C(u)=\rate^{AG}$ for a.e.\ $u\in(\tau,T)$. Hence $\int_\tau^T \rate_C(u)\,du=\rate^{AG}(T-\tau)$.

Since $z_A(\tau)=z_C(\tau)$ and $z_A(T)=z_C(T)=1$, we must have
\[
\int_\tau^T \rate_A(u)\,du=\int_\tau^T \rate_C(u)\,du=\rate^{AG}(T-\tau).
\]
Since $z_B(\tau)<z_C(\tau)$ but $z_B(T)=z_C(T)=1$, we must have
\[
\int_\tau^T \rate_B(u)\,du>\int_\tau^T \rate_C(u)\,du=\rate^{AG}(T-\tau).
\]
Finally, \eqref{Equation:IC-C} holds for a.e.\ $u\in(\tau,T)$ by Lemma~\ref{Lemma: C doesn't stop conceding}
and Lemma~\ref{Lemma: IC hold if CDFs increasing}. Integrating \eqref{Equation:IC-C} over $(\tau,T)$ yields
\[
\int_\tau^T\Big(\rate_A(u)(\piac+\atom_B(u))+\rate_B(u)(1+\piac\atom_A(u))\Big)\,du
=(1+\piac)\rate^{AG}(T-\tau).
\]
But $\atom_A,\atom_B\ge 0$, so the left-hand side is at least
\[
\piac\int_\tau^T \rate_A(u)\,du+\int_\tau^T \rate_B(u)\,du
>(1+\piac)\rate^{AG}(T-\tau),
\]
a contradiction. This completes the proof.
\eprf

\blemma\label{Lemma: Cannot have zA = zC > z_B}If $\cdf$ is an equilibrium, then $\nexists$ $t > 0$ such that $z_A(t) = z_C(t) > z_B(t)$. \elemma
\bprf 
Suppose, toward a contradiction, that $z_A(t)=z_C(t)>z_B(t)$ for some $t>0$.
Fix $t'>t$ sufficiently close to $t$.

\smallskip
\noindent\textbf{Step 1: none of $F_A,F_B,F_C$ can be constant on $(t,t')$.}

\begin{itemize}
\item[-] If $F_C$ were constant on $(t,t')$, then by Lemma~\ref{Lemma: at most one CDF constant}
both $F_A$ and $F_B$ would be strictly
increasing there, hence for some $s\in(t,t')$ we would have $z_A(s)>z_C(s)>z_B(s)$, contradicting Lemma~\ref{Lemma: if zA > zB,zC then, A must not have conceded}.
\item[-] If $F_A$ were constant on $(t,t')$, then $F_B$ and $F_C$ would be strictly increasing there, implying
$z_C(s)>\max\{z_A(s),z_B(s)\}$ for some $s\in(t,t')$, contradicting Lemma~\ref{Lemma: Not possible to have zC > zA,zB}.
\item[-] If $F_B$ were constant on $(t,t')$, then $F_A$ and $F_C$ would be strictly increasing there, so \eqref{Equation:IC-A} and
\eqref{Equation:IC-C} would hold a.e.\ on $(t,t')$. This forces $\rate_C=\rate^{AG}$ a.e.\ while $\rate_A>\rate^{AG}$
a.e., implying $z_A(s)>z_C(s)$ for some $s\in(t,t')$, contradicting Lemma~\ref{Lemma: if zA > zB,zC then, A must not have conceded}.
\end{itemize}

Hence all three $F_A,F_B,F_C$ are strictly increasing on $(t,t')$.

\smallskip
\noindent\textbf{Step 2: all three indifference conditions hold locally and imply a contradiction.}
Because all three players are active on $(t,t')$, Lemma~\ref{Lemma: IC hold if CDFs increasing} implies that
\eqref{Equation:IC-A}--\eqref{Equation:IC-C} hold a.e.\ on $(t,t')$.

Define
\[
\tau:=\sup\Big\{s\ge t:\ z_A(u)\ge z_C(u)>z_B(u)\ \text{for all }u\in[t,s]\Big\},
\]

By Lemma~\ref{Lemma: if zA > zB,zC then, A must not have conceded}, the strict inequality $z_A>z_C$ cannot occur at any time $u>t$
once $A$ is active, hence we must have $z_A(u)=z_C(u)$ for all $u\in(t,\tau)$.

Then $g_C^A(u)=g_C^B(u)$ on $(t,\tau)$, and \eqref{Equation:IC-A}--\eqref{Equation:IC-B} imply that $\rate_C=\rate^{AG}$
a.e.\ on $(t,\tau)$. Further, since $z_A(u) = z_C(u)$ for a.e. $u \in (t,\tau) \implies \rate_A(u) = \rate^{AG}$ for a.e. $u \in (t,\tau)$.  Plugging into \eqref{Equation:IC-C} yields $\rate_B<\rate^{AG}$ a.e.\ on $(t,\tau)$ since $\atom_B(\cd) > 0$ on this interval.
Therefore, for $u>t$ close enough to $t$, we have $z_A(u)=z_C(u)>z_B(u)$ and $z_B(u)$ is growing strictly more slowly,
which implies that $z_A$ and $z_C$ reach $1$ strictly before $z_B$ can, contradicting Lemma~\ref{Lemma: All WoAs end at the same time}.

\eprf

\blemma\label{Lemma: equilibrium has zA= zB=zC once they start conceding} Once all three players have conceded with positive probability, posteriors coincide:
if $z_i(t)>z_i(0)$ for all $i\in N$, then
\[
z_A(t)=z_B(t)=z_C(t).
\]
 \elemma

\bprf 
Fix $t>0$ such that $z_i(t)>z_i(0)$ for all $i\in N$. If the three posteriors are not equal at $t$, then (up to swapping
labels $A$ and $B$) one of the following orderings must hold:
\begin{align*}\begin{split}
&z_A(t)>z_C(t)\ge z_B(t),\\
&z_A(t)\ge z_B(t)>z_C(t),\\
&z_C(t)>\max\{z_A(t),z_B(t)\},\\
&z_A(t)=z_C(t)>z_B(t).
\end{split}
\end{align*}
These are ruled out by Lemmas~\ref{Lemma: if zA > zB,zC then, A must not have conceded} --~\ref{Lemma: Cannot have zA = zC > z_B}. Therefore $z_A(t)=z_B(t)=z_C(t)$.

\eprf 

\bprf[Proof of Proposition~\ref{prop:full}]

\medskip
\noindent\textbf{Step 1: Player $C$ is active immediately after $0$.}
We claim $t_C=0$.
Suppose instead that $t_C>0$, so that $F_C$ is constant on $(0,t_C)$ and hence
$\rate_C(t)=0$ for a.e.\ $t\in(0,t_C)$.
By Lemma~\ref{Lemma: at most one CDF constant}, at most one player's cdf can be constant on an interval
before $T$, so both $F_A$ and $F_B$ must be strictly increasing on $(0,t_C)$, i.e.\ $\rate_A(t),\rate_B(t)>0$
for a.e.\ $t\in(0,t_C)$.
Then Lemma~\ref{Lemma: IC hold if CDFs increasing} implies that the local indifference conditions for $A$ and $B$
hold a.e.\ on $(0,t_C)$.
Since $\rate_C=0$ there, those indifference conditions force $g_C^{A}(t)>0$ and $g_C^{B}(t)>0$ a.e.\ on $(0,t_C)$,
hence $z_A(t)>z_C(t)$ and $z_B(t)>z_C(t)$ on a set of positive measure.
This contradicts Lemma~\ref{Lemma: not possible to have za, zb>zc}.
Therefore $t_C=0$.

\medskip
\noindent\textbf{Step 2: Once all three players are active, posteriors (and hazard rates) coincide.}
Let $t^*:=\max\{t_A,t_B,t_C\}$ be the time at which the last player becomes active.
By Lemma~\ref{Lemma: equilibrium has zA= zB=zC once they start conceding}, we have
\[
z_A(t)=z_B(t)=z_C(t)\qquad\text{for all }t\in[t^*,T).
\]
Plugging this equality into the local indifference conditions, together with \[z_i(t)=z_i(t^*)\exp\left(\int_{t^*}^t \lambda_i(s)ds\right),\] implies that all three players' hazard rates coincide
a.e.\ on $(t^*,T)$, and hence must equal the AG hazard rate
\[
\rate_A(t)=\rate_B(t)=\rate_C(t)=\rate^{AG}:=\frac{r(1-\alpha)}{2\alpha-1}
\qquad\text{for a.e.\ }t\in(t^*,T).
\]

\medskip

We now pin down the \emph{time-0 atoms} and the \emph{identity of the initially inactive player}.

\begin{enumerate}
\item[\text{Case 1:}] \textbf{$z_C(0)>\max\{z_A(0),z_B(0)\}$.}
Since $t_C=0$ (Step 1), $C$ concedes with positive density immediately after $0$.
If neither $A$ nor $B$ concedes with an atom at $0$, then for all sufficiently small $t>0$ we have
$z_C(t)>z_A(t)$ and $z_C(t)>z_B(t)$, contradicting Lemma~\ref{Lemma: Not possible to have zC > zA,zB}.
Thus at least one of $A,B$ must have an atom at $0$.

We claim both must.
Suppose wlog that only $A$ has an atom at $0$, so $F_A(0)>0$ and $F_B(0)=0$.
Then $z_A(0+)\ge z_C(0)>z_B(0)$.
If $z_A(0+)>z_C(0)$, then by continuity of posteriors we would have
$z_A(t)>z_C(t)>z_B(t)$ for all small $t>0$, contradicting Lemma~\ref{Lemma: if zA > zB,zC then, A must not have conceded}.
If instead $z_A(0+)=z_C(0)>z_B(0)$, then for small $t>0$ we obtain the forbidden configuration
$z_A(t)=z_C(t)>z_B(t)$, contradicting Lemma~\ref{Lemma: Cannot have zA = zC > z_B}.
Hence it is impossible that only one of $A$ and $B$ has an atom at $0$.

Therefore both $A$ and $B$ concede with atoms at $0$, and those atoms are uniquely pinned down by
\[
z_A(0+)=z_B(0+)=z_C(0)
\quad\Longleftrightarrow\quad
\frac{z_A(0)}{1-F_A(0)}=\frac{z_B(0)}{1-F_B(0)}=z_C(0).
\]
After time $0$, all three are active immediately and hence concede at rate $\rate^{AG}$ until $T$ (Step 2).

\item[\text{Case 2:}] \textbf{$z_A(0)=z_B(0)>z_C(0)$.}
By Step 1, $C$ is active immediately after $0$.
If $F_C(0)=0$, then $z_C(t)=z_C(0)<z_A(0)=z_B(0)$ at $t=0$, and since $A,B$ are at least as reputable as $C$,
for all small $t>0$ we obtain $\min\{z_A(t),z_B(t)\}>z_C(t)$, contradicting Lemma~\ref{Lemma: not possible to have za, zb>zc}.
Hence $F_C(0)>0$.

Because $A$ and $B$ start with equal posteriors and Lemma~\ref{Lemma: Not possible to have zC > zA,zB}
rules out $z_C(0+)>z_A(0)$, the time-0 atom of $C$ must satisfy
\[
z_C(0+)=z_A(0)=z_B(0),
\qquad\text{i.e.}\qquad
F_C(0)=1-\frac{z_C(0)}{z_A(0)}.
\]
No atom by $A$ or $B$ is compatible with equilibrium.
After time $0$, all three are active and concede at rate $\rate^{AG}$ until $T$ (Step 2).

\item[\text{Case 3:}] \textbf{$z_A(0)=z_B(0)=z_C(0)$.}
Any time-0 atom by some player would create unequal posteriors at $0+$, which is impossible by the ordering lemmas.
Hence $F_A(0)=F_B(0)=F_C(0)=0$ and all three start conceding immediately at the common hazard rate $\rate^{AG}$.

\item[\text{Case 4:}] \textbf{$z_A(0)>z_B(0)\ge z_C(0)$.}
We claim that $A$ is the unique initially inactive player:
\[
\rate_A(t)=0\ \text{ on }[0,t_A)\quad\text{and}\quad \rate_B(t),\rate_C(t)>0\ \text{ on }(0,t_A),
\]
for some $t_A>0$.
Indeed, by Step 1 we have $t_C=0$.
If $A$ were active on an interval immediately after $0$, then for small $t>0$ we would have
$z_A(t)>z_C(t)$ and $z_A(t)>z_B(t)$, contradicting Lemma~\ref{Lemma: if zA > zB,zC then, A must not have conceded}.
Thus $A$ must be inactive initially; and by Lemma~\ref{Lemma: at most one CDF constant}, both $B$ and $C$ must be active.

On the initial interval $(0,t_A)$, we therefore have $\rate_A=0$ and $\rate_B,\rate_C>0$, so by
Lemma~\ref{Lemma: IC hold if CDFs increasing} the indifference conditions for $B$ and $C$ hold a.e.\ on $(0,t_A)$.
Since $\rate_A=0$, these conditions uniquely pin down the hazard rates on $(0,t_A)$ as
\[
\rate_C(t)=\rate^{AG}
\qquad\text{and}\qquad
\rate_B(t)=(1+\pi_{AC})\rate^{AG}
\qquad\text{for a.e.\ }t\in(0,t_A).
\]
By Step 2, the moment $A$ becomes active we must have $z_A=z_B=z_C$.
Since $A$ is inactive up to $t_A$, we have $z_A(t)=z_A(0)$ for $t<t_A$, so the alignment condition at $t_A$ is
\[
z_B(t_A)=z_C(t_A)=z_A(0).
\]
Using the hazard rates above, this system uniquely determines $t_A$ and $z_C(0+)$ and hence $F_C(0)$.
Equivalently,
\[
z_A(0)=z_B(0)\exp\big((1+\pi_{AC})\rate^{AG}t_A\big)
      =z_C(0+)\exp\big(\rate^{AG}t_A\big),
\quad\text{and}\quad
z_C(0+)=\frac{z_C(0)}{1-F_C(0)},
\]
which yields
\[
F_C(0)=1-\frac{z_C(0)}{z_A(0)^{\frac{\pi_{AC}}{1+\pi_{AC}}}\,z_B(0)^{\frac{1}{1+\pi_{AC}}}}.
\]
After $t_A$, all three are active and concede at hazard rate $\rate^{AG}$ until $T$.

\item[\text{Case 5:}] \textbf{$z_A(0)>z_C(0)>z_B(0)$.}
As in Case~4, $A$ is initially inactive while $B$ and $C$ are active on an initial interval $(0,t_A)$, and the same
indifference argument implies
\[
\rate_C(t)=\rate^{AG},\qquad \rate_B(t)=(1+\pi_{AC})\rate^{AG}\qquad\text{for a.e.\ }t\in(0,t_A),
\]
until posteriors align at $t_A$.

Let $\tilde{t}_C$ be the time it would take $C$ to reach $z_A(0)$ under hazard rate $\rate^{AG}$ absent any atom,
and let $\tilde{t}_B$ be the analogous time for $B$ under hazard rate $(1+\pi_{AC})\rate^{AG}$ absent any atom:
\[
\tilde{t}_C:= -\frac{1}{\rate^{AG}}\log\!\Big(\frac{z_C(0)}{z_A(0)}\Big),
\qquad
\tilde{t}_B:= -\frac{1}{(1+\pi_{AC})\rate^{AG}}\log\!\Big(\frac{z_B(0)}{z_A(0)}\Big).
\]
If $\tilde{t}_C>\tilde{t}_B$, then (absent atoms) $C$ would reach $z_A(0)$ strictly later than $B$.
Since only one opponent can concede with an atom at time $0$ in the negotiation between $B$ and $C$, $C$ must be the one who
has an atom at $0$ to speed up her posterior (and $B$ must have none). Hence,
\[
F_C(0)=1-\frac{z_C(0)}{z_A(0)^{\frac{\pi_{AC}}{1+\pi_{AC}}}\,z_B(0)^{\frac{1}{1+\pi_{AC}}}},
\qquad F_B(0)=0.
\]
If instead $\tilde{t}_B>\tilde{t}_C$, then $B$ must be the one who has an atom at $0$ (and $C$ must have none), yielding
\[
F_B(0)=1-\frac{z_B(0)\,z_A(0)^{\pi_{AC}}}{z_C(0)^{1+\pi_{AC}}},
\qquad F_C(0)=0.
\]
After $t_A$, all three are active and concede at hazard rate $\rate^{AG}$ until $T$.

\end{enumerate}

\medskip
\noindent\textbf{Uniqueness.}
In each case, the identity of the initially inactive player is pinned down by the posterior-ordering lemmas
together with Lemma~\ref{Lemma: at most one CDF constant}.
Given the set of active players on any interval, the local indifference conditions uniquely determine the hazards
on that interval. Finally, the size (and identity) of the time-0 atom(s) is uniquely pinned down by the requirement
that posteriors coincide at the time $t^*$ when the last player becomes active (Step 2), together with the fact that on
any contested negotiation at most one player can place an atom at time~$0$.
This proves that the equilibrium described in the proposition is the unique equilibrium.

\eprf 

\subsection{Proof of Proposition~\ref{Proposition: unique equilibrium payoffs}}\label{Appendix: Proposition of equilibrium payoffs}
\bprf
Fix the unique equilibrium $\sigma^*$ characterized in Proposition~\ref{prop:full}
(and let $F^*$ denote the induced cdfs in the no-concession subgame).
For each player $k\in\{A,B,C\}$, let $v_k^*:=v_k(\sigma^*)$ denote the time-0 expected payoff
of $k$'s rational type.

\paragraph{Step 0 (Payoff representations and a useful limit).}
Recall $U_i(t;\sigma_{-i})$ for $i\in\{A,B\}$ and
$U_C(t;\sigma_{-C})$ given in equations \eqref{eq:Ui_plan_payoff}--\eqref{eq:UC_plan_payoff}
in the main text. In particular, in equilibrium we may compute
\[
v_i^* = U_i(t;\sigma^*_{-i}) \quad\text{for any } t \text{ in the support of } dF_i^* \text{ on }(0,\infty),
\]
and analogously for $C$, by Lemma~\ref{Lemma: IC hold if CDFs increasing}.

We will repeatedly use the following observation: if player $k$ is active immediately after
$0$ (i.e.\ $F_k^*$ is strictly increasing on $(0,\varepsilon)$ for some $\varepsilon>0$), then
\[
v_k^*=\lim_{t\downarrow 0} U_k(t;\sigma_{-k}^*).
\]
This limit is obtained from \eqref{eq:Ui_plan_payoff}--\eqref{eq:UC_plan_payoff} by retaining
the time-$0$ contributions. Which terms matter depends on the case: for a peripheral, a
time-$0$ atom by the other peripheral may contribute through the induced AG continuation
term, while for $C$ the limit combines the time-$0$ contributions from both negotiations.
Accordingly, in the case-by-case comparisons below we evaluate the $t\downarrow 0$ limit
directly from \eqref{eq:Ui_plan_payoff}--\eqref{eq:UC_plan_payoff}.

\paragraph{Step 1 (Benchmark payoffs in the bilateral AG game).}
In the bilateral AG benchmark over surplus $\pi_{ij}$, the unique equilibrium features a
time-0 atom
\[
F^{AG}_{ji}(0) \;=\; \max\Big\{1-\frac{z_j(0)}{z_i(0)},\,0\Big\}.
\]
Player $i$'s (time-0) equilibrium payoff in that bilateral game is therefore
\begin{equation}\label{eq:AG_payoff_formula}
v^{AG}_{ij}(z_i(0),z_j(0))
\;=\;
\pi_{ij}\Big((1-\alpha)+(2\alpha-1)\,F^{AG}_{ji}(0)\Big).
\end{equation}
We write $v^{AG}_{iC}:=v^{AG}_{iC}(z_i(0),z_C(0))$ and similarly for $v^{AG}_{Ci}$.

\medskip
We now prove each payoff comparison in Proposition~\ref{Proposition: unique equilibrium payoffs}
case by case, using the equilibrium characterization in Proposition~\ref{prop:full}.

\begin{enumerate}
\item[\text{Case 1:}] \textbf{$z_C(0)\ge \max\{z_A(0),z_B(0)\}$.}
By Proposition~\ref{prop:full}, $C$ has no time-0 atom and each peripheral
$i\in\{A,B\}$ (if weaker than $C$) concedes with an atom at $0$ that raises $z_i(0+)$ to $z_C(0)$.
Equivalently,
\[
F_i^*(0) = \max\Big\{1-\frac{z_i(0)}{z_C(0)},\,0\Big\} = F^{AG}_{iC}(0),
\qquad F_C^*(0)=0.
\]
After time $0$, all posteriors are aligned and all players concede at the common AG hazard
$\rate^{AG}= \frac{r(1-\alpha)}{2\alpha-1}$, exactly as in the bilateral benchmark in each negotiation.
Since continuation play after any concession is the unique AG continuation, the induced outcome
distribution in each negotiation coincides with the AG benchmark, yielding
\[
v_i^* = v^{AG}_{iC}\ \text{ for }i\in\{A,B\},
\qquad
v_C^* = v^{AG}_{CA}+v^{AG}_{CB}.
\]
\item[Case~2:] \textbf{$z_C(0)<z_B(0)<z_A(0)$.}
Proposition~\ref{prop:full} implies: $A$ is initially inactive, $B$ and $C$
are active immediately after $0$, and $C$ concedes with a time-0 atom $F_C^*(0)\in(0,1)$,
while $F_A^*(0)=F_B^*(0)=0$.

\smallskip
\underline{$C$'s payoff.}
Since $F_C^*(0)>0$, time $0$ is in the support of $C$'s concession distribution. Sequential rationality
therefore implies that $C$ is indifferent between conceding at $0$ and following $\sigma^*$, hence
\[
v_C^* = (1-\alpha)(\pi_{AC}+1).
\]
In the AG benchmark, $C$ is the weaker party in both negotiations (since $z_C(0)<z_A(0)$ and $z_C(0)<z_B(0)$),
so $v^{AG}_{C,A}=(1-\alpha)\pi_{AC}$ and $v^{AG}_{C,B}=(1-\alpha)\cdot 1$, implying
$v_C^* = v^{AG}_{C,A}+v^{AG}_{C,B}$.

\smallskip
\underline{$B$'s payoff.}
Player $B$ is active immediately after $0$, so $v_B^*=\lim_{t\downarrow 0}U_B(t;\sigma^*_{-B})$.
Because $F_A^*(0)=0$ and $A$ is initially inactive, the only time-0 atom relevant for this limit is
$C$'s atom $F_C^*(0)$, yielding
\[
v_B^*
= \pi_{BC}\Big(\alpha F_C^*(0) + (1-\alpha)(1-F_C^*(0))\Big)
= \pi_{BC}\Big((1-\alpha)+(2\alpha-1)F_C^*(0)\Big).
\]
In the bilateral AG benchmark between $B$ and $C$, $C$ is the weak party, so the AG atom against $B$ is
$F^{AG}_{CB}(0)=1-\frac{z_C(0)}{z_B(0)}$ and
\[
v^{AG}_{B,C}=\pi_{BC}\Big((1-\alpha)+(2\alpha-1)F^{AG}_{CB}(0)\Big).
\]
Finally, Proposition~\ref{prop:full} gives
\[
F_C^*(0)=1-\frac{z_C(0)}{z_A(0)^{\frac{\pi_{AC}}{1+\pi_{AC}}}z_B(0)^{\frac{1}{1+\pi_{AC}}}}
\quad\Rightarrow\quad
F_C^*(0) > 1-\frac{z_C(0)}{z_B(0)} = F^{AG}_{CB}(0),
\]
because $z_A(0)^{\frac{\pi_{AC}}{1+\pi_{AC}}}z_B(0)^{\frac{1}{1+\pi_{AC}}} > z_B(0)$ when $z_A(0)>z_B(0)$.
Therefore $v_B^*>v^{AG}_{BC}$.

\smallskip
\underline{$A$'s payoff.}
In the AG benchmark between $A$ and $C$, $C$ concedes with a time-0 atom of size
$F^{AG}_{CA}(0)=1-\frac{z_C(0)}{z_A(0)}$, which raises $A$'s payoff above $(1-\alpha)\pi_{AC}$
via \eqref{eq:AG_payoff_formula}.
In the three-player equilibrium, $A$ is initially inactive and posteriors align only at the strictly
positive time $t_A>0$ at which $A$ becomes active. Hence part of the probability mass of
``$C$ concedes before $A$ concedes'' that is realized at calendar time $0$ in the bilateral benchmark
is shifted to calendar times in $(0,t_A]$ in the three-player equilibrium.
Because $\alpha>(1-\alpha)$ and discounting is strict, shifting any positive probability mass of
receiving a concession from time $0$ to a strictly positive time strictly lowers the time-0
expected payoff. Therefore $v_A^*<v^{AG}_{A,C}$.

Combining the three comparisons yields the claim in Case~2.
\item[Case~3:] \textbf{$z_B(0)<z_C(0)<z_A(0)$.}
Proposition~\ref{prop:full} implies $A$ is initially inactive and $B$ and $C$
are active immediately after $0$. Moreover, exactly one of the following holds:
either (a) $C$ concedes with a time-0 atom $F_C^*(0)>0$ and $F_B^*(0)=0$,
or (b) $B$ concedes with a time-0 atom $F_B^*(0)>0$ and $F_C^*(0)=0$.

\smallskip
\underline{$A$'s payoff.}
In both subcases, $C$ is weaker than $A$ initially and the bilateral benchmark for $A$ and $C$
features a strictly positive time-0 atom from $C$ to $A$ of size $1-\frac{z_C(0)}{z_A(0)}$.
In the three-player equilibrium, $A$ remains inactive for an initial interval and posteriors align
only at a strictly positive time $t_A>0$. Hence, relative to the bilateral benchmark, some
probability mass of being conceded to by $C$ is shifted away from time $0$ to later calendar
times, which strictly reduces the discounted value. Therefore $v_A^*<v^{AG}_{AC}$.

\smallskip
\underline{$B$'s payoff.}
If (a) $F_C^*(0)>0$, then $B$ is active immediately after $0$ and the limit argument from Case~2 gives
\[
v_B^*=\pi_{BC}\Big((1-\alpha)+(2\alpha-1)F_C^*(0)\Big)>(1-\alpha)\pi_{BC}.
\]
Since $z_B(0)<z_C(0)$, $B$ is the weak player in the bilateral benchmark on $BC$, so $v^{AG}_{B,C}=(1-\alpha)\pi_{BC}$.
Hence $v_B^*>v^{AG}_{B,C}$.
If instead (b) $F_B^*(0)>0$, then time $0$ is in the support of $B$'s concession distribution, so
sequential rationality implies $v_B^*=(1-\alpha)\pi_{BC}=v^{AG}_{B,C}$.
Thus, in either subcase, $v_B^*\ge v^{AG}_{B,C}$.

\smallskip
\underline{$C$'s payoff.}
In the bilateral benchmark, $C$ is weak against $A$ but strong against $B$, so
\[
v^{AG}_{C,A}+v^{AG}_{C,B}
= (1-\alpha)(\pi_{AC}+1) + (2\alpha-1)\Big(1-\frac{z_B(0)}{z_C(0)}\Big),
\]
where the last (strictly positive) term is the strong-player premium in the negotiation between $B$ and $C$.

If (a) $F_C^*(0)>0$, then $C$ concedes at time $0$ with positive probability, hence
$v_C^*=(1-\alpha)(\pi_{AC}+1)$ as in Case~2. This is strictly smaller than
$v^{AG}_{C,A}+v^{AG}_{C,B}$ because $z_B(0)<z_C(0)$ implies $1-\frac{z_B(0)}{z_C(0)}>0$.

If (b) $F_B^*(0)>0$ and $F_C^*(0)=0$, then $C$ is active immediately after $0$ and we may evaluate
$v_C^*=\lim_{t\downarrow 0}U_C(t;\sigma^*_{-C})$.
In this limit, the only time-0 atom is $B$'s concession at $0$, and the continuation between $A$ and $C$
following $B$'s concession is the bilateral AG game between $A$ and $C$ starting from
$(z_A(0),z_C(0))$, in which $C$ is the weak player and hence earns $(1-\alpha)\pi_{AC}$.
Therefore,
\[
v_C^* = (1-\alpha)(\pi_{AC}+1) + (2\alpha-1)F_B^*(0).
\]
Moreover, Proposition~\ref{prop:full} yields

\[
F_B(0)=
1-\frac{z_B(0)\,z_A(0)^{\pi_{AC}}}{z_C(0)^{1+\pi_{AC}}}\;<\;
1-\frac{z_B(0)}{z_C(0)}
= F^{AG}_{BC}(0),
\]

since $z_A(0)^{\pi_{AC}}>z_C(0)^{\pi_{AC}}$ when $z_A(0)>z_C(0)$.
Hence $v_C^*<v^{AG}_{CA}+v^{AG}_{CB}$ also in subcase (b).

This completes the proof of Case~3 and hence of Proposition~\ref{Proposition: unique equilibrium payoffs}.

\end{enumerate}

\eprf 
\subsection*{A.4. Proof of Proposition~\ref{prop:vanishing}}

\begin{proof}
Set $\pi_{AC}=\pi_{BC}=1$ and write
\[
z_B^\varepsilon(0)=\varepsilon,\qquad
z_C^\varepsilon(0)=\kappa_B\varepsilon,\qquad
z_A^\varepsilon(0)=\kappa_A\kappa_B\varepsilon,
\]
with $\kappa_A,\kappa_B>1$ fixed. To lighten notation, suppress the dependence on
$\varepsilon$.

In the bilateral AG benchmark,
\[
F_{CA}^{AG}(0)=1-\frac{z_C(0)}{z_A(0)}=1-\frac{1}{\kappa_A},
\qquad
F_{BC}^{AG}(0)=1-\frac{z_B(0)}{z_C(0)}=1-\frac{1}{\kappa_B}.
\]
Hence
\[
v_{CA}^{AG}=1-\alpha,\qquad
v_{CB}^{AG}=\alpha-\frac{2\alpha-1}{\kappa_B},\qquad
v_{AC}^{AG}=\alpha-\frac{2\alpha-1}{\kappa_A},\qquad
v_{BC}^{AG}=1-\alpha.
\]

We now consider the two cases identified in Proposition~A.1.

\medskip
\noindent\emph{Case 1: $\kappa_A>\kappa_B$.}
By Proposition~A.1, player $C$ concedes at $t=0$ with atom
\[
F_C^*(0)=1-\sqrt{\frac{\kappa_B}{\kappa_A}},
\qquad
F_B^*(0)=0.
\]
Since $C$ is indifferent at $t=0$,
\[
v_C^*=2(1-\alpha).
\]
Therefore
\[
v_{CA}^{AG}+v_{CB}^{AG}-v_C^*
=(1-\alpha)+\left(\alpha-\frac{2\alpha-1}{\kappa_B}\right)-2(1-\alpha)
=(2\alpha-1)\left(1-\frac{1}{\kappa_B}\right).
\]
Moreover, $B$ is active immediately after $0$, so
\[
v_B^*
=\alpha F_C^*(0)+(1-\alpha)\bigl(1-F_C^*(0)\bigr)
=\alpha-(2\alpha-1)\sqrt{\frac{\kappa_B}{\kappa_A}}
>1-\alpha
= v_{BC}^{AG}.
\]
Finally, part~(ii) follows directly from Proposition~\ref{Proposition: unique equilibrium payoffs}: for every prior
vector satisfying $z_A(0)>z_C(0)>z_B(0)$, we have
\[
v_A^*<v_{AC}^{AG}.
\]

\medskip
\noindent\emph{Case 2: $\kappa_A\le \kappa_B$.}
By Proposition~A.1, player $B$ concedes at $t=0$ with atom
\[
F_B^*(0)=1-\frac{\kappa_A}{\kappa_B},
\qquad
F_C^*(0)=0.
\]
Because $B$ is indifferent at $t=0$,
\[
v_B^*=1-\alpha = v_{BC}^{AG}.
\]
When $B$ concedes at $t=0$, player $C$ receives $\alpha$ from the $B$--$C$ negotiation
immediately and $(1-\alpha)$ from the continuation against $A$. When $B$ does not concede
at $t=0$, player $C$ earns $2(1-\alpha)$ by indifference. Hence
\[
v_C^*
=2(1-\alpha)+(2\alpha-1)\left(1-\frac{\kappa_A}{\kappa_B}\right).
\]
Therefore
\[
v_{CA}^{AG}+v_{CB}^{AG}-v_C^*
=(2\alpha-1)\left(1-\frac{1}{\kappa_B}\right)
-(2\alpha-1)\left(1-\frac{\kappa_A}{\kappa_B}\right)
=(2\alpha-1)\frac{\kappa_A-1}{\kappa_B}.
\]
Again, part~(ii) is immediate from Proposition~\ref{Proposition: unique equilibrium payoffs}.

Combining the two cases yields
\[
v_{CA}^{AG}+v_{CB}^{AG}-v_C^*
=(2\alpha-1)\frac{\min\{\kappa_A,\kappa_B\}-1}{\kappa_B}>0,
\]
and
\[
v_B^*\ge v_{BC}^{AG},
\]
with strict inequality if and only if $\kappa_A>\kappa_B$. Since the expressions above do
not depend on $\varepsilon$, the same formulas describe the limit as $\varepsilon\downarrow0$.
\end{proof}

\bibliography{ref-bargaining}
\newpage 

\section*{Online Appendix}
\setcounter{page}{1}

\section{Sequential Negotiations}\label{Sec: sequential}

In this section, we relax the assumption of simultaneous negotiations. To this end, consider an environment in which players $A$ and $C$ bargain first over surplus $\pi_{AC}>0$ according to the same continuous-time war-of-attrition protocol as in the baseline, and only after the first-stage dispute ends do players $B$ and $C$ bargain over surplus $\pi_{BC}=1$. All concessions and agreement times are publicly observed. Types are as in the baseline: each player $i\in\{A,B,C\}$ is behavioral with prior $z_i(0)\in(0,1)$ (never concedes) and rational otherwise. Player $C$ has a single global type. Histories and strategies are adapted in the obvious way. Payoffs are additive across negotiations and discounted (in calendar time).

Let $t$ denote the calendar time at which the first-stage dispute (between $A$ and $C$) ends. If $A$ concedes at $t$, then $C$'s type is not revealed; the second-stage game between $B$ and $C$ is a bilateral AG reputational war of attrition with initial reputations $(z_B(0),z_C(t))$, where $z_C(t)$ is the posterior that $C$ is behavioral after observing no concession by $C$ up to time $t$. If instead $C$ concedes to $A$ at $t$, then $C$ is revealed rational at $t^+$, and hence in the second-stage AG continuation with initial reputations $(z_B(0),0)$ the (unique) equilibrium has $C$ conceding immediately with probability $1$.

Formally, define (as before) $V^{AG}_{CB}(z_C,z_B)$ as $C$'s time-$0$ equilibrium payoff in a bilateral AG game over surplus $1$ against $B$ with reputations $(z_C,z_B)$:
\[
V^{AG}_{CB}(z_C,z_B)
= (1-\alpha) + (2\alpha-1)\max\left\{1-\frac{z_B}{z_C},\,0\right\}.
\]
(Equivalently, $B$ concedes at the start of stage 2 with probability
$\max\{1-z_B/z_C,0\}$.)

As Proposition~\ref{prop:SequentialUnique} shows, there exists a unique equilibrium in this environment. The first-stage bargaining between $A$ and $C$ features a common terminal time $T$ at which posteriors reach $1$. However, the presence of the second-stage negotiation fundamentally alters first-stage incentives. In particular, the central player's continuation value from a future negotiation between $B$ and $C$ enters the local indifference condition of the first-stage war of attrition. Specifically, player $C$ is indifferent at time $t$ if
\[r (\pi_{AC}+1)(1-\alpha)=\left(\alpha \pi_{AC}+V^{AG}_{CB}(z_C(t),z_B(0))-\left(\pi_{AC}+1\right)(1-\alpha)\right)\frac{f_{A}(t)}{1-F_A(t)}.\]

As a result, equilibrium concession hazard rates in the first stage are asymmetric even when initial reputations are symmetric. Specifically, for $t\in(0,T)$,
\begin{equation}\rate_C(t)=\rate^{AG},\;\; \text{and} \;\;
\rate_A(t)=\frac{(\pi_{AC}+1)}{\left(\pi_{AC}+\max\left\{1-\frac{z_B(0)}{z_C(t)},0\right\}\right)}\lambda^{AG}\geq \rate^{AG}.
\label{eq:seqhaz}\end{equation}

Because $A$ concedes at a higher rate, $A$'s posterior reputation rises faster than $C$'s. To reconcile this with the common-terminal-time requirement, equilibrium may require an atom of concession by $C$ at time $0$, even if $z_A(0)<z_C(0)$.

Two implications follow. First, even when initial reputations are symmetric, the central player is strictly worse off than in the bilateral benchmark, while $A$ is strictly better off. Second, this disadvantage may persist when $C$ is stronger than $A$: the faster posterior growth of $A$ may still force an initial adjustment by $C$. Once the first-stage negotiation ends, the second stage bargaining between $B$ and $C$ reduces to a standard bilateral war of attrition with beliefs updated from the first-stage history. In particular, if $A$ concedes, the continuation game between $B$ and $C$ is exactly the bilateral benchmark; if $C$ concedes, she is revealed to be rational and concedes immediately in the second stage. Let us now take a closer look at each stage.

\paragraph{Stage 2 (benchmark continuation).}
Given any stage~2 start time $\tau\ge 0$, the continuation game between $B$ and $C$ is the
bilateral AG reputational war of attrition with priors $(z_B(0),\widehat z_C)$, where
$\widehat z_C$ is the public posterior on $C$ at the start of stage~2. Since $\pi_{BC}=1$, we denote
by $V^{AG}_{CB}(\widehat z_C,z_B(0))$ the (current-value) equilibrium payoff of a \emph{rational}
$C$ in this bilateral game. Using the benchmark characterization recalled in Section~\ref{Section: benchmark},
\begin{equation}\label{eq:VCB_AG}
V^{AG}_{CB}(z_C,z_B)
= (1-\alpha) + (2\alpha-1)\,g(z_C;z_B),
\qquad
g(z_C;z_B):=\max\Bigl\{1-\frac{z_B}{z_C},\,0\Bigr\}.
\end{equation}
Here $g(z_C;z_B)$ is exactly the equilibrium time-$0$ concession probability of $B$ in the
bilateral $B$--$C$ game when $z_C>z_B$ (and equals $0$ when $z_C\le z_B$).

\paragraph{Stage 1 (payoffs as functions of the stage-2 continuation).}
Fix any PBE. Consider stage~1 at time $t$ along the stage-1 no-concession history.
If $A$ concedes at time $t$, then the stage~1 negotiation ends with $C$ receiving $\alpha\pi_{AC}$
and $A$ receiving $(1-\alpha)\pi_{AC}$, and then stage~2 begins immediately with public posterior
$\widehat z_C = z_C(t)$ (because $C$ has not conceded up to $t$).
Thus the \emph{current-value} payoff to a rational $C$ at time $t$ from an $A$-concession equals
\begin{equation}\label{eq:WC_Aconcedes}
W_C(t):=\alpha\pi_{AC} + V^{AG}_{CB}(z_C(t),z_B(0)).
\end{equation}
If instead $C$ concedes at time $t$ in stage~1, then she is revealed rational; in stage~2 she
concedes immediately to $B$ (a rational $C$ strictly prefers immediate concession against a
known-rational type than any delay, and $B$ never concedes against a player known to concede
immediately). Therefore the current-value payoff to a rational $C$ from conceding at time $t$
in stage~1 is
\begin{equation}\label{eq:WC_Cconcedes}
(1-\alpha)\pi_{AC} + (1-\alpha)\cdot 1 = (1-\alpha)(\pi_{AC}+1).
\end{equation}

\setcounter{proposition}{0}
\renewcommand{\theproposition}{B.\arabic{proposition}}
\begin{proposition}[Sequential negotiations: unique equilibrium]\label{prop:SequentialUnique}
There exists a unique PBE of the sequential game. In this PBE:

\begin{enumerate}
\item (\textbf{Stage 2}) If stage~1 ends at time $\tau$ because $A$ concedes, then stage~2 is the
bilateral AG equilibrium between $B$ and $C$ with initial posteriors
$(z_B(0),z_C(\tau))$.
If stage~1 ends at time $\tau$ because $C$ concedes, then $C$ is revealed rational and concedes
immediately to $B$ at the start of stage~2.

\item (\textbf{Stage 1 hazards}) Along the stage-1 no-concession history, there is a finite terminal
time $T\in(0,\infty)$ such that $z_A(T)=z_C(T)=1$.
On $(0,T)$ both players concede with positive density a.e., and their equilibrium hazards satisfy
\begin{equation}\label{eq:stage1_hazards}
\lambda_C(t)=\lambda^{AG}:=\frac{r(1-\alpha)}{2\alpha-1},
\qquad
\lambda_A(t)=\lambda^{AG}\cdot\frac{\pi_{AC}+1}{\pi_{AC}+g(z_C(t);z_B(0))}
\quad\text{for a.e.\ }t\in(0,T),
\end{equation}
where $g(\cdot\,;\cdot)$ is defined in \eqref{eq:VCB_AG}.

Equivalently, if $z_C(t)\le z_B(0)$ then $\lambda_A(t)=\lambda^{AG}\cdot\frac{\pi_{AC}+1}{\pi_{AC}}$,
while if $z_C(t)>z_B(0)$ then
$\lambda_A(t)=\lambda^{AG}\cdot\frac{\pi_{AC}+1}{\pi_{AC}+1-z_B(0)/z_C(t)}$.

\item (\textbf{Time-0 atom and terminal time})
At most one player concedes with positive probability at time $0$ in stage~1.
Let $z_i(0+):=z_i(0)/(1-F_i(0))$ denote the posteriors after any time-$0$ atom in stage~1.
The terminal time is
\begin{equation}\label{eq:T_seq}
T=-\frac{1}{\lambda^{AG}}\log z_C(0+).
\end{equation}

Define the threshold $\bar z_A$ by
\[
\bar z_A :=
\begin{cases}
\displaystyle
\frac{(\pi_{AC}+1)\,z_C(0)-z_B(0)}{\pi_{AC}+1-z_B(0)}, & \text{if } z_C(0)\ge z_B(0),\\[1.25em]
\displaystyle
\frac{\pi_{AC}\,z_B(0)}{\pi_{AC}+1-z_B(0)}
\Bigl(\frac{z_C(0)}{z_B(0)}\Bigr)^{\frac{\pi_{AC}+1}{\pi_{AC}}},
& \text{if } z_C(0)< z_B(0).
\end{cases}
\]

\begin{enumerate}[(i)]
\item If $z_A(0)<\bar z_A$, then $A$ bears the unique time-$0$ atom and $C$ does not:
\[
F_A(0)=1-\frac{z_A(0)}{\bar z_A},
\qquad
F_C(0)=0,
\qquad
z_C(0+)=z_C(0),
\qquad
z_A(0+)=\bar z_A.
\]

\item If $z_A(0)=\bar z_A$, then there is no time-$0$ atom:
\[
F_A(0)=F_C(0)=0,
\qquad
z_A(0+)=z_A(0),
\qquad
z_C(0+)=z_C(0).
\]

\item If $z_A(0)>\bar z_A$, then $C$ bears the unique time-$0$ atom and $A$ does not:
\[
F_A(0)=0,
\qquad
F_C(0)=1-\frac{z_C(0)}{z_C(0+)},
\qquad
z_A(0+)=z_A(0),
\]
where $z_C(0+)$ is uniquely determined as follows. Let
\[
\tilde z_A := \frac{\pi_{AC}\,z_B(0)}{\pi_{AC}+1-z_B(0)}.
\]
If $z_A(0)\ge \tilde z_A$, then $z_C(0+)\ge z_B(0)$ and
\[
z_C(0+)=\frac{z_B(0)+z_A(0)\,(\pi_{AC}+1-z_B(0))}{\pi_{AC}+1}.
\]
If $z_A(0)<\tilde z_A$, then $z_C(0+)< z_B(0)$ and
\[
z_C(0+)= z_B(0)\Biggl(\frac{z_A(0)\,(\pi_{AC}+1-z_B(0))}{\pi_{AC}\,z_B(0)}\Biggr)^{\frac{\pi_{AC}}{\pi_{AC}+1}}.
\]
\end{enumerate}
\end{enumerate}
\end{proposition}

\begin{proof}[Sketch of proof.]
\textbf{Step 1 (Stage 2 is pinned down).}
Fix any public history at the start of stage~2. Conditional on this history, the continuation
game is exactly a bilateral reputational war of attrition with one commitment type for each
player. By the bilateral benchmark (AG), the equilibrium is unique and yields the
current-value payoff $V^{AG}_{CB}(\cdot,\cdot)$ for a rational $C$, given by \eqref{eq:VCB_AG}.
This proves part~1 of the proposition.

\textbf{Step 2 (No stalling in stage 1).}
Consider stage~1 and any interval $(t_0,t_1)\subset(0,\infty)$ along the stage-1 no-concession
history.
If $\lambda_A(t)=\lambda_C(t)=0$ for a.e.\ $t\in(t_0,t_1)$, then conditional on reaching $t_0$
the negotiation cannot end on $(t_0,t_1)$. Since discounting is strict, any rational player
strictly prefers conceding at $t_0$ to conceding at any later time in $(t_0,t_1]$, contradicting
sequential rationality. Hence it cannot be that both hazards are zero a.e.\ on any positive-length
interval.

\textbf{Step 3  (Local indifference for $A$  pins down $\lambda_C$).}
Fix $t\in(0,T)$ such that $A$ concedes with positive density at $t$ (equivalently, $\lambda_A(t)>0$)
and $F_A(t)<1$. Then following standard arguments, 

\[
r(1-\alpha)\pi_{AC}=(2\alpha-1)\pi_{AC}\lambda_C(t).
\]

Hence, $\lambda_C(t)=\lambda^{AG}$ for a.e.\ $t\in(0,T)$.

\textbf{Step 4 (Local indifference for $C$ pins down $\lambda_A(t)$).}
Fix $t\in(0,T)$ such that $C$ concedes with positive density at $t$ (so $\lambda_C(t)>0$) and
$F_C(t)<1$. Indifference of $C$ at time $t$ implies
\begin{equation}\label{eq:C_indiff_raw}
r(1-\alpha)(\pi_{AC}+1)=\lambda_A(t)\bigl(W_C(t)-(1-\alpha)(\pi_{AC}+1)\bigr).
\end{equation}
Substituting $W_C(t)=\alpha\pi_{AC}+V^{AG}_{CB}(z_C(t),z_B(0))$ from \eqref{eq:WC_Aconcedes}
and using \eqref{eq:VCB_AG}, and $\lambda^{AG}=\frac{r(1-\alpha)}{2\alpha-1}$, and plugging into
 \eqref{eq:C_indiff_raw} yields
\[
\lambda_A(t)
=\frac{r(1-\alpha)(\pi_{AC}+1)}{(2\alpha-1)\bigl(\pi_{AC}+g(z_C(t);z_B(0))\bigr)}
=\lambda^{AG}\cdot\frac{\pi_{AC}+1}{\pi_{AC}+g(z_C(t);z_B(0))}.
\]
This proves \eqref{eq:stage1_hazards}.

\textbf{Step 5 (Posteriors and existence of a finite terminal time).}
Since $\lambda_C(t)=\lambda^{AG}>0$ a.e.\ on $(0,T)$, Bayes' rule implies
$z_C(t)=z_C(0+)\exp(\lambda^{AG}t)$ along the stage-1 no-concession history until the
strategy support ends. Because $z_C(0+)\in(0,1)$, there is a unique finite time $T$ at which
$z_C(T)=1$, namely \eqref{eq:T_seq}. At this time, $F_C(T)=1-z_C(0)$, so a rational $C$ has
conceded in stage~1 with probability $1$ by time $T$.

Next we show that necessarily $z_A(T)=1$ as well. Suppose instead that $z_A(T)<1$.
Then conditional on no concession up to time $T$, there is positive probability that $A$ is rational.
But at time $T$, we have $z_C(T)=1$, so conditional on the same history player $C$ is behavioral
with probability one, hence will never concede. Therefore any rational $A$ who has not conceded by
$T$ strictly prefers conceding immediately at $T$ (yielding $(1-\alpha)\pi_{AC}$) to any strategy that
delays concession beyond $T$ (which can only weakly reduce his discounted payoff because it cannot
induce concession by $C$). This contradicts sequential rationality unless a rational $A$ concedes by time
$T$ with probability one, which is equivalent to $z_A(T)=1$.
Hence along the stage-1 no-concession history we must have $z_A(T)=z_C(T)=1$, proving the
terminal-time claim in part~2.

\textbf{Step 6 (Solving for the time-0 atom).}
Because $z_i(0+)=z_i(0)/(1-F_i(0))$, any time-$0$ atom by player $i$ weakly increases $z_i(0+)$.
We first note that at most one player concedes with positive probability at time $0$:
if both $A$ and $C$ placed positive mass at $0$, then conditional on being rational each would
strictly prefer shifting an $\varepsilon$-amount of mass from $0$ to a very small $\delta>0$ to avoid
the strictly worse simultaneous-concession outcome at $0$, contradicting optimality.
Hence at most one of $F_A(0),F_C(0)$ is positive.

We now compute the boundary condition $z_A(T)=1$ explicitly as a function of $(z_A(0+),z_C(0+))$.
Recall $z_C(t)=z_C(0+)\exp(\lambda^{AG}t)$.

\emph{Case 6.1: $z_C(0+)\ge z_B(0)$.}
Then for all $t\in[0,T]$ we have $z_C(t)\ge z_B(0)$, so
$g(z_C(t);z_B(0))=1-z_B(0)/z_C(t)$.
From \eqref{eq:stage1_hazards} we obtain the ODE
\[
\frac{\dot z_A(t)}{z_A(t)}=\lambda_A(t)
=\lambda^{AG}\cdot\frac{\pi_{AC}+1}{\pi_{AC}+1-z_B(0)/z_C(t)}.
\]
Using $\dot z_C(t)=\lambda^{AG}z_C(t)$ and the change of variables $z_A$ as a function of $z_C$,
one verifies that
\[
\frac{d}{dt}\log\Bigl((\pi_{AC}+1)z_C(t)-z_B(0)\Bigr)=\lambda_A(t),
\]
hence integrating from $0$ to $t$ yields
\[
z_A(t)=z_A(0+)\cdot
\frac{(\pi_{AC}+1)z_C(t)-z_B(0)}{(\pi_{AC}+1)z_C(0+)-z_B(0)}.
\]
At $t=T$ we have $z_C(T)=1$, so the boundary condition $z_A(T)=1$ becomes
\begin{equation}\label{eq:boundary_case_ge}
1=z_A(0+)\cdot\frac{\pi_{AC}+1-z_B(0)}{(\pi_{AC}+1)z_C(0+)-z_B(0)}.
\end{equation}

\emph{Case 6.2: $z_C(0+)< z_B(0)$.}
Let $t_B:=\frac{1}{\lambda^{AG}}\log\frac{z_B(0)}{z_C(0+)}\in(0,T)$ so that $z_C(t_B)=z_B(0)$.
For $t\in[0,t_B]$ we have $g(z_C(t);z_B(0))=0$, so $\lambda_A(t)=\lambda^{AG}\frac{\pi_{AC}+1}{\pi_{AC}}$
is constant and thus
\[
z_A(t_B)=z_A(0+)\exp\!\left(\lambda^{AG}\frac{\pi_{AC}+1}{\pi_{AC}}t_B\right)
=z_A(0+)\left(\frac{z_B(0)}{z_C(0+)}\right)^{\!\frac{\pi_{AC}+1}{\pi_{AC}}}.
\]
For $t\in[t_B,T]$ we have $z_C(t)\ge z_B(0)$, hence the linear formula from Case~6.1 applies
starting at $t_B$:
\[
z_A(t)=z_A(t_B)\cdot\frac{(\pi_{AC}+1)z_C(t)-z_B(0)}{(\pi_{AC}+1)z_C(t_B)-z_B(0)}
=z_A(t_B)\cdot\frac{(\pi_{AC}+1)z_C(t)-z_B(0)}{\pi_{AC}z_B(0)}.
\]
Evaluating at $t=T$ where $z_C(T)=1$, the boundary condition $z_A(T)=1$ becomes
\begin{equation}\label{eq:boundary_case_lt}
1
= z_A(0+)\Bigl(\frac{z_B(0)}{z_C(0+)}\Bigr)^{\frac{\pi_{AC}+1}{\pi_{AC}}}
\cdot\frac{\pi_{AC}+1-z_B(0)}{\pi_{AC}z_B(0)}.
\end{equation}

\textbf{Step 7 (Uniqueness and the explicit atoms).}
Clearly, there can be at most one atom at $0$.

\emph{(a) $C$ has no atom: $z_C(0+)=z_C(0)$.}
If $z_C(0)\ge z_B(0)$, then \eqref{eq:boundary_case_ge} implies the unique required value of $z_A(0+)$ is
\[
z_A(0+)=\frac{(\pi_{AC}+1)z_C(0)-z_B(0)}{\pi_{AC}+1-z_B(0)}=: \bar z_A.
\]
If $z_C(0)<z_B(0)$, then \eqref{eq:boundary_case_lt} implies the unique required value of $z_A(0+)$ is
\[
z_A(0+)=\frac{\pi_{AC}z_B(0)}{\pi_{AC}+1-z_B(0)}
\Bigl(\frac{z_C(0)}{z_B(0)}\Bigr)^{\frac{\pi_{AC}+1}{\pi_{AC}}}
=: \bar z_A.
\]
In either subcase, if $z_A(0)<\bar z_A$ then the only way to achieve $z_A(0+)=\bar z_A$ is that $A$
places a time-$0$ atom of size $F_A(0)=1-z_A(0)/\bar z_A$, while $F_C(0)=0$.
If $z_A(0)=\bar z_A$, no atom is needed.
If $z_A(0)>\bar z_A$, then $A$ cannot reduce $z_A(0+)$ (atoms only increase posteriors), so $C$ must
instead have the unique time-$0$ atom.

\emph{(b) $A$ has no atom: $z_A(0+)=z_A(0)$.}
In this case $C$ must choose $z_C(0+)\ge z_C(0)$ such that the boundary condition holds.
If we seek a solution with $z_C(0+)\ge z_B(0)$, then \eqref{eq:boundary_case_ge} uniquely yields
\[
(\pi_{AC}+1)z_C(0+)-z_B(0) = z_A(0)\,(\pi_{AC}+1-z_B(0)),
\quad\text{so}\quad
z_C(0+)=\frac{z_B(0)+z_A(0)(\pi_{AC}+1-z_B(0))}{\pi_{AC}+1}.
\]
This candidate indeed satisfies $z_C(0+)\ge z_B(0)$ if and only if
$z_A(0)\ge \tilde z_A:=\frac{\pi_{AC}z_B(0)}{\pi_{AC}+1-z_B(0)}$.
If instead we seek a solution with $z_C(0+)<z_B(0)$, then \eqref{eq:boundary_case_lt} uniquely yields
\[
z_C(0+) = z_B(0)\Biggl(\frac{z_A(0)(\pi_{AC}+1-z_B(0))}{\pi_{AC}z_B(0)}\Biggr)^{\frac{\pi_{AC}}{\pi_{AC}+1}},
\]
and this candidate satisfies $z_C(0+)<z_B(0)$ if and only if $z_A(0)<\tilde z_A$.
In either subcase $C$'s time-$0$ atom size is uniquely
$F_C(0)=1-z_C(0)/z_C(0+)$ and $F_A(0)=0$.

Combining parts (a) and (b) yields the case distinction and formulas stated in part~3.
Given $(z_A(0+),z_C(0+))$, the hazards \eqref{eq:stage1_hazards} pin down the full
stage-1 concession-time distributions (via $\dot z_i(t)=\lambda_i(t)z_i(t)$ and Bayes' rule),
and stage~2 play is uniquely pinned down by Step~1.
Therefore, the PBE is unique.
\end{proof}

\section{Partial observability of concessions}
\label{Section: partial}
 
We relax the assumption that the identity of the conceding player is publicly
observed. Suppose that when the negotiation between $i\in\{A,B\}$ and $C$ ends
at time $t$, the uninvolved peripheral $k$ observes only that an agreement was
reached at time $t$, not which party conceded. We restrict attention to the
symmetric case $z_A(0)=z_B(0)=z_C(0)=:z_0$ and $\pi_{AC}=\pi_{BC}=1$.

Until the first agreement, public histories coincide with those in the baseline, so posteriors are common along the no-concession path. After an agreement in the other negotiation, the uninvolved peripheral’s belief about $C$ need not coincide with the involved players’ beliefs. Accordingly, write $z_C^k(t)$ for peripheral $k$'s posterior probability at time $t$ that $C$ is behavioral.

\paragraph{Belief updating.}
Along the no-concession path, posteriors are common. When the $i$--$C$
negotiation ends at time $t$, peripheral $k$'s posterior about $C$ jumps to
\begin{equation}\label{eq:belief_jump_sym}
z_C^k(t^+)
\;=\;
z_C(t^-)\,\frac{\lambda_i(t)}{\lambda_i(t)+\lambda_C(t)},
\end{equation}
which is a strict downward revision whenever $\lambda_C(t)>0$. In the unique
AG continuation of the $k$--$C$ negotiation, $C$ then concedes immediately at
$t^+$ with atom
\[
g_C^k(t)
\;=\;
\max\!\left\{1-\frac{z_C^k(t^+)}{z_k(t^+)},\,0\right\}.
\]
 
\paragraph{Candidate equilibrium.}
We conjecture that $A$ and $B$ concede at rate $\lambda^{AG}$ throughout,
while $C$ concedes with a strictly positive atom $F_C(0)>0$ at $t=0$ and
then at a time-varying rate $\lambda_C(t)$ on $(0,T)$. Under this profile,
\[
z_A(t)=z_B(t)=z_0 e^{\lambda^{AG}t},
\qquad
z_C(t)=z_C(0^+)\exp\!\Bigl(\int_0^t\lambda_C(s)\,ds\Bigr),
\]
where $z_C(0^+)=z_0/(1-F_C(0))$.
 
\paragraph{$A$'s indifference condition pins down $\lambda_C$.}
Under \eqref{eq:belief_jump_sym} with $\lambda_i=\lambda^{AG}$, denote $A$'s
updated belief about $C$ upon observing the $B$--$C$ agreement at time $t$ by
\[
\hat{z}_C(t) \;:=\; z_C^A(t^+)
\;=\; z_C(t)\,\frac{\lambda^{AG}}{\lambda^{AG}+\lambda_C(t)}.
\]
The atom that $C$ makes to $A$ in the AG continuation is then
\[
g_C^A(t)
\;=\;
\max\!\left\{1-\frac{\hat{z}_C(t)}{z_A(t)},0\right\}
\;=\;
1 - \frac{z_C(t)}{z_A(t)}\cdot\frac{\lambda^{AG}}{\lambda^{AG}+\lambda_C(t)}.
\]
For $A$ to be indifferent, $\lambda_C(t)+\lambda^{AG}\cdot g_C^A(t)=\lambda^{AG}$,
which simplifies to
\begin{equation}\label{eq:quadratic}
\lambda_C(t)\bigl(\lambda^{AG}+\lambda_C(t)\bigr)
\;=\;
\bigl(\lambda^{AG}\bigr)^2\,\frac{z_C(t)}{z_A(t)}.
\end{equation}
This quadratic has unique positive root
\begin{equation}\label{eq:lam_C}
\lambda_C(t)
\;=\;
\frac{\lambda^{AG}}{2}\!\left(-1+\sqrt{1+4\,\frac{z_C(t)}{z_A(t)}}\right),
\end{equation}
and by symmetry $B$'s indifference condition yields the same expression.

\setcounter{proposition}{0}
\renewcommand{\theproposition}{C.\arabic{proposition}}
\begin{proposition}[Partial observability]\label{prop:partial}
Suppose $z_A(0)=z_B(0)=z_C(0)=z_0$, $\pi_{AC}=\pi_{BC}=1$, and concessions
are partially observable. There exists an equilibrium in which $A$ and $B$
concede at rate $\lambda^{AG}$ throughout and $C$ concedes with atom
$F_C(0)=1-1/h(0)\in(0,1)$ at $t=0$ and thereafter at rate
$\lambda_C(t)<\lambda^{AG}$ given by \eqref{eq:lam_C}, where $h(0)>1$ is
uniquely determined. Equilibrium payoffs satisfy
\[
v_A^* = v_B^*
= (1-\alpha)+(2\alpha-1)F_C(0)
> 1-\alpha = v_A^{AG},
\qquad
v_C^* = 2(1-\alpha) = v_C^{AG}.
\]
Thus $A$ and $B$ are strictly better off than in the bilateral benchmark while
$C$'s payoff is unchanged.
\end{proposition}
 
\bprf[Sketch of Proof]
\textit{Step 1: $A$'s and $B$'s indifference.}
Under \eqref{eq:belief_jump_sym}, if $B$--$C$ settles at time $t$, peripheral
$A$'s posterior about $C$ drops to
$\hat{z}_C(t)=z_C(t)\cdot\lambda^{AG}/(\lambda^{AG}+\lambda_C(t))$.
The AG atom that $C$ then makes to $A$ is $g_C^A(t)=1-\hat{z}_C(t)/z_A(t)$.
Substituting into $A$'s indifference condition yields exactly
\eqref{eq:quadratic}, which is satisfied by construction. By symmetry, $B$'s
indifference holds as well.
 
\smallskip
\textit{Step 2: Existence and uniqueness of the path.}
Define $h(t):=z_C(t)/z_A(t)$. Since $z_A(t)=z_0e^{\lambda^{AG}t}$ is known
explicitly, the evolution of $z_C$ reduces to the scalar ODE
\begin{equation}\label{eq:h_ode}
\dot{h}(t) \;=\; h(t)\,\bigl(\lambda_C(t)-\lambda^{AG}\bigr),
\end{equation}
with terminal condition $h(T)=1$, where $\lambda_C(t)=\frac{\lambda^{AG}}{2}
(-1+\sqrt{1+4h(t)})$ is an explicit locally Lipschitz function of $h$.
Picard--Lindel\"{o}f therefore gives a unique local solution near $T$.

Since $h(T)=1<2$, we have $\lambda_C(T)<\lambda^{AG}$ and hence $\dot{h}(T)<0$,
so $h$ is strictly increasing as we move backward from $T$. Moreover, $h$
cannot reach $2$ in finite backward time: as $h\nearrow 2$,
$\lambda_C\nearrow\lambda^{AG}$, so $\dot{h}\to 0$, and the solution slows to
a halt before reaching $2$. Therefore $h(t)\in(1,2)$ for all $t\in[0,T)$,
the solution extends uniquely to all of $[0,T]$, and in particular $h(0)>1$.
The atom $F_C(0)=1-1/h(0)\in(0,1)$ is strictly positive and uniquely pinned
down.
 
\smallskip
\textit{Step 3: $C$'s optimality.}
At any $t$ in the support of $C$'s concession distribution, $C$ is indifferent between conceding to both $A$ and $B$ simultaneously and receiving $2(1-\alpha)$, or conceding to only one and receiving $1-\alpha$ from that concession, and then receiving $1-\alpha$ from the other interaction.
 
\smallskip
\textit{Step 4: Payoffs.}
Since $A$ is indifferent throughout and concedes with positive density from
$t=0$, her payoff equals the value of conceding at $t=0$:
$v_A^*=F_C(0)\cdot\alpha+(1-F_C(0))\cdot(1-\alpha)=(1-\alpha)+(2\alpha-1)F_C(0)>1-\alpha$.
Player $C$'s payoff is $2(1-\alpha)$ by indifference.
\eprf

\brem Naturally, if $z_C(0) \approx z_0$ but $z_C(0) > z_0$, then also $C$ would concede with an atom in this equilibrium. Thus, $C'$s payoff would be strictly lower than the bilateral AG benchmark. Thus, partial observability of concessions amplifies C's disadvantage in this equilibrium.
\erem 

\section{The Four-Player Star}\label{sec:online-four-player}

This appendix extends the analysis to a four-player star network in which
player $C$ bargains simultaneously with three peripheral players $1$, $2$,
and $3$. All surplus shares are equal ($\pi_{iC} = 1$ for all $i$), and
all other features of the model are as in Section~2.

\subsection{Indifference conditions}

Fix the no-concession subgame and suppose all three peripherals and $C$ are
active at time $t$. Let $\lambda_i(t)$ denote the hazard rate of peripheral
$i \in \{1,2,3\}$ and $\lambda_C(t)$ that of $C$. When peripheral $j$
concedes at time $t$, the continuation game is the three-player star among
$\{i, k, C\}$ (where $\{i,k\} = \{1,2,3\} \setminus \{j\}$) with posteriors
$(z_i(t), z_k(t), z_C(t))$. We denote by $\mathcal{V}_i^{(j)}(t)$ peripheral
$i$'s equilibrium payoff in this three-player continuation and by
$\mathcal{W}_C^{(j)}(t)$ the center's equilibrium payoff from the remaining
two negotiations in that continuation. These objects are given in closed form
by Proposition~\ref{prop:full}.

When $C$ concedes, she is revealed rational and concedes in all negotiations
simultaneously, yielding $\alpha$ to each peripheral. Each peripheral $i$'s
indifference condition is:
\begin{equation}\label{eq:IC-i-four}
\lambda_C(t) + \sum_{j \neq i} \lambda_j(t) \, \hat{g}_i^{(j)}(t) = \lambda^{AG},
\end{equation}
where
\[
\hat{g}_i^{(j)}(t) := \frac{\mathcal{V}_i^{(j)}(t) - (1-\alpha)}{2\alpha - 1}
\]
measures the ``reputational bonus'' peripheral $i$ receives when rival $j$
concedes, relative to $i$'s payoff from conceding herself.

The center's indifference condition is:
\begin{equation}\label{eq:IC-C-four}
\sum_{j=1}^{3} \lambda_j(t) \, \Gamma_j(t) = 3\lambda^{AG},
\end{equation}
where
\[
\Gamma_j(t) := \frac{\alpha + \mathcal{W}_C^{(j)}(t) - 3(1-\alpha)}{2\alpha - 1}
\]
measures the center's gain when peripheral $j$ concedes (receiving $\alpha$
from $j$ plus the continuation value from the remaining two negotiations)
relative to the center's payoff from conceding on all three negotiations.

In the three-player model, the analogous continuation values are bilateral
AG payoffs, which are simple closed-form functions of posteriors.
This yields constant hazard rates in each phase. In the four-player model,
$\hat{g}_i^{(j)}(t)$ and $\Gamma_j(t)$ encode three-player equilibrium
payoffs that depend on the evolving posterior vector, so
\eqref{eq:IC-i-four}--\eqref{eq:IC-C-four} generally produce time-varying
hazard rates.

\subsection{Equilibrium structure}

Despite the added complexity, a numerical example illustrates that the equilibrium retains the sequential activation structure. For the parameterization $z_1(0) = 0.5 > z_2(0) = 0.4 > z_C(0) = 0.23
> z_3(0) = 0.2$ with $r = 1$ and $\alpha = 0.7$, we solve the system
\eqref{eq:IC-i-four}--\eqref{eq:IC-C-four} numerically via backward
integration from the terminal time $T$ (at which all posteriors reach $1$) and
obtain a three-phase equilibrium:
\begin{itemize}
    \item \textbf{Phase I} ($t \in [0, t_2]$): Only peripheral $3$ and $C$
    are active, with constant hazard rates $\lambda_3 = 3\lambda^{AG}$ and
    $\lambda_C = \lambda^{AG}$. Peripherals $1$ and $2$ are inactive. $C$
    concedes with an atom $F_C(0) \approx 0.33$ at $t = 0$.

    \item \textbf{Phase II} ($t \in [t_2, t_1]$): Peripheral $2$ activates.
    All of $\{2, 3, C\}$ concede with approximately constant hazard rates
    $\lambda_2 \approx \lambda_3 \approx 1.5\lambda^{AG}$ and $\lambda_C
    \approx \lambda^{AG}$, while peripheral $1$ remains inactive.

    \item \textbf{Phase III} ($t \in [t_1, T]$): All four players concede at
    $\lambda^{AG}$ until the common terminal time $T$.
\end{itemize}

\noindent Figure~\ref{Figure: four player comparison} illustrates the
posterior dynamics alongside the bilateral AG benchmark.

\begin{figure}[ht]
\centering
\begin{minipage}{0.45\textwidth}
\centering
\begin{tikzpicture}
  \begin{axis}[
    xlabel={$t$},
    xlabel style={at={(ticklabel* cs:1)}, anchor=west, yshift=5pt},
    ylabel={Pr[behavioral]},
    ymin=0, ymax=1.12,
    xmin=0, xmax=1.7,
    samples=200,
    xtick={0.308, 0.506, 1.431},
    xticklabels={$t_2$, $t_1$, $T$},
    ytick={0.23, 0.4, 0.5, 1},
    yticklabels={$z_C(0)$, $z_2(0)$, $z_1(0)$, $1$},
    extra y ticks={0.2},
    extra y tick labels={$z_3(0)$},
    extra y tick style={yticklabel style={yshift=-4pt}},
    major tick length=4pt,
    ytick style={thick},
    axis lines=left,
    width=\textwidth,
    height=0.9\textwidth,
  ]

  \addplot[blue, very thick, domain=0:0.308] {0.2*exp(2.25*x)};

  \addplot[red, very thick, domain=0:0.308] {0.3433*exp(0.75*x)};

  \addplot[green!50!black, very thick, domain=0:0.308] {0.4};

  \addplot[black, very thick, domain=0:0.506] {0.5};

  \addplot[green!50!black, very thick, domain=0.308:0.506]
    {0.4*exp(1.125*(x - 0.308))};

  \addplot[blue, very thick, domain=0.308:0.506]
    {0.4*exp(1.125*(x - 0.308))};

  \addplot[red, very thick, domain=0.308:0.506]
    {0.4325*exp(0.72*(x - 0.308))};

  \addplot[black, very thick, domain=0.506:1.431]
    {0.5*exp(0.75*(x - 0.506))};

  \addplot[black, very thick, domain=1.431:1.7] {1};

  \draw[dotted] (axis cs:0.308,0) -- (axis cs:0.308,0.5);
  \draw[dotted] (axis cs:0.506,0) -- (axis cs:0.506,0.5);
  \draw[dotted] (axis cs:1.431,0) -- (axis cs:1.431,1);

  \addplot[black, dotted, domain=0:0.506] {0.5};
  \addplot[black, dotted, domain=0:1.431] {1};

  \draw[->, red, thick]
    (axis cs:0,0.23)
    .. controls (axis cs:0.10,0.26) and (axis cs:0.10,0.32) ..
    (axis cs:0,0.3433);

  \end{axis}
\end{tikzpicture}
\vspace{-0.4cm}
\subcaption{Multilateral bargaining}
\end{minipage}\hfill
\begin{minipage}{0.45\textwidth}
\centering
\begin{tikzpicture}
  \begin{axis}[
    xlabel={$t$},
    xlabel style={at={(ticklabel* cs:1)}, anchor=west, yshift=5pt},
    ylabel={Pr[behavioral]},
    ymin=0, ymax=1.12,
    xmin=0, xmax=2.3,
    samples=200,
    xtick={0.924, 1.222, 1.960},
    xticklabels={$T^{1C}_{AG}$, $T^{2C}_{AG}$, $T^{3C}_{AG}$},
    ytick={0.23, 0.4, 0.5, 1},
    yticklabels={$z_C(0)$, $z_2(0)$, $z_1(0)$, $1$},
    extra y ticks={0.2},
    extra y tick labels={$z_3(0)$},
    extra y tick style={yticklabel style={yshift=-4pt}},
    major tick length=4pt,
    ytick style={thick},
    axis lines=left,
    width=\textwidth,
    height=0.9\textwidth,
  ]

  \addplot[black, very thick, domain=0:0.924] {0.5*exp(0.75*x)};
  \addplot[black, very thick, domain=0.924:2.3] {1};

  \addplot[green!50!black, very thick, domain=0:1.222] {0.4*exp(0.75*x)};
  \addplot[green!50!black, very thick, domain=1.222:2.3] {1};

  \addplot[blue, very thick, domain=0:1.960] {0.23*exp(0.75*x)};
  \addplot[blue, very thick, domain=1.960:2.3] {1};

  \draw[dotted] (axis cs:0.924,0) -- (axis cs:0.924,1);
  \draw[dotted] (axis cs:1.222,0) -- (axis cs:1.222,1);
  \draw[dotted] (axis cs:1.960,0) -- (axis cs:1.960,1);

  \addplot[black, dotted, domain=0:1.960] {1};

  \draw[->, black, thick]
    (axis cs:0,0.23)
    .. controls (axis cs:0.12,0.30) and (axis cs:0.12,0.44) ..
    (axis cs:0,0.5);

  \draw[->, green!50!black, thick]
    (axis cs:0,0.23)
    .. controls (axis cs:0.08,0.27) and (axis cs:0.08,0.36) ..
    (axis cs:0,0.4);

  \draw[->, blue, thick]
    (axis cs:0,0.2)
    .. controls (axis cs:0.06,0.21) and (axis cs:0.06,0.225) ..
    (axis cs:0,0.23);

  \end{axis}
\end{tikzpicture}
\vspace{-0.4cm}
\subcaption{Bilateral bargaining (AG benchmark)}
\end{minipage}
\caption{Four-player star: comparison of concession behavior. Multilateral
bargaining (left) vs.\ bilateral bargaining (right). Parameters: $r = 1$,
$\alpha = 0.7$, $z_1(0) = 0.5$, $z_2(0) = 0.4$, $z_C(0) = 0.23$,
$z_3(0) = 0.2$.}\label{Figure: four player comparison}
\end{figure}
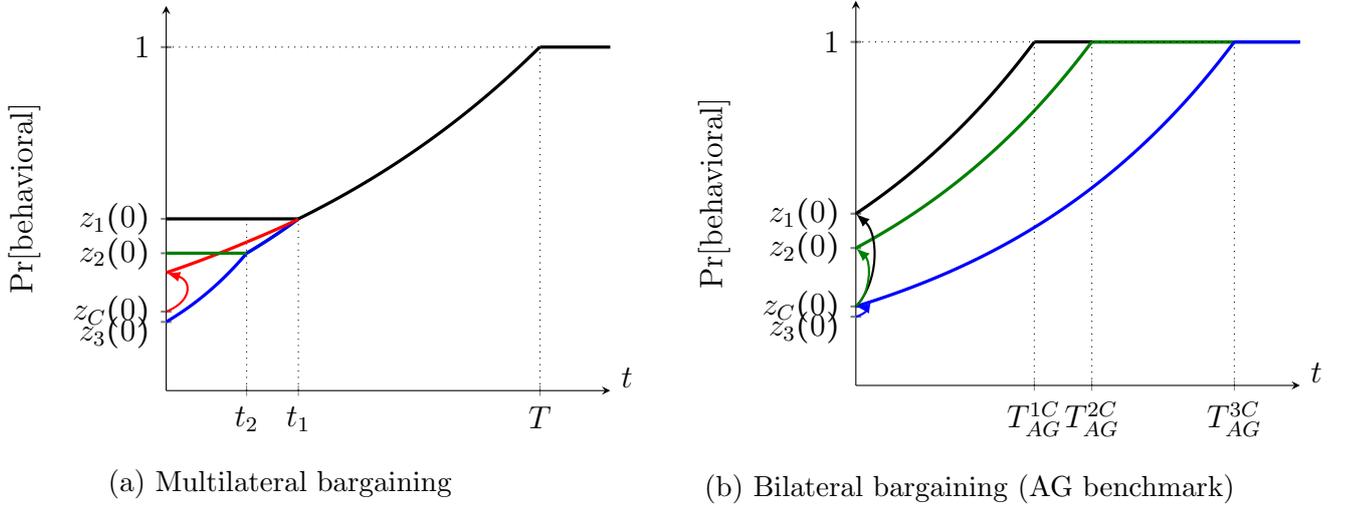

\subsection{Payoff comparison}

Table~\ref{tab:four-player-payoffs} compares each player's equilibrium payoff
under multilateral bargaining with the bilateral AG benchmark.

\begin{table}[ht]
\centering
\begin{tabular}{lcccc}
\hline
& Peripheral 1 & Peripheral 2 & Peripheral 3 & Center $C$ \\
\hline
Multilateral ($v_i^*$) & 0.397 & 0.449 & 0.432 & 0.900 \\
AG benchmark ($v_i^{AG}$) & 0.516 & 0.470 & 0.300 & 0.952 \\
Difference & $-0.119$ & $-0.021$ & $+0.132$ & $-0.052$ \\
\hline
\end{tabular}
\caption{Equilibrium payoffs: four-player multilateral vs.\ bilateral AG
benchmark.}\label{tab:four-player-payoffs}
\end{table}

The payoff pattern extends the three-player results but reveals an additional
implication. The weakest peripheral (player~3) is strictly better off under
multilateral bargaining, gaining $0.132$. The center is strictly worse off,
losing $0.052$ relative to the sum of her bilateral payoffs. But now
peripheral~2---who is the \emph{intermediate} peripheral---is also worse off,
losing $0.021$, and the strongest peripheral (player~1) suffers the largest
loss of $0.119$. This suggests a sharper conclusion than in the three-player
case: reputational spillovers benefit only the weakest peripheral, while all
other players---including moderately weak peripherals---can be made worse off.
The magnitude of the strongest peripheral's loss is substantial, reflecting
the extended initial phase during which player~1 remains inactive while $C$
and the weaker peripherals ``bargain in earnest.''
\end{document}